\documentclass[oneside,english,final]{siamltex}
\usepackage[T1]{fontenc}
\usepackage[latin9]{inputenc}
\usepackage{amsbsy}
\usepackage{amstext}
\usepackage{amssymb}
\usepackage{graphicx}
\usepackage{nomencl}
\providecommand{\printnomenclature}{\printglossary}
\providecommand{\makenomenclature}{\makeglossary}
\makenomenclature

\makeatletter
\newcommand\eqref[1]{(\ref{#1})}

\@ifundefined{date}{}{\date{}}
\usepackage{calc}
\usepackage{graphicx}
\usepackage{tikz}
\usepackage{ngerman}
\usepackage{placeins}
\usepackage{amssymb}

\usetikzlibrary{automata,positioning,decorations.pathreplacing,topaths,arrows,shapes,snakes}

\newcommand\tikzscaled[2][\textwidth]{\resizebox{\minof{\width}{#1}}{!}{\tikzpicture[node distance=2cm,on grid,auto] #2\endtikzpicture}}

\newcommand{\cov}{\mathrm{Cov}}
\newcommand{\var}{\mathrm{Var}}
\newcommand{\E}{\mathbb{E}}

\renewcommand\nompreamble{For convenience, we summarize some notations used throughout the paper as follows:}

\newcommand{\plim}{\mathop{\mathrm{plim}}}

\newtheorem{remark}[theorem]{\it Remark \scshape}
\newtheorem{assumption}[theorem]{Assumption}

\renewcommand\paragraph{\@startsection{paragraph}{4}{.25in}{\parskip}{-.5em plus -.1em}{\reset@font\normalsize\bfseries}}

\@ifundefined{showcaptionsetup}{}{%
 \PassOptionsToPackage{caption=false}{subfig}}
\usepackage{subfig}
\makeatother

\usepackage{babel}

\begin{document}

\title{Optimal estimation of free energies and stationary densities from
multiple biased simulations}

\author{Hao Wu\footnotemark[2] and Frank Noé\footnotemark[2]}

\maketitle
\renewcommand{\thefootnote}{\fnsymbol{footnote}}
\footnotetext[2]{Mathematics Institute, Department of Mathematics and Computer Science,         Free University of Berlin,         Arnimallee 6, 14195 Berlin, Germany         ({\tt hwu@zedat.fu-berlin.de}, {\tt frank.noe@fu-berlin.de}).} 
\begin{abstract}
When studying high-dimensional dynamical systems such as macromolecules,
quantum systems and polymers, a prime concern is the identification
of the most probable states and their stationary probabilities or
free energies. Often, these systems have metastable regions or phases,
prohibiting to estimate the stationary probabilities by direct simulation.
Efficient sampling methods such as umbrella sampling, metadynamics
and conformational flooding have developed that perform a number of
simulations where the system\textquoteright{}s potential is biased
such as to accelerate the rare barrier crossing events. A joint free
energy profile or stationary density can then be obtained from these
biased simulations with weighted histogram analysis method (WHAM).
This approach (a) requires a few essential order parameters to be
defined in which the histogram is set up, and (b) assumes that each
simulation is in global equilibrium. Both assumptions make the investigation
of high-dimensional systems with previously unknown energy landscape
difficult. Here, we introduce the transition matrix unweighting (TMU)
method, a simple and efficient estimation method which dismisses both
assumptions. The configuration space is discretized into sets, but
these sets are not only restricted to the selected slow coordinate
but can be clusters that form a partition of high-dimensional state
space. The assumption of global equilibrium is replaced by requiring
only local equilibrium within the discrete sets, and the stationary
density or free energy is extracted from the transitions between clusters.
We prove the asymptotic convergence and normality of TMU, give an
efficient approximate version of it and demonstrate its usefulness
on numerical examples.\end{abstract}
\pagestyle{myheadings}
\thispagestyle{plain}
\markboth{OPTIMAL ESTIMATION OF MULTIPLE BIASED SIMULATIONS}{HAO WU AND FRANK NOÉ}

\section{Introduction}

Stochastic simulations of chemical, physical or biological processes
often involve rare events that render the exploration of relevant
states, or the calculation of expectation values by direct numerical
simulation difficult or impossible. Examples include phase transitions
in spin systems \cite{SherringtonKirkpatrick_PRL75_SpinGlasses,BinderYoung_RMP86_SpinGlasses},
transitions between different chemical states in quantum dynamics
simulations \cite{MarxHutter_NICseries00_AbInitioMD} or conformational
transitions in biomolecules \cite{Frauenfelder_Science254_1598}.
For this reason, many enhanced sampling techniques have been developed
to modify the dynamics of the original simulation system such that
the relevant rare events become more frequent and can be accessed
by direct numerical simulation of the modified simulation system.
Such an approach is of course only reasonable if there exists a way
to reliably compute at least some quantities of interest of the original
simulation system from the realizations of the modified simulation
system.

In this paper we focus on processes that are asymptotically stationary
and ergodic, and on enhanced sampling approaches that use bias potentials
(or equivalently conservative bias force fields) that attempt to modify
the original system's dynamics so as to avoid rare events. Well-known
examples of such approaches are Umbrella Sampling \cite{Torrie_JCompPhys23_187},
conformational flooding \cite{Grubmueller_PhysRevE52_2893}, Metadynamics
and variants \cite{LaioParrinello_PNAS99_12562,BarducciBussiParrinello_PRL08_WellTemperedMetadynamics}.
These approaches assume that the modeler has some knowledge of coordinates
or order parameters that are ``slow'', i.e. that the rare event
dynamics of the system is resolved by state transitions in these selected
coordinates. 

Umbrella sampling defines a series of $K$ biased simulations, each
of which uses the forces from the original dynamics plus the forces
due to one of $K$ harmonic potentials. These potentials restrain
the biased simulations to stay close to positions $\mathbf{x}_{k}$
in the selected coordinates which are the centers of the umbrella
potentials. The force constant(s) of these potentials must be chosen
such that the corresponding biased stationary densities overlap significantly
and the unification of all biased stationary densities covers the
part of state space in which the original stationary density is significantly
greater than zero. In this case, the $K$ biased simulations, together
with the knowledge of the umbrella potentials can be used in order
to estimate the original stationary density in the selected coordinates
(or the corresponding free energy landscape). 

Metadynamics is based on an opposite philosophy. Rather than by constraining
the simulation to a set of points, it adds bias potentials to drive
the simulation away from regions it has sampled sufficiently well.
In practice this is often done by adding Gaussian hat functions to
the bias potential every $n$ simulation steps, centered at the current
simulation state. We consider that this happens a number $K$ times,
leading to $K$ simulation snippets, each with a different biasing
potential. Due to limitations of filling high-dimensional space volumes,
these bias potentials live usually also in a few pre-defined coordinates.
Since the sequence of added bias potentials depends on the simulation
history, metadynamics is usually used to first ``fill up'' the free
energy wells until the states that cause the rare event waiting times
have been destabilized and the corresponding free energy landscape
is approximately ``flat''. It can be shown that at this point continuing
the metadynamics simulation will sample bias potentials that are the
negative free energy landscapes of the original system, up to an arbitrary
additive constant. Since metadynamics does not require the modeler
to know the relevant states along the slow coordinates, it is not
only an approach to quantify the stationary distribution / free energy
landscape of the original system but has been very successful in terms
of exploring the state space in complex systems \cite{PianaLaio_JPCB07_MetadynamicsProteinFolding}.
Unfortunately, this approach of using metadynamics also appears to
suggest that all simulation effort that has been spent until the free
energy surface is approximately flat cannot be used for quantitative
estimations.

Here we concentrate on the step of unbiased the modified dynamics
so as to obtain the stationary distribution of the original dynamical
system. For both, umbrella sampling and metadynamics, the step of
``unbiasing'' is usually done with the weighted histogram method
(WHAM) \cite{FerrenbergSwendsen_PRL89_WHAM}. WHAM uses a discretization
of the selected coordinates in which the biased simulation was done,
collects a set of $K$ histograms, one for each of the biased simulations.
These biased histograms are then combined to a single unbiased histogram
by solving the self-consistent set of equations: $\pi_{i}=\left(\sum_{k=1}^{K}n_{i}^{k}\right)/\left(\sum_{k=1}^{K}M_{k}c_{i}^{k}/z^{k}\right)$
and $z^{k}=\sum_{j}c_{j}^{k}\pi_{j}$, where $n_{i}^{k}$ is the number
of counts in histogram bin $i$ for simulation $k$, $M_{k}$ is the
total number of samples generated by the $k$-th simulation, $z^{k}$
is a normalization constant, $c_{i}^{k}$ is the unbiasing factor
of state $i$ at simulation $k$, and $\pi_{i}$ is the unbiased probability
of state $i$ at simulation condition $k$.

The assumption used by WHAM is that each of the $K$ simulations done
at different conditions is sufficiently long such that they generate
unbiased samples of the corresponding biased stationary density $\boldsymbol{\pi}_{k}$.
In order words, each sub-simulation is assumed to be in global equilibrium
at its conditions. We will see that restricting to this assumption
is unnecessary, and a method that does not rely on this assumption
can provide estimates that are substantially more precise (or, equivalently,
require substantially less simulation effort for a given level of
precision), even in trivial double-well examples. 

This paper develops the transition-matrix based unbiasing method (TMU)
which replaces the assumption that the $K$ biased simulations are
in global equilibria by the much weaker assumption that each simulation
is only in local equilibrium in the discrete states on which the stationary
distribution is estimated. TMU has been motivated by the recent progresses
in Markov modeling \cite{SwopePiteraSuits_JPCB108_6571,ChoderaEtAl_JCP07,NoeHorenkeSchutteSmith_JCP07_Metastability,prinz2011markov},
and constructs the joint unbiased stationary distribution from a series
of $K$ transition count matrices estimated from the $K$ biased simulations.
However, it is important to note that TMU does \emph{not} need the
discrete dynamics to be Markovian. 

Subsequently, we describe the basic mathematical assumptions underlying
our method, then describe TMU in its most general form and show that
the method has always a solution that is asymptotically normal and
convergent. We then provide an approximate TMU that is efficient for
very large state spaces and a large number $K$ of sub-simulations.
The method is demonstrated in conjunction with umbrella sampling and
metadynamics on double-well potentials and its performance is compared
with that of the standard WHAM method and a recently introduced method
MMMM that had a similar motivation \cite{sakuraba2009multiple}.

\settowidth{\nomlabelwidth}{$\mathrm{span}\left(v_1,\ldots,v_r\right)$}
\printnomenclature{}

\nomenclature{$\mathrm{null}\left(G\right)$}{null space of $G$: $\left\{ v\vert Gv=\mathbf{0}\right\}$}\nomenclature{$\mathrm{span}\left(v_1,\ldots,v_r\right)$}{vector space spanned by $v_1,\ldots,v_r$}\nomenclature{$\mathbb R$}{set of real numbers}\nomenclature{$\mathbb R^{m\times n}$}{set of real $m\times n$-matrices}\nomenclature{$\left\vert S\right\vert$}{cardinality of set $S$}\nomenclature{$x_{k:l}$}{subsequence $\{x_{k},x_{k+1},\ldots,x_{l}\}$ of a given discrete-time sequence $\{x_{0},x_{1},\ldots\}$}\nomenclature{$G\succ 0 (\succeq 0)$}{each element of matrix $G$ is positive (non-negative)}\nomenclature{$G<0 (\le 0)$}{matrix $G$ is negative-definite (negative-semidefinite)}\nomenclature{$G>0 (\ge 0)$}{matrix $G$ is positive-definite (positive-semidefinite)}\nomenclature{$\mathcal V\left(G\right)$}{vector $\left(G_{11},G_{12},\ldots,G_{mn}\right)^\mathrm{T}$ which consists of elements of $G=\left[G_{ij}\right]\in\mathbb{R}^{m\times n}$}\nomenclature{$\mathcal N\left(\mu,\Sigma\right)$}{multivariate normal distribution with mean $\mu$ and covariance matrix $\Sigma$}\nomenclature{$x_{n}\stackrel{p}{\to}x$}{$x_n$ converges in probability to $x$ with respect to $n$, which is also denoted as $\plim_{n\to\infty}x_{n}=x$}\nomenclature{$x_{n}\stackrel{d}{\to}x$}{$x_n$ converges in distribution to $x$ with respect to $n$}\nomenclature{$\E\left[x\right]$}{expected value of $x$}\nomenclature{$\cov\left(x,y\right)$}{$\left(x-\mathbb{E}\left[x\right]\right)\left(y-\mathbb{E}\left[y\right]\right)^{\mathrm{T}}$ for column vectors $x,y$}\nomenclature{$\var\left(x\right)$}{$\cov\left(x,x\right)$}\nomenclature{${\left\Vert G\right\Vert}_{\max}$}{$\max_{i,j}\vert G_{ij}\vert$ for $G=[G_{ij}]$}\nomenclature{${\left\Vert G\right\Vert}$}{$2$-norm of $G$}\nomenclature{$\mathrm{tr}\left(G\right)$}{trace of square matrix $G$}\nomenclature{$1_\omega$}{indicator function of event $\omega$, taking value $1$ if $\omega$ holds and $0$ otherwise}\nomenclature{$\mathbf 0,\mathbf 1$}{vectors of zeros and ones in appropriate dimensions}\nomenclature{$\mathbf I$}{identity matrix}\nomenclature{$\nabla_x y$}{Jacobian matrix $\left[\partial y_i/\partial x_j\right]$ of $y=[y_i]$ with respect to $x=[x_i]$}\nomenclature{$\nabla_{yx} z$}{$\nabla_{y}\left(\nabla_{x}z\right)^{\mathrm{T}}$ for $z\in\mathbb R$}\nomenclature{$G^+$}{Moore-Penrose pseudoinverse of $G$}\nomenclature{$\mathrm{dim}\left(x\right)$}{dimension of $x$}\nomenclature{$\mathrm{diag}\left(v\right)$}{diagonal matrix with the $i$-th diagonal element $v_i$ for $v=\left[v_i\right]$ being a vector}\nomenclature{$\bar x$, $\hat x$, $\tilde x$}{``true'' value, estimate obtained from simply ``counting'' and maximum likelihood estimate of an unknown variable $x$ (except if stated otherwise)}\nomenclature{$\mathrm{d_{TV}}\left(\mu,\nu\right)$}{total variation distance between probability distributions $\mu$ and $\nu$}\nomenclature{$I_S(s)$}{index of element $s$ in the set $S$}

\section{Background\label{sec:Background}}

In this section, we briefly review the mathematical background of
biased simulation techniques. Let us consider a reference system in
which the configuration space can be decomposed into a discrete state
set $\mathcal{S}$ with free energy $V=\left[V_{i}\right]$ by using
some sort of discretization (where $V_{i}$ is the energy of the $i$-th
state in $\mathcal{S}$). If we denote the system state at time $t$
by $x_{t}$, the state sequence $\left\{ x_{t}\right\} $ is then
a stochastic process. In this paper, we focus on processes $\left\{ x_{t}\right\} $
with the following properties, which are relevant for many physical
simulation processes:

\paragraph*{Asymptotic stationarity}

It means that the sequence $\left\{ x_{t}|t\ge\tau\right\} $ is approximately
stationary if $\tau$ is large enough. More formally, $\left\{ x_{t}\right\} $
is said to be asymptotically stationary if there exists a family of
distribution functions $F_{n}:\mathcal{S}^{n}\mapsto\mathbb{R}$ such
that $\lim_{\tau\to\infty}\Pr(x_{\tau+1}=s_{1},\ldots,x_{\tau+n}=s_{n})=F_{n}(s_{1},\ldots,s_{n})$
for all $n\ge1$ and $s_{1},\ldots,s_{n}\in\mathcal{S}$. Specifically,
$\lim_{\tau\to\infty}\Pr\left(I_{\mathcal{S}}\left(x_{\tau}\right)=i\right)=\pi_{i}$
with
\begin{equation}
\pi_{i}=\frac{\exp\left(-\beta V_{i}\right)}{\sum_{j}\exp\left(-\beta V_{j}\right)}\label{eq:free-energy-definition}
\end{equation}
where $\pi=\left[\pi_{i}\right]$ denotes the system stationary distribution
and $\beta$ is a constant and generally proportional to the inverse
temperature in physical systems, and we denote the limit $\lim_{\tau\to\infty}\Pr\left(I_{\mathcal{S}}\left(x_{\tau+1}\right)=j|I_{\mathcal{S}}\left(x_{\tau}\right)=i\right)$
by $T_{ij}$.

\paragraph*{Wide-sense ergodicity}

That is,
\begin{equation}
\plim_{t\to\infty}\frac{1}{t+1}\sum_{\tau=0}^{t}1_{I_{\mathcal{S}}\left(x_{\tau}\right)=i}=\pi_{i}\label{eq:ergodicity-pi}
\end{equation}
and
\begin{equation}
\plim_{t\to\infty}\frac{1}{t}\sum_{\tau=0}^{t-1}1_{I_{\mathcal{S}}\left(x_{\tau}\right)=i}\cdot1_{I_{\mathcal{S}}\left(x_{\tau+1}\right)=j}=\pi_{i}T_{ij}
\end{equation}

\paragraph*{Detailed balance}

The detailed balance condition can be written as
\begin{equation}
\lim_{t\to\infty}\Pr(I_{\mathcal{S}}(x_{t})=i,I_{\mathcal{S}}(x_{t+1})=j)=\lim_{t\to\infty}\Pr(I_{\mathcal{S}}(x_{t})=j,I_{\mathcal{S}}(x_{t+1})=i)\label{eq:detailed-balance}
\end{equation}
or equivalently $\pi_{i}T_{ij}=\pi_{j}T_{ji}$, for all $i,j$, which
implies that each state transition has the same probability as its
reverse. Note that this property clearly holds for systems that are
time-reversible at equilibrium.

\begin{remark}{\rm For convenience, we measure energies $V_{i}$
in units of thermal energy $k_{\mathrm{B}}T_{\mathrm{thermal}}$ with
$k_{\mathrm{B}}$ the Boltzmann constant and $T_{\mathrm{thermal}}$
the thermodynamic temperature, yielding $\beta=1$ in \eqref{eq:free-energy-definition}.
Furthermore, we assume without losing generality that all involved
free energies in this paper are zero-mean. Then \eqref{eq:free-energy-definition}
has a unique solution
\begin{equation}
V_{j}=-\log\pi_{j}+\frac{1}{\left|\mathcal{S}\right|}\sum_{i}\log\pi_{i}\label{eq:V-pi}
\end{equation}
for a given stationary distribution.}\end{remark}

Based on the asymptotic stationarity, we can define a matrix $T=\left[T_{ij}\right]\in\mathbb{R}^{\left|\mathcal{S}\right|\times\left|\mathcal{S}\right|}$
to represent the stationary state transition probabilities. It is
easy to see that $T$ satisfies the condition that each row sums to
$1$, so here we call $T$ the transition matrix of $\left\{ x_{t}\right\} $
for simplicity even if $\left\{ x_{t}\right\} $ is not a Markov chain.

For now, our goal is to estimate $V$, or equivalently $\pi$, from
simulations in the case that it is unknown. Due to the ergodicity
\eqref{eq:ergodicity-pi}, when sufficient simulation data can be
generated, we can simply carry out one or multiple simulations of
the reference system and get the estimate of $V$ through computing
the histogram of simulation data. This approach, however, is very
inefficient when the reference system has multiple metastable states,
because the simulation process is very likely to get stuck in some
local minima of the energy landscape for a long time. To this end,
biased simulation techniques, such as umbrella sampling \cite{souaille2001extension}
and metadynamics \cite{LaioParrinello_PNAS99_12562}, were developed
to solve this problem, which perform simulations with a set of biased
potentials so that the energy landscape can be explored more efficiently.

Although many practical algorithms use a different approach, we can
roughly summarize the estimation of stationary distributions through
biased simulations in terms of the following pseudocode:
\begin{description}
\item [{Step 1}] Design a set of biasing potentials $\left\{ U^{k}|k=1,\ldots,K\right\} $
where $U^{k}=\left[U_{i}^{k}\right]\in\mathbb{R}^{\left|\mathcal{S}\right|\times1}$.
\item [{Step 2}] Repeat Steps 2.1 and 2.2 for $k=1,\ldots,K$:

\begin{description}
\item [{Step 2.1}] Reduce the state set to $\mathcal{S}^{k}\subseteq\mathcal{S}$
and change the system potential as
\begin{equation}
V_{i}^{k}=V_{\mathrm{ID}^{k}\left(i\right)}+U_{\mathrm{ID}^{k}\left(i\right)}^{k}-\frac{1}{\left|\mathcal{S}^{k}\right|}\sum_{j\in\mathcal{J}^{k}}\left(V_{j}+U_{j}^{k}\right)\label{eq:biased-potential-original}
\end{equation}
where $V^{k}=\left[V_{i}^{k}\right]\in\mathbb{R}^{\left|\mathcal{S}^{k}\right|\times1}$
is called the biased potential, $\mathrm{ID}^{k}\left(i\right)=I_{\mathcal{S}}\left(\text{\ensuremath{i}-th element of \ensuremath{\mathcal{S}^{k}}}\right)$,
$\mathcal{J}^{k}=\left\{ I_{\mathcal{S}}\left(s\right)|s\in\mathcal{S}^{k}\right\} $,
and the last term is used to shift the mean of $V^{k}$ to zero.
\item [{Step 2.2}] Perform a biased simulation using the same simulation
model as the reference system except that the state set and potential
energy are changed from $\mathcal{S},V$ to $\mathcal{S}^{k},V^{k}$,
and record the simulation trajectory $x_{0:M_{k}}^{k}$. 
\end{description}
\item [{Step 3}] Estimate the reference (unbiased) free energy $V$ from
$\left\{ x_{0:M_{k}}^{k}|k=1,\ldots,K\right\} $.
\end{description}
In this paper, we will focus on the estimation problem in Step 3.
We start with the assumption that each simulation $x_{0:M_{k}}^{k}$
is a Markov chain, and the developed estimation method will then be
proved to be applicable to more general simulation models.

\begin{remark}{\rm \eqref{eq:biased-potential-original} can be written
in a more compact form by defining a potential transformation matrix
$A^{k}=\left[A_{ij}^{k}\right]\in\mathbb{R}^{\left|\mathcal{S}^{k}\right|\times\left|\mathcal{S}\right|}$
with $A_{ij}^{k}=1_{j\in\mathcal{J}^{k}}\cdot\left(1_{\mathrm{ID}^{k}\left(i\right)=j}-1/\left|\mathcal{S}^{k}\right|\right)$
such that $V^{k}=A^{k}\left(V+U^{k}\right)$.}\end{remark}

\section{Maximum likelihood estimation from multiple simulations}

\subsection{Maximum likelihood estimation\label{sub:Maximum-likelihood-estimation}}

In this section, we investigate a maximum likelihood approach to the
estimation problem described in Section \ref{sec:Background}. Suppose
that each simulation $x_{0:M_{k}}^{k}$ is a time-homogeneous Markov
chain. According to the third simulation property stated in Section
\ref{sec:Background}, it is easy to see that $x_{0:M_{k}}^{k}$ is
a reversible Markov chain. (It can further be proved that $x_{0:M_{k}}^{k}$
is irreducible and positive recurrent under the condition of finite
$V^{k}$ by using the first and second properties.) The maximum likelihood
estimation (MLE) of the unbiased stationary distribution $\pi$ can
then be obtained by solving the following optimization problem:

\begin{equation}
\begin{array}{lll}
\max\limits _{\pi,T^{1},\ldots,T^{K}} & L=\sum_{k}\log\Pr\left(x_{1:M_{k}}^{k}|x_{0}^{k},T^{k}\right)=\sum_{i,j,k}C_{ij}^{k}\log T_{ij}^{k}\\
s.t. & \pi\textrm{ is a probability vector}\\
 & T^{k}\textrm{ is a transition matrix}\\
 & \mathrm{diag}\left(\pi^{k}\right)T^{k}=T^{k\mathrm{T}}\mathrm{diag}\left(\pi^{k}\right), & k=1,\ldots K
\end{array}\label{eq:original_ml}
\end{equation}
where $T^{k}=\left[T_{ij}^{k}\right]$ and $\pi^{k}=\left[\pi_{i}^{k}\right]$
denote the transition matrix and stationary distribution of the $k$-th
simulation, $C^{k}=\left[C_{ij}^{k}\right]=\sum_{t=1}^{M_{k}}\Delta C_{t}^{k}$
is the count matrix of the $k$-th simulation and $\Delta C_{t}^{k}=\left[\Delta C_{t,ij}^{k}\right]=\left[1_{\left(I_{\mathcal{S}^{k}}\left(x_{t-1}^{k}\right),I_{\mathcal{S}^{k}}\left(x_{t}^{k}\right)\right)=\left(i,j\right)}\right]$.
(Here we set $0\log0=0$ and $a\log0=-\infty$ if $a>0$.) The last
constraint is obviously the detailed balance constraint for the biased
simulations, and can be equivalently written as $\mathrm{diag}\left(\Lambda^{k}\pi\right)T^{k}=T^{k\mathrm{T}}\mathrm{diag}\left(\Lambda^{k}\pi\right)$
with $\Lambda^{k}=\left[\Lambda_{ij}^{k}\right]=\left[1_{\mathrm{ID}^{k}\left(i\right)=j}\cdot\exp\left(-U_{j}^{k}\right)\right]\in\mathbb{R}^{\left|\mathcal{S}^{k}\right|\times\left|\mathcal{S}\right|}$.
After performing the MLE of $\pi$, the optimal estimate of $V$ can
also be obtained by using \eqref{eq:V-pi}.
\begin{theorem}
\label{thm:optimization-problem}The optimization problem \eqref{eq:original_ml}
has at least one optimal solution satisfying
\begin{enumerate}
\item $T_{ij}^{k}=0$ for $\left(i,j,k\right)\in\left\{ \left(i,j,k\right)|C_{ij}^{k}+C_{ji}^{k}=0\text{ and }i\neq j\right\} $.
\item $1_{T_{ij}^{k}>0}\equiv1_{C_{ij}^{k}>0}$, if $C_{ii}^{k}>0$ and
$\ensuremath{1_{C_{ij}^{k}>0}=1_{C_{ji}^{k}>0}}$ for all $i,j,k$.
\end{enumerate}
\end{theorem}
\begin{proof}
See Appendix \ref{sec:Proof-of-Theorem-optimization-problem}.
\end{proof}
~

According to Theorem \ref{thm:optimization-problem}, we can reduce
the dimension of the optimization variable of \eqref{eq:original_ml},
and the reduced problem can be written as
\begin{equation}
\begin{array}{lll}
\max\limits _{\pi,\left\{ T_{ij}^{k}|C_{ij}^{k}+C_{ji}^{k}>0\text{ or }i=j\right\} } & L=\sum_{i,j,k}C_{ij}^{k}\log T_{ij}^{k}\\
s.t. & \pi\textrm{ \textrm{is a probability vector}}\\
 & T^{k}\textrm{ is a transition matrix}\\
 & \mathrm{diag}\left(\Lambda^{k}\pi\right)T^{k}=T^{k\mathrm{T}}\mathrm{diag}\left(\Lambda^{k}\pi\right), & k=1,\ldots K
\end{array}\label{eq:reduced_ml}
\end{equation}
with $T_{ij}^{k}=0$ for $\left(i,j,k\right)\in\left\{ \left(i,j,k\right)|C_{ij}^{k}+C_{ji}^{k}=0,i\neq j\right\} $.
It is clear that the problem size of \eqref{eq:reduced_ml} is much
smaller than that of \eqref{eq:original_ml} especially when count
matrices are sparse. However, even if each $C^{k}$ is sparse with
$O\left(\left|\mathcal{S}\right|\right)$ nonzero elements, the reduced
problem \eqref{eq:reduced_ml} involves $O\left(K\left|\mathcal{S}\right|\right)$
decision variables and nonlinear equality constraints. (Note that
$\pi$ and $T^{k}$ are both unknown in the last constraint.) It is
still inefficient to search the optimal solution of \eqref{eq:reduced_ml}
by direct methods. In Section \ref{sec:Approximate-maximum-likelihood},
we will adopt an approximate MLE method to improve the efficiency
based on the decomposition strategy.

\subsection{Convergence analysis\label{sub:Convergence-analysis}}

The MLE method of stationary distribution in Section \ref{sub:Maximum-likelihood-estimation}
is motivated by the assumption that $x_{0:M_{k}}^{k}$ is a Markov
chain. Interestingly, it turns out that the Markov property is not
necessary for the convergence of MLE. In this section we will prove
the convergence of MLE under more general conditions.

First of all, we provide an intuitive explanation for why the MLE
can work for non-Markovian stochastic processes. Generally speaking,
if there is no other knowledge available, $T^{k}$ can be estimated
as $\hat{T}^{k}=\left[\hat{T}_{ij}^{k}\right]$ with $\hat{T}_{ij}^{k}$
being the fraction of observed transitions from the $i$-th state
to the $j$-th state:
\begin{equation}
\hat{T}_{ij}^{k}=\frac{C_{ij}^{k}}{\sum_{l}C_{il}^{k}}\label{eq:count_estimate}
\end{equation}
But the transition matrix estimates obtained from \eqref{eq:count_estimate}
generally do not satisfy the detailed balance condition and do not
share the same unbiased stationary distribution for finite-time simulations.
Therefore we search for the feasible transition matrices which are
close to $\hat{T}^{1},\ldots,\hat{T}^{K}$: 
\begin{equation}
\begin{array}{lll}
\min\limits _{\pi,T^{1},\ldots,T^{K}} & \sum_{i,k}\hat{w}_{k}\hat{\pi}_{i}^{k}\mathrm{KL}\left(\hat{T}_{i}^{k}||T_{i}^{k}\right)\\
s.t. & \pi\textrm{ is a probability vector}\\
 & T^{k}\textrm{ is a transition matrix}\\
 & \mathrm{diag}\left(\Lambda^{k}\pi\right)T^{k}=T^{k\mathrm{T}}\mathrm{diag}\left(\Lambda^{k}\pi\right), & k=1,\ldots K
\end{array}\label{eq:T_modification}
\end{equation}
where $\hat{w}_{k}=M_{k}/M$ is the weight of the simulation $k$,
$M=\sum_{k}M_{k}$ denotes the total length of performed simulations,
$\hat{\pi}_{i}^{k}\propto\sum_{j}C_{ij}^{k}$ is an estimate of $\pi_{i}^{k}$,
and $\mathrm{KL}\left(\hat{T}_{i}^{k}||T_{i}^{k}\right)=\sum_{j}\hat{T}_{ij}^{k}\log\left(\hat{T}_{ij}^{k}/T_{ij}^{k}\right)$
denotes the KL divergence \cite{kullback1997information} between
the $i$-th rows of $\hat{T}^{k}$ and $T^{k}$. $\hat{w}_{k}\hat{\pi}_{i}^{k}$
gives the fraction of occurrence of state $i$ in simulation $k$.
Note that
\begin{eqnarray}
\sum_{i,k}\hat{w}_{k}\hat{\pi}_{i}^{k}\mathrm{KL}\left(\hat{T}_{i}^{k}||T_{i}^{k}\right) & = & \sum_{i,j,k}\hat{w}_{k}\hat{\pi}_{i}^{k}\hat{T}_{ij}^{k}\left(\log\hat{T}_{ij}^{k}-\log T_{ij}^{k}\right)\nonumber \\
 & = & \frac{1}{M}\left(\sum_{i,j,k}C_{ij}^{k}\left(\log\hat{T}_{ij}^{k}\right)-\sum_{i,j,k}C_{ij}^{k}\left(\log T_{ij}^{k}\right)\right)\label{eq:equivalence-ml-kl}
\end{eqnarray}
and thus \eqref{eq:T_modification} is equivalent to \eqref{eq:original_ml}.
Therefore the MLE can be considered to be a projector which projects
the counting estimates \eqref{eq:count_estimate} onto the feasible
space. Fig.~\ref{fig:projector} shows the relationship between different
estimates and the true values of $\left(T^{1},\ldots,T^{K}\right)$,
where the estimates obtained by counting converge to their true values
due to the wide-sense ergodicity of simulations.

\begin{figure}
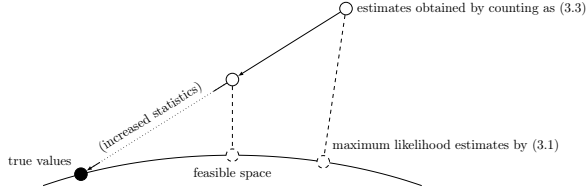

\begin{centering}
\tikzscaled[0.6\textwidth]{
\def\x{3}
	\draw[domain=-5:5,smooth,variable=\x] plot ({\x},{sqrt(1-\x*\x/100)*6-6});
	\node [circle, draw, fill=white, dashed, label=45:{maximum likelihood estimates by \eqref{eq:original_ml}}] (ML1) at (2.3993,-0.1753) {};
	\node [circle, draw, fill=white, dashed] (ML2) at (0,0) {};
	\node [circle, draw, fill=white, label=right:{estimates obtained by counting as \eqref{eq:count_estimate}}] (Count1) at (3,3.875675) {};
	\node [circle, draw, fill=white] (Count2) at (0,2) {};
	\node [circle, draw, fill, label=135:true values] (True) at (-4,-0.5009) {};
	\draw [dashed] (Count1) -- (ML1);
	\draw [dashed] (Count2) -- (ML2);
	\draw (0,-0.5) node {feasible space};
	\draw (Count2)--(-0.5,1.6873875);
	\draw [-latex](-3.5,-0.1882875)--(True);
	\draw [-latex](Count1) -- (Count2);
	\draw [dotted](Count2) -- (True) node[midway,sloped,above] {(increased statistics)};
}
\par\end{centering}

\caption{Relationship between the estimates obtained by counting, the maximum
likelihood estimates and the true values of transition matrices\label{fig:projector}}
\end{figure}

\begin{remark}{\rm In \cite{rached2004kullback}, a KL divergence
rate was derived to measure the distance between Markov chains. Suppose
$\left\{ x'_{t}\right\} $ and $\left\{ x''_{t}\right\} $ are two
Markov chains on the same state set with stationary distributions
$\pi',\pi''$ and transition matrices $T',T''$, then the KL divergence
rate between them can be defined as
\begin{eqnarray}
\mathrm{KLR}\left(\pi',T'||\pi'',T''\right) & \triangleq & \frac{1}{t+1}\lim_{t\to\infty}\mathrm{KL}\left(x'_{0:t}||x''_{0:t}\right)\nonumber \\
 & = & \sum_{i}\pi'_{i}\mathrm{KL}\left(T'_{i}||T''_{i}\right)
\end{eqnarray}
where $\pi'_{i}$ denotes the $i$-th element of $\pi'$ and $T'_{i},T''_{i}$
denote the $i$-th rows of $T',T''$. It is easy to see that the objective
function of \eqref{eq:T_modification} is equivalent to the weighted
sum of $\mathrm{KLR}\left(\hat{\pi}^{k},\hat{T}^{k}||\pi^{k},T^{k}\right)$
with weight $\hat{w}_{k}$.}\end{remark}

The equivalence of \eqref{eq:original_ml} and \eqref{eq:T_modification}
implies that the consistency of the MLE can hold without assuming
the Markov property if it can be shown that the error of the MLE converges
to zero when $\hat{T}^{1},\ldots,\hat{T}^{K}$ converge to their true
values. Before proving the main theorems, we make some assumptions
and introduce some notation. 

Unless stated otherwise, all convergence statements are made with
respect to $M\to\infty$ in this paper.

\begin{assumption}\label{ass:potential}$U^{1},\ldots,U^{k}$ and
$\bar{V}$ are finite.\end{assumption}

\begin{assumption}\label{ass:weights}The limit $\hat{w}\stackrel{p}{\to}\bar{w}$
exists and $\bar{w}\succ0$, where $\hat{w}=\left[\hat{w}_{i}\right]$
and $\bar{w}=\left[\bar{w}_{i}\right]$.\end{assumption}

\begin{assumption}\label{ass:simulations}$x_{0:M_{k}}^{1},\ldots,x_{0:M_{k}}^{K}$
are all asymptotically stationary and wide-sense ergodic processes
with detailed balance.\end{assumption}

\begin{assumption}\label{ass:X-property}For any $i,j,k$, if there
exists some $t$ such that $\Pr(I_{\mathcal{S}^{k}}\left(x_{t}^{k}\right)=i,I_{\mathcal{S}^{k}}\left(x_{t+1}^{k}\right)=j)>0$,
then $\lim_{n\to\infty}\Pr(I_{\mathcal{S}^{k}}\left(x_{n}^{k}\right)=i,I_{\mathcal{S}^{k}}\left(x_{n+1}^{k}\right)=j)>0$.\end{assumption}

\begin{assumption}\label{ass:pi-property} $\pi=\bar{\pi}$ is the
unique solution of the following set of equations and inequalities:
\begin{equation}
\left\{ \begin{array}{l}
\mathrm{diag}\left(\Lambda^{k}\pi\right)\bar{T}^{k}=\bar{T}^{k\mathrm{T}}\mathrm{diag}\left(\Lambda^{k}\pi\right),\quad\text{for }k=1,\ldots,K\\
\sum_{i}\pi_{i}=1\text{ and }\pi\succeq0
\end{array}\right.\label{eq:pi-unique}
\end{equation}
\end{assumption}

The above assumption means the unbiased stationary distribution can
be uniquely determined if all the transition matrices are given.

Furthermore, here we let $\theta$ be the vector consisting of elements
of the unbiased stationary distribution $\pi$ and transition matrices
$T^{1},\ldots,T^{K}$, $X^{k}=\left[X_{ij}^{k}\right]=\left[\pi_{i}^{k}T_{ij}^{k}\right]$
with $\bar{X}^{k}=\left[\bar{X}_{ij}^{k}\right]=\left[\bar{\pi}_{i}^{k}\bar{T}_{ij}^{k}\right]$
and $\hat{X}^{k}=\left[\hat{X}_{ij}^{k}\right]=\left[\hat{\pi}_{i}^{k}\hat{T}_{ij}^{k}\right]$,
$\Theta$ be the feasible set defined by constraints in \eqref{eq:original_ml},
$V_{Xw}=\left(w_{1}\mathcal{V}\left(X^{1}\right)^{\mathrm{T}},\ldots,w_{K}\mathcal{V}\left(X^{K}\right)^{\mathrm{T}}\right)^{\mathrm{T}}$
and
\begin{equation}
\hat{Q}\left(\theta\right)=\sum_{i,j,k}\hat{w}_{k}\hat{X}_{ij}^{k}\log T_{ij}^{k},\quad\bar{Q}\left(\theta\right)=\sum_{i,j,k}\bar{w}_{k}\bar{X}_{ij}^{k}\log T_{ij}^{k}
\end{equation}
Notice that functions $\hat{Q}\left(\theta\right)$ and $\bar{Q}\left(\theta\right)$
is linear in parameters $\hat{V}_{Xw}$ and $\bar{V}_{Xw}$, then
we can construct a function $\Phi\left(\theta\right)$ such that
\begin{equation}
\hat{Q}\left(\theta\right)=\hat{V}_{Xw}^{\mathrm{T}}\Phi\left(\theta\right),\quad\bar{Q}\left(\theta\right)=\bar{V}_{Xw}^{\mathrm{T}}\Phi\left(\theta\right)
\end{equation}

Based on the above assumptions and notations, the maximum likelihood
estimate can be expressed as $\tilde{\theta}=\arg\max_{\theta\in\Theta}\hat{Q}\left(\theta\right)$,
and we have the following theorems on the convergence of $\tilde{\theta}$.
\begin{theorem}
\label{thm:mle-convergence-p}If Assumptions \ref{ass:potential}-\ref{ass:pi-property}
hold, then $\tilde{\theta}\stackrel{p}{\to}\bar{\theta}$.\end{theorem}
\begin{proof}
See Appendix \ref{sec:Proof-of-Theorem-mle-convergence-p}.\end{proof}
\begin{theorem}
\label{thm:asymptotic-convergence}If Assumptions \ref{ass:potential}-\ref{ass:pi-property}
hold and the following conditions are satisfied:
\begin{enumerate}
\item \label{enu:condition-clt-thm-asymptotic-convergence}For each $k$,
there exists a $\Sigma_{X}^{k}$ such that
\begin{equation}
\sqrt{M_{k}}\left(\mathcal{V}\left(\hat{X}^{k}\right)-\mathcal{V}\left(\bar{X}^{k}\right)\right)\stackrel{d}{\to}\mathcal{N}\left(\mathbf{0},\Sigma_{X}^{k}\right)\label{eq:X-clt}
\end{equation}

\item $\tilde{\theta}$ is obtained with variable reduction as in \eqref{eq:reduced_ml}.
\item Diagonal elements of $\bar{X}^{1},\ldots,\bar{X}^{K}$ are positive.
\item All $K$ simulations are statistically independent.
\item $\sqrt{M}\left(\hat{w}-\bar{w}\right)\stackrel{p}{\to}\mathbf{0}$.
\item $H=\nabla_{\theta_{r}\theta_{r}}\bar{Q}\left(\theta\left(\bar{\theta}_{r}\right)\right)$
is nonsingular with $\theta_{r}$ the vector consisting of $\{T_{ij}^{k}|\bar{X}_{ij}^{k}>0,i<j\}$
and $\{\pi_{1},\ldots,\pi_{\left|S\right|-1}\}$.
\end{enumerate}
then
\begin{equation}
\sqrt{M}\left(\tilde{\theta}_{r}-\bar{\theta}_{r}\right)\stackrel{d}{\to}\mathcal{N}\left(0,H^{-1}\Sigma H^{-1}\right)
\end{equation}
where $\Sigma=\left(\nabla_{\theta_{r}}\Phi\left(\theta\left(\bar{\theta}_{r}\right)\right)\right)^{\mathrm{T}}\Sigma_{Xw}\nabla_{\theta_{r}}\Phi\left(\theta\left(\bar{\theta}_{r}\right)\right)$
and $\Sigma_{Xw}=\mathrm{diag}\left(\bar{w}_{1}\Sigma_{X}^{1},\ldots,\bar{w}_{K}\Sigma_{X}^{K}\right)$.\end{theorem}
\begin{proof}
See Appendix \ref{sec:Proof-of-Theorem-asymptotic-convergence}.
\end{proof}
\begin{remark}\label{rem:clt}{\rm Since $\hat{X}^{k}=C^{k}/M^{k}$,
Condition \ref{enu:condition-clt-thm-asymptotic-convergence} stated
in Theorem \ref{thm:asymptotic-convergence} means that the central
limit theorem holds for $\left\{ \mathcal{V}\left(\Delta C_{t}^{k}\right)\right\} $,
and can be justified by using Markov chain central limit theorems
\cite{jones2004markov} in many practical situations. For example,
one sufficient condition for this is that for each simulation $k$,
$\left\{ x_{t}^{k}\right\} $ is obtained by coarse-graining a geometrically
ergodic Markov chain $\left\{ y_{t}^{k}\right\} $ with $x_{t}^{k}=f^{k}\left(y_{t}^{k}\right)$,
which implies that $\mathcal{S}^{k}$ defines a partition of the state
space of $\left\{ y_{t}^{k}\right\} $ and each $s\in\mathcal{S}^{k}$
is associated with a cluster $\left\{ r|f^{k}\left(r\right)=s\right\} $
in the state space. (See Appendix \ref{sec:Proof-of-Remark-clt} for
details.)}\end{remark}

\begin{remark}{\rm In this section we only characterize the convergence
of $\tilde{\pi}$. For the MLE of free energy $V$, the consistency
and asymptotic normality are immediate consequences of Theorems \ref{thm:mle-convergence-p}
and \ref{thm:asymptotic-convergence} by considering that $V$ is
a smooth function of $\pi$. Here we omit the detailed description
and proof as they are trivial.}\end{remark}

\section{Approximate maximum likelihood estimation\label{sec:Approximate-maximum-likelihood}}

In this section, we develop an approximate MLE method based on a decomposition
strategy in order to improve the efficiency of MLE, and the convergence
of the method is shown.

For convenience of analysis and computation, here we introduce two
new variables, $\underline{C}^{k}$ and $Z^{k}$. $\underline{C}^{k}=\left[\underline{C}_{ij}^{k}\right]$
is a modified count matrix used to replace $C^{k}$ in the approximate
MLE, and is assumed to satisfy the following assumption:

\begin{assumption}\label{ass:count-matrix}$\underline{C}^{1},\ldots,\underline{C}^{K}$
are irreducible matrices with positive diagonal elements, and satisfy
that $1_{\underline{C}_{ij}^{k}>0}=1_{\underline{C}_{ji}^{k}>0}$
and $1_{\underline{C}^{k}=C^{k}}\stackrel{p}{\to}1$ for all $i,j,k$.\end{assumption}

One way to perform the count matrix modification is as follows:

\begin{equation}
\underline{C}_{ij}^{k}=\left\{ \begin{array}{ll}
\max\left\{ C_{ij}^{k},\delta\right\} , & C_{ji}^{k}>0\text{ or }i=j\\
C_{ij}^{k}, & \text{otherwise}
\end{array}\right.\label{eq:modified-C}
\end{equation}
where $\delta\in\left(0,1\right)$ is a small number. (This approach
is similar to the so-called neighbor prior used in \cite{Prinz2011efficient,BeauchampEtAl_MSMbuilder2}.)
\begin{theorem}
\label{thm:modified-count-matrix}If Assumptions \ref{ass:potential},
\ref{ass:simulations} and \ref{ass:X-property} hold, and if $\bar{X}_{ii}^{k}>0$
and $\sum_{t=0}^{M_{k}}1_{I_{\mathcal{S}^{k}}\left(x_{t}^{k}\right)=i}>0$
for all $i,k$, then the modified count matrices defined in \eqref{eq:modified-C}
satisfy Assumption \ref{ass:count-matrix}.\end{theorem}
\begin{proof}
See Appendix \ref{sec:Proof-of-Theorem-modified-count-matrix}.
\end{proof}
~

The variable $Z^{k}=\left[Z_{ij}^{k}\right]$ is defined by
\begin{equation}
\exp\left(-Z_{ij}^{k}\right)\propto X_{ij}^{k}\label{eq:Z-definition}
\end{equation}
which can be interpreted as the ``free energy matrix'' of state transitions
in the $k$-th simulation because $\exp\left(-Z_{ij}^{k}\right)\propto\lim_{t\to\infty}\Pr\left(I_{\mathcal{S}}\left(x_{t}^{k}\right)=i,I_{\mathcal{S}}\left(x_{t+1}^{k}\right)=j\right)$.
Like the free energies of states, we also assume that $\sum_{\left(i,j\right)\in\left\{ \left(i,j\right)|X_{ij}^{k}>0\right\} }Z_{ij}^{k}=0$
such that \eqref{eq:Z-definition} has a unique solution with given
$X^{k}$.

\subsection{Decomposition and approximation of MLE problem}

Under the above assumption and variable definitions and replacing
$C^{k}$ by $\underline{C}^{k}$, \eqref{eq:reduced_ml} can be decomposed
into $K+1$ subproblems as follows:

\paragraph*{Local subproblems ($k=1\ldots,K$)}

\begin{equation}
\begin{array}{lll}
L^{k}\left(V^{k}\right)=\max\limits _{\left\{ Z_{ij}^{k}|\underline{C}_{ij}^{k}>0\right\} } & -\sum_{i,j}\underline{C}_{ij}^{k}Z_{ij}^{k}+\sum_{i}\underline{C}_{i}^{k}Z_{i}^{k}\\
s.t. & Z^{k}=Z^{k\mathrm{T}}\\
 & \sum_{\left(i,j\right)\in\left\{ \left(i,j\right)|\underline{C}_{ij}^{k}>0\right\} }Z_{ij}^{k}=0\\
 & V_{Z}^{k}=V^{k}
\end{array}\label{eq:local-subproblem-original}
\end{equation}
where $Z_{ij}^{k}=\infty$ if $\underline{C}_{ij}^{k}=0$, $0\cdot\infty$
is set to be $0$, $Z_{i}^{k}=-\log\sum_{j}\exp\left(-Z_{ij}^{k}\right)$,
$\underline{C}_{i}^{k}=\sum_{j}\underline{C}_{ij}^{k}$ and $V_{Z}^{k}=\left(Z_{1}^{k},\ldots,Z_{\left|\mathcal{S}^{k}\right|}^{k}\right)^{\mathrm{T}}-\sum_{i}Z_{i}^{k}/\left|\mathcal{S}^{k}\right|$
denotes the the state potential obtained from $Z^{k}$. A brief description
of the objective function and constraints is given below:
\begin{enumerate}
\item The objective function is the log-likelihood of a Markov chain model
with free energy matrix $Z^{k}$ given $\underline{C}^{k}$, because
$\log T_{ij}^{k}=Z_{i}^{k}-Z_{ij}^{k}$.
\item The first two constraints guarantee that the detailed balance condition
holds for the $k$-th simulation and the finite elements of $Z^{k}$
have zero mean, and can be eliminated by substituting 
\begin{eqnarray}
 &  & Z_{ij}^{k}\left(\rho^{k}\right)=\nonumber \\
 &  & \left\{ \begin{array}{ll}
Z_{ji}^{k}, & \underline{C}_{ij}^{k}>0,i>j\\
-\sum_{\left(i,j\right)\in\left\{ \left(i,j\right)|\underline{C}_{ij}^{k}>0\right\} \backslash\left(\left|\mathcal{S}^{k}\right|,\left|\mathcal{S}^{k}\right|\right)}Z_{ij}^{k}, & \left(i,j\right)=\left(\left|\mathcal{S}^{k}\right|,\left|\mathcal{S}^{k}\right|\right)
\end{array}\right.
\end{eqnarray}
in the objective function, where $\rho^{k}$ denotes the vector consisting
of $\{Z_{ij}^{k}|\underline{C}_{ij}^{k}>0,i\le j,\left(i,j\right)\neq\left(\left|\mathcal{S}^{k}\right|,\left|\mathcal{S}^{k}\right|\right)\}$.
\item The third constraint means the state potential obtained from $Z^{k}$
is equal to the given $V^{k}$. Note that both $V_{Z}^{k}$ and $V^{k}$
are zero-mean vectors, so this constraint can be equivalently written
as $\left[\begin{array}{cc}
\mathbf{I} & \mathbf{0}\end{array}\right]V_{Z}^{k}=\left[\begin{array}{cc}
\mathbf{I} & \mathbf{0}\end{array}\right]V^{k}$.
\end{enumerate}
Thus, $L^{k}\left(V^{k}\right)$ is equal to the maximum log-likelihood
of a reversible Markov chain model under the condition that $V^{k}$
is known, and the local subproblem has the following equivalent formulation:
\begin{equation}
\begin{array}{lll}
L^{k}\left(V^{k}\right)=\max\limits _{\rho^{k}} & -\sum_{i,j}\underline{C}_{ij}^{k}Z_{ij}^{k}+\sum_{i}\underline{C}_{i}^{k}Z_{i}^{k}\\
s.t. & \left[\begin{array}{cc}
\mathbf{I} & \mathbf{0}\end{array}\right]V_{Z}^{k}=\left[\begin{array}{cc}
\mathbf{I} & \mathbf{0}\end{array}\right]V^{k}
\end{array}\label{eq:local-subproblem}
\end{equation}
where $Z^{k}$ and $V_{Z}^{k}$ are both viewed as functions of $\rho^{k}$.
In comparison with \eqref{eq:reduced_ml}, each local subproblem with
a given $V^{k}$ only depends on one simulation and independent of
the others.

\paragraph*{Global subproblem}

\begin{equation}
\begin{array}{ll}
\max\limits _{V} & \sum_{k=1}^{K}L^{k}\left(A^{k}\left(V+U^{k}\right)\right)\\
s.t. & \mathbf{1}^{\mathrm{T}}V=0
\end{array}\label{eq:global-subproblem}
\end{equation}
The objective of this subproblem is to maximize the sum of maximum
log-likelihoods for $K$ biased simulations by selecting the best
unbiased free energy.

Obviously, if local subproblems can be efficiently and exactly calculated,
the MLE of unbiased free energy $V$ can be easily obtained from \eqref{eq:global-subproblem}
because only $\left|\mathcal{S}\right|$ variables and one linear
equality constraint are involved in the global subproblem. However,
local subproblems are nonlinear optimization problems with nonlinear
equality constraints, of which solutions are generally computationally
expensive and time-consuming. It is therefore necessary to find tractable
approximate solutions of local subproblems. In the following, we address
the approximations and solutions of the subproblems in detail.

\subsubsection{Local subproblems\label{sub:Local-subproblems}}

Inspired by the fact that equality constrained quadratic programming
problems can be solved analytically, here we present a tractable approximate
solution of \eqref{eq:global-subproblem} based on Taylor expansions
of the objective function and the constraint equation. The approximate
solution method consists of the following three steps.

\paragraph*{Search for the optimal solution of \eqref{eq:local-subproblem} without
the constraint on the free energy}

Omitting the constraint on the free energy, \eqref{eq:local-subproblem}
can be written as
\begin{equation}
\max\limits _{\rho^{k}}L^{k}=-\sum_{i,j}\underline{C}_{ij}^{k}Z_{ij}^{k}+\sum_{i}\underline{C}_{i}^{k}Z_{i}^{k}\label{eq:local-subproblem-without-energy-constraint}
\end{equation}
and the gradient and Hessian matrix of the objective function used
in optimization can be simply obtained by the following equations:
\begin{eqnarray}
\frac{\partial Z_{i}^{k}}{\partial Z_{lj}^{k}} & = & 1_{l=i}\cdot T_{ij}^{k}\label{eq:Zik-grad}\\
\frac{\partial^{2}Z_{i}^{k}}{\partial Z_{jm}^{k}\partial Z_{ln}^{k}} & = & 1_{j=l=i}\cdot\left(T_{im}^{k}T_{in}^{k}-1_{m=n}T_{im}^{k}\right)\label{eq:Zik-hess}
\end{eqnarray}
Furthermore, we can verify that the Hessian matrix $H_{\rho}^{k}\left(\underline{C}^{k},\rho^{k}\right)=\sum_{i}\underline{C}_{i}^{k}\nabla_{\rho^{k}\rho^{k}}Z_{i}^{k}$
is negative-definite under Assumption \ref{ass:count-matrix} (see
Appendix \ref{sec:Properties-of-Hrhok}), which means \eqref{eq:local-subproblem}
is a convex optimization problem with a single global optimum that
can be solved efficiently using any local optimization algorithm.
Here we denote the solution of \eqref{eq:local-subproblem-without-energy-constraint}
by $\check{\rho}^{k}$, and we will employ a conjugate gradient algorithm
\cite{yuan2009conjugate} to find $\check{\rho}^{k}$ in our numerical
experiments.

\paragraph*{Construct a quadratic programming approximation of \eqref{eq:local-subproblem}
by Taylor expansions}

Replacing the objective function and the nonlinear constraint of \eqref{eq:local-subproblem}
by their second-order and first-order Taylor expansions around $\check{\rho}^{k}$,
we can get the following approximate local subproblem:
\begin{equation}
\begin{array}{lll}
L_{q}^{k}\left(V^{k}\right)=\max\limits _{\rho^{k}} & \check{L}^{k}+\frac{1}{2}\left(\rho^{k}-\check{\rho}^{k}\right)^{\mathrm{T}}H_{\rho}^{k}\left(\underline{C}^{k},\check{\rho}^{k}\right)\left(\rho^{k}-\check{\rho}^{k}\right)\\
s.t. & \left[\begin{array}{cc}
\mathbf{I} & \mathbf{0}\end{array}\right]\nabla_{\rho^{k}}V_{Z}^{k}\left(\check{\rho}^{k}\right)\left(\rho^{k}-\check{\rho}^{k}\right)=\left[\begin{array}{cc}
\mathbf{I} & \mathbf{0}\end{array}\right]\left(V^{k}-\check{V}^{k}\right)
\end{array}\label{eq:local-subproblem-approximate}
\end{equation}
where $\check{V}^{k}=V_{Z}^{k}\left(\check{\rho}^{k}\right)$ and
$\check{L}^{k}=L^{k}\left(\check{V}^{k}\right)$.

\paragraph*{Calculate $L_{q}^{k}\left(V^{k}\right)$ from \eqref{eq:local-subproblem-approximate}}

It is easy to verify that $\left[\begin{array}{cc}
\mathbf{I} & \mathbf{0}\end{array}\right]\nabla_{\rho^{k}}V_{Z}^{k}\left(\check{\rho}^{k}\right)$ is full row rank. Then using the Lagrange multiplier method, we get
\begin{equation}
L_{q}^{k}\left(V^{k}\right)=\frac{1}{2}\left(V^{k}-\check{V}^{k}\right)^{\mathrm{T}}\Xi^{k}\left(\underline{C}^{k},\check{\rho}^{k}\right)\left(V^{k}-\check{V}^{k}\right)+\check{L}^{k}\label{eq:Lqk}
\end{equation}
where
\begin{eqnarray}
\Xi^{k}\left(\underline{C}^{k},\check{\rho}^{k}\right) & =\left[\begin{array}{c}
\mathbf{I}\\
\mathbf{\mathbf{0}^{\mathrm{T}}}
\end{array}\right] & \left(\left[\begin{array}{cc}
\mathbf{I} & \mathbf{0}\end{array}\right]\nabla_{\rho^{k}}V_{Z}^{k}\left(\check{\rho}^{k}\right)\left(H_{\rho}^{k}\left(\underline{C}^{k},\check{\rho}^{k}\right)\right)^{-1}\right.\nonumber \\
 &  & \cdot\left.\left(\left[\begin{array}{cc}
\mathbf{I} & \mathbf{0}\end{array}\right]\nabla_{\rho^{k}}V_{Z}^{k}\left(\check{\rho}^{k}\right)\right)^{\mathrm{T}}\right)^{-1}\left[\begin{array}{cc}
\mathbf{I} & \mathbf{0}\end{array}\right]\label{eq:Xik}
\end{eqnarray}
is negative-semidefinite and satisfies $\left[\begin{array}{cc}
\mathbf{I} & \mathbf{0}\end{array}\right]\Xi^{k}\left(\underline{C}^{k},\check{\rho}^{k}\right)\left[\begin{array}{cc}
\mathbf{I} & \mathbf{0}\end{array}\right]^{\mathrm{T}}<0$.

\subsubsection{Global subproblem}

Substituting \eqref{eq:Lqk} into \eqref{eq:global-subproblem} leads
to the approximate global subproblem:
\begin{equation}
\begin{array}{ll}
\max\limits _{V} & \sum_{k=1}^{K}\frac{1}{2}\left(A^{k}\left(V+U^{k}\right)-\check{V}^{k}\right)^{\mathrm{T}}\Xi^{k}\left(\underline{C}^{k},\check{\rho}^{k}\right)\\
 & \cdot\left(A^{k}\left(V+U^{k}\right)-\check{V}^{k}\right)+\check{L}^{k}\\
s.t. & \mathbf{1}^{\mathrm{T}}V=0
\end{array}\label{eq:global-subproblem-approximate}
\end{equation}
By applying the KKT conditions, it is easy to verify that the solution
of \eqref{eq:global-subproblem-approximate} is
\begin{equation}
\check{V}=\left(\sum_{k=1}^{K}A^{k\mathrm{T}}\Xi^{k}\left(\underline{C}^{k},\check{\rho}^{k}\right)A^{k}\right)^{+}\left(\sum_{k=1}^{K}A^{k\mathrm{T}}\Xi^{k}\left(\underline{C}^{k},\check{\rho}^{k}\right)\left(\check{V}^{k}-A^{k}U^{k}\right)\right)\label{eq:approximate-mle-V}
\end{equation}

\begin{remark}{\rm We can now explain why variables $Z^{k},V^{k}$
are used instead of $T^{k},\pi^{k}$ in the problem decomposition
and approximation. First, the approximate local subproblem based on
Taylor expansions w.r.t.\ $T^{k},\pi^{k}$ involves inequality constraints
and therefore has no analytic solution. Secondly, even if we can get
a quadratic polynomial approximation to the function $L^{k}\left(V^{k}\left(\pi^{k}\right)\right)$
by some method (e.g., omitting the inequality constraints in approximate
local subproblems as in \cite{sakuraba2009multiple}), the resulting
approximate global problem cannot be solved as easily as \eqref{eq:global-subproblem-approximate}
because both $V$ and $\pi$ are nonlinear functions of $\pi^{k}$.}\end{remark}

\subsection{Approximate MLE algorithm\label{sub:Approximate-MLE-algorithm}}

Summarizing the above discussions, the approximate MLE algorithm of
the unbiased free energy is described as follows:
\begin{description}
\item [{Step 1.}] For $k=1,\ldots,K$, calculate the modified count matrix
$\underline{C}^{k}$ by \eqref{eq:modified-C}.
\item [{Step 2.}] For $k=1,\ldots,K$, search for the optimal point $\check{\rho}^{k}$
of the programming problem \eqref{eq:local-subproblem-without-energy-constraint}.
\item [{Step 3.}] For $k=1,\ldots,K$, calculate $\Xi^{k}\left(\underline{C}^{k},\check{\rho}^{k}\right)$
by \eqref{eq:Xik}.
\item [{Step 4.}] Calculate the approximate maximum likelihood estimate
of $V$ by \eqref{eq:approximate-mle-V}.
\end{description}

\subsection{Convergence analysis\label{sub:Convergence-analysis-approximate-mle}}

In this section, we will analyze consistency and asymptotic normality
of the approximate MLE as the exact MLE under the assumptions stated
in Section \ref{sub:Convergence-analysis} and Assumption \ref{ass:count-matrix}.
Before introducing the main theorem, some definitions and a lemma
are needed. Let $\underline{\rho}^{k}$ and $\bar{\rho}^{k}$ be vectors
consisting of $\left\{ Z_{ij}^{k}|\bar{X}_{ij}^{k}>0,i\le j,\left(i,j\right)\neq\left(\left|\mathcal{S}^{k}\right|,\left|\mathcal{S}^{k}\right|\right)\right\} $
and $\left\{ \bar{Z}_{ij}^{k}|\bar{X}_{ij}^{k}>0,i\le j,\left(i,j\right)\neq\left(\left|\mathcal{S}^{k}\right|,\left|\mathcal{S}^{k}\right|\right)\right\} $,
\begin{equation}
\check{Q}^{k}\left(\underline{\rho}^{k}\right)=-\sum_{i,j}\hat{w}_{k}\hat{\underline{X}}_{ij}^{k}Z_{ij}^{k}\left(\underline{\rho}^{k}\right)+\sum_{i,j}\hat{w}_{k}\hat{\underline{X}}_{ij}^{k}Z_{i}^{k}\left(\underline{\rho}^{k}\right)
\end{equation}
\begin{equation}
\bar{Q}^{k}\left(\underline{\rho}^{k}\right)=-\sum_{i,j}\bar{w}_{k}\bar{X}_{ij}^{k}Z_{ij}^{k}\left(\underline{\rho}^{k}\right)+\sum_{i,j}\bar{w}_{k}\bar{X}_{ij}^{k}Z_{i}^{k}\left(\underline{\rho}^{k}\right)
\end{equation}
and $\hat{\underline{X}}^{k}=\left[\hat{\underline{X}}_{ij}^{k}\right]=\underline{C}^{k}/M^{k}$.
Like $\hat{Q}\left(\theta\right)$ and $\bar{Q}\left(\theta\right)$
in Section \ref{sub:Convergence-analysis}, $\check{Q}^{k}\left(\underline{\rho}^{k}\right)$
and $\bar{Q}^{k}\left(\underline{\rho}^{k}\right)$ can also be written
as
\begin{equation}
\check{Q}^{k}\left(\underline{\rho}^{k}\right)=\hat{w}_{k}\mathcal{V}\left(\hat{\underline{X}}^{k}\right)^{\mathrm{T}}\Phi^{k}\left(\underline{\rho}^{k}\right),\quad\bar{Q}^{k}\left(\underline{\rho}^{k}\right)=\bar{w}_{k}\mathcal{V}\left(\bar{X}^{k}\right)^{\mathrm{T}}\Phi^{k}\left(\underline{\rho}^{k}\right)
\end{equation}
and $\check{\rho}^{k}=\arg\max_{\underline{\rho}^{k}}\check{Q}^{k}\left(\underline{\rho}^{k}\right)$
if $\mathrm{dim}\left(\check{\rho}^{k}\right)=\mathrm{dim}\left(\bar{\rho}^{k}\right)$.
\begin{lemma}
\label{lem:rhoV}Provided that Assumptions \ref{ass:potential}-\ref{ass:pi-property}
and \ref{ass:count-matrix} hold.\end{lemma}
\begin{enumerate}
\item $\check{\rho}^{k}\stackrel{p}{\to}\bar{\rho}^{k}$ and $\check{V}^{k}\stackrel{p}{\to}\bar{V}^{k}$.
\item $\sqrt{M_{k}}\left(\check{V}^{k}-\bar{V}^{k}\right)\stackrel{d}{\to}\mathcal{N}\left(\mathbf{0},\Sigma_{V}^{k}\right)$
if \eqref{eq:X-clt} is satisfied, where
\begin{eqnarray}
\Sigma_{V}^{k} & = & \nabla_{\underline{\rho}^{k}}V_{Z}^{k}\left(\bar{\rho}^{k}\right)\Sigma_{\rho}^{k}\left(\nabla_{\underline{\rho}^{k}}V_{Z}^{k}\left(\bar{\rho}^{k}\right)\right)^{\mathrm{T}}\label{eq:SigmaVk}\\
\Sigma_{\rho}^{k} & = & \bar{w}_{k}^{2}\left(H_{\rho}^{k}\left(\bar{w}_{k}\bar{X}^{k},\bar{\rho}^{k}\right)\right)^{-1}\left(\nabla_{\underline{\rho}^{k}}\Phi^{k}\left(\bar{\rho}^{k}\right)\right)^{\mathrm{T}}\Sigma_{X}^{k}\nonumber \\
 &  & \cdot\nabla_{\underline{\rho}^{k}}\Phi^{k}\left(\bar{\rho}^{k}\right)\left(H_{\rho}^{k}\left(\bar{w}_{k}\bar{X}^{k},\bar{\rho}^{k}\right)\right)^{-1}\label{eq:Sigmarhok}
\end{eqnarray}
\end{enumerate}
\begin{proof}
See Appendix \ref{sec:Proof-of-Lemma-rhoV}.\end{proof}
\begin{theorem}
\label{thm:convergence-approximate-mle}Provided that Assumptions
\ref{ass:potential}-\ref{ass:pi-property} and \ref{ass:count-matrix}
 hold.\end{theorem}
\begin{enumerate}
\item $\check{V}\stackrel{p}{\to}\bar{V}$.
\item $\sqrt{M}\left(\check{V}-\bar{V}\right)\stackrel{d}{\to}\mathcal{N}\left(\mathbf{0},\Sigma_{V}\right)$
if \eqref{eq:X-clt} is satisfied for all $k$ and $K$ simulations
are statistically independent, where
\begin{eqnarray}
\Sigma_{V} & = & \left(\sum_{k=1}^{K}A^{k\mathrm{T}}\Xi^{k}\left(\bar{w}^{k}\bar{X}^{k},\bar{\rho}^{k}\right)A^{k}\right)^{+}\nonumber \\
 &  & \cdot\left(\sum_{k=1}^{K}A^{k\mathrm{T}}\Xi^{k}\left(\bar{w}^{k}\bar{X}^{k},\bar{\rho}^{k}\right)\Sigma_{V}^{k}\Xi^{k}\left(\bar{w}^{k}\bar{X}^{k},\bar{\rho}^{k}\right)A^{k}\right)\nonumber \\
 &  & \cdot\left(\sum_{k=1}^{K}A^{k\mathrm{T}}\Xi^{k}\left(\bar{w}^{k}\bar{X}^{k},\bar{\rho}^{k}\right)A^{k}\right)^{+}\label{eq:SigmaV}
\end{eqnarray}
and $\Sigma_{V}^{k}$ has the same definition as in Lemma \ref{lem:rhoV}.\end{enumerate}
\begin{proof}
See Appendix \ref{sec:Proof-of-Theorem-convergence-approximate-mle}.
\end{proof}

\subsection{Error analysis\label{sub:Error-analysis}}

According to Theorem \ref{thm:convergence-approximate-mle}, the estimation
error of $\check{V}$ follows approximately a multivariate normal
distribution $\mathcal{N}\left(\mathbf{0},\Sigma_{V}/M\right)$ when
$M$ is large enough, and $\Sigma_{V}$ can be estimated by \eqref{eq:SigmaVk},
\eqref{eq:Sigmarhok} and \eqref{eq:SigmaV} with replacing $\bar{w}_{k},\bar{X}^{k},\bar{\rho}^{k}$
with $\hat{w}_{k},\hat{\underline{X}}^{k},\check{\rho}^{k}$ if $\Sigma_{X}^{k}$
is given. Therefore, the remaining key problem is how to estimate
$\Sigma_{X}^{k}$. In this section we present an algorithm for estimating
$\Sigma_{X}^{k}$ based on the following assumption, which is similar
to the assumption proposed in Remark \ref{rem:clt} and implies that
each simulation is driven by a stationary, irreducible and reversible
Markov model.

\begin{assumption}\label{ass:error}For each simulation $k$, $x_{t}^{k}$
can be expressed as a function of a latent variable $y_{t}^{k}$ with
$x_{t}^{k}=f^{k}\left(y_{t}^{k}\right)$ and $\left\{ y_{t}^{k}\right\} $
is a stationary, irreducible and reversible Markov chain.\end{assumption}

Under Assumption \ref{ass:error}, it can be seen that $\left\{ \mathcal{V}\left(\Delta C_{t}^{k}\right)\right\} $
is also a stationary process. To describe the estimation algorithm,
we also need some new notation. We denote by
\begin{equation}
\kappa^{k}\left(h\right)=\cov\left(\mathcal{V}\left(\Delta C_{t}^{k}\right),\mathcal{V}\left(\Delta C_{t+h}^{k}\right)\right)
\end{equation}
the autocovariance of $\left\{ \mathcal{V}\left(\Delta C_{t}^{k}\right)\right\} $
with lag $h$. It is clear that the $((i-1)\left|\mathcal{S}^{k}\right|+j,(m-1)\left|\mathcal{S}^{k}\right|+n)$-th
element of $\kappa^{k}\left(h\right)$ is equal to the covariance
between $1_{\left(I_{\mathcal{S}^{k}}\left(x_{t-1}^{k}\right),I_{\mathcal{S}^{k}}\left(x_{t}^{k}\right)\right)=\left(i,j\right)}$
and $1_{\left(I_{\mathcal{S}^{k}}\left(x_{t+h-1}^{k}\right),I_{\mathcal{S}^{k}}\left(x_{t+h}^{k}\right)\right)=\left(m,n\right)}$.
Furthermore, let $\Gamma^{k}\left(l\right)=\kappa^{k}\left(2l+1\right)+\kappa^{k}\left(2l+2\right)$,
and $\eta^{k}\left(l\right)$ be a sum of some elements of $\Gamma^{k}\left(l\right)$,
which can be represented as
\begin{eqnarray}
\eta^{k}\left(l\right) & = & \sum_{i,j}\cov\left(1_{\left(I_{\mathcal{S}^{k}}\left(x_{t-1}^{k}\right),I_{\mathcal{S}^{k}}\left(x_{t}^{k}\right)\right)=\left(j,i\right)},1_{\left(I_{\mathcal{S}^{k}}\left(x_{t+2l}^{k}\right),I_{\mathcal{S}^{k}}\left(x_{t+2l+1}^{k}\right)\right)=\left(i,j\right)}\right)\nonumber \\
 &  & +\cov\left(1_{\left(I_{\mathcal{S}^{k}}\left(x_{t-1}^{k}\right),I_{\mathcal{S}^{k}}\left(x_{t}^{k}\right)\right)=\left(j,i\right)},1_{\left(I_{\mathcal{S}^{k}}\left(x_{t+2l+1}^{k}\right),I_{\mathcal{S}^{k}}\left(x_{t+2l+2}^{k}\right)\right)=\left(i,j\right)}\right)
\end{eqnarray}

\begin{theorem}
\label{thm:SigmaX-kappa}If Assumption \ref{ass:error} and \eqref{eq:X-clt}
holds, and the series $\sum_{h=0}^{\infty}\kappa^{k}\left(h\right)$
is convergent, then
\begin{enumerate}
\item $\Sigma_{X}^{k}=\kappa^{k}\left(0\right)+\sum_{l=0}^{\infty}\left(\Gamma^{k}\left(l\right)+\Gamma^{k}\left(l\right)^{\mathrm{T}}\right)$.
\item $\eta^{k}\left(l\right)\ge\left\Vert \Gamma^{k}\left(l\right)\right\Vert _{\max}$
and $\eta^{k}\left(l\right)\le\eta^{k}\left(l+1\right)$ for $l\ge0$.
\end{enumerate}
\end{theorem}
\begin{proof}
See Appendix \ref{sec:Proof-of-Theorem-SigmaX-kappa}.
\end{proof}
~

The above theorem provides an intuitive way to estimate $\Sigma_{X}^{k}$:
\begin{equation}
\hat{\Sigma}_{X}^{k}=\hat{\kappa}^{k}\left(0\right)+\sum_{l=0}^{\left\lfloor \frac{M_{k}-3}{2}\right\rfloor }\left(\hat{\Gamma}^{k}\left(l\right)+\hat{\Gamma}^{k}\left(l\right)^{\mathrm{T}}\right)\label{eq:Sigma-X-estimate}
\end{equation}
where $\hat{\kappa}^{k}\left(0\right)$ and $\hat{\Gamma}^{k}\left(l\right)$
denote the estimates of $\kappa^{k}\left(0\right)$ and $\Gamma^{k}\left(l\right)$.
We now investigate the calculation of $\hat{\kappa}^{k}\left(0\right)$
and $\hat{\Gamma}^{k}\left(l\right)$.

\paragraph*{Estimation of $\kappa^{k}\left(0\right)$}

It is easy to verify that the $((i-1)\left|\mathcal{S}^{k}\right|+j,(m-1)\left|\mathcal{S}^{k}\right|+n)$-th
element of $\kappa^{k}\left(0\right)$ equals $1_{\left(i,j\right)=\left(m,n\right)}\bar{X}_{ij}^{k}-\bar{X}_{ij}^{k}\bar{X}_{mn}^{k}$,
therefore we can calculate the the element in the same position in
$\hat{\kappa}^{k}\left(0\right)$ by $1_{\left(i,j\right)=\left(m,n\right)}\check{X}_{ij}^{k}-\check{X}_{ij}^{k}\check{X}_{mn}^{k}$
with $\check{X}_{ij}^{k}=X_{ij}^{k}\left(\check{\rho}^{k}\right)$.

\paragraph*{Estimation of $\Gamma^{k}\left(l\right)$}

For $h>0$, $\kappa^{k}\left(h\right)$ can be estimated by the empirical
autocovariance:
\begin{equation}
\hat{\kappa}'{}^{k}\left(h\right)=\frac{1}{M_{k}}\sum_{t=1}^{M_{k}-h}\left(\mathcal{V}\left(\Delta C_{t}^{k}\right)-\mathcal{V}\left(\check{X}^{k}\right)\right)\left(\mathcal{V}\left(\Delta C_{t+h}^{k}\right)-\mathcal{V}\left(\check{X}^{k}\right)\right)^{\mathrm{T}}\label{eq:kappa-nonzero-h}
\end{equation}
where $\check{X}^{k}$ is an estimate of $\E\left[\Delta C_{t}^{k}\right]=\bar{X}^{k}$.
Then the $\Gamma^{k}\left(l\right)$ can be estimated as
\begin{equation}
\hat{\Gamma}'{}^{k}\left(l\right)=\hat{\kappa}'{}^{k}\left(2l+1\right)+\hat{\kappa}'{}^{k}\left(2l+2\right)\label{eq:Gamma-count}
\end{equation}
However, the estimation error \eqref{eq:Gamma-count} will increase
substantially as $l$ approaches to $\left\lfloor \left(M_{k}-3\right)/2\right\rfloor $.
So here we modify the $\hat{\Gamma}'{}^{k}\left(l\right)$ by correcting
the corresponding estimated value of $\eta^{k}\left(l\right)$:

\begin{equation}
\hat{\Gamma}^{k}\left(l\right)=\left\{ \begin{array}{ll}
\hat{\Gamma}'{}^{k}\left(l\right), & l=0\\
\min\left\{ \frac{\hat{\eta}^{k}\left(l-1\right)}{\hat{\eta}'{}^{k}\left(l\right)},1\right\} \cdot\hat{\Gamma}'{}^{k}\left(l\right), & l>1\text{ and }\hat{\eta}'{}^{k}\left(l\right)>0\\
0, & l>1\text{ and }\hat{\eta}'{}^{k}\left(l\right)\le0
\end{array}\right.\label{eq:Gamma-hat}
\end{equation}
where $\hat{\eta}'{}^{k}\left(l\right)$ and $\hat{\eta}^{k}\left(l\right)$
denote the values of $\eta^{k}\left(l\right)$ obtained from $\hat{\Gamma}'{}^{k}\left(l\right)$
and $\hat{\Gamma}^{k}\left(l\right)$. It can be seen that $\hat{\eta}^{k}\left(l\right)$
is non-negative and decreasing with $l$, which is consistent with
the conclusion of Theorem \ref{thm:SigmaX-kappa}. Besides, we can
show that $\hat{\Gamma}^{k}\left(l\right)\equiv0$ for $l\ge l_{n}$
if there exists an $l_{n}$ such that $\hat{\eta}^{k}\left(l_{n}\right)\le0$.
Thus the estimator of $\Sigma_{X}^{k}$ in this section is in fact
a time window estimator \cite{geyer1992practical} where the large-lag
terms outside the window are set to be zero, and the window size $l_{w}=\min\left\{ l|\hat{\eta}^{k}\left(l\right)\le0\right\} $
implies that the curve of $\left\Vert \Gamma^{k}\left(l\right)\right\Vert _{\max}$
goes below the noise level at $l=l_{w}$.

\begin{remark}{\rm From the definition of $\Sigma_{X}^{k}$ we can
deduce that $\Sigma_{X}^{k}\ge0$ and $\mathbf{1}^{\mathrm{T}}\Sigma_{X}^{k}\mathbf{1}=0$,
but the $\hat{\Sigma}_{X}^{k}$ obtained by \eqref{eq:Sigma-X-estimate}
may not satisfy the constraints. For this problem, we can correct
the value of $\hat{\Sigma}_{X}^{k}$ as $\hat{\Sigma}_{X}^{k}:=\mathcal{M}_{1}\left(\mathcal{M}_{P}\left(\hat{\Sigma}_{X}^{k}\right)\right)$,
where $\mathcal{M}_{P}\left(G\right)=G-\min\left\{ \lambda_{\min}\left(G\right),0\right\} \mathbf{I}$
can map a symmetric matrix to a positive-semidefinite matrix with
$\lambda_{\min}\left(G\right)$ denoting the smallest eigenvalue of
$G$, and $\mathcal{M}_{1}\left(G\right)=\left(\mathbf{I}-\frac{1}{n}\mathbf{1}\mathbf{1}^{\mathrm{T}}\right)\cdot G\cdot\left(\mathbf{I}-\frac{1}{n}\mathbf{1}\mathbf{1}^{\mathrm{T}}\right)$
for $G\in\mathbb{R}^{n\times n}$ is a mapping from the positive-semidefinite
matrix set to the set $\left\{ \Sigma|\Sigma\ge0,\mathbf{1}^{\mathrm{T}}\Sigma\mathbf{1}=0\right\} $.}\end{remark}

\section{Numerical experiments\label{sec:Numerical-experiments}}

In this section, the approximate MLE proposed in this paper will be
applied to some numerical examples of multiple biased simulations,
and the performance will be compared to that of WHAM and MMMM. For
convenience, here we denote a set of multiple biased simulations described
in Section \ref{sec:Background} by $\mathrm{MBS}\left(K,M_{0}\right)$
if there are $K$ biased simulations and each simulation has the same
length with $M_{k}\equiv M_{0}$.

\subsection{Umbrella sampling with Markovian simulations\label{sub:Umbrella-sampling-Markov}}

Umbrella sampling is a commonly used biased simulation technique,
where each biasing potential (also called ``umbrella potential'')
is designed to confine the system around some region of state space
and achieve a more efficient sampling especially at transition states
which the unbiased simulation would visit only rarely. In this example,
the umbrella sampling simulations are employed on a reference system
with state set $\mathcal{S}=\left\{ s_{i}=-5+10\left(i-1\right)/99|i=1,\ldots,100\right\} $
and free energy $V=\left[V_{i}\right]=[0.25s_{i}^{4}-5s_{i}^{2}-9.9874]$.
As shown in Fig.~\ref{fig:projector}, the reference system has two
metastable states centered at $A$ and $B$, and the switching between
metastable states is blocked by an energy barrier with peak position
$O$.

\begin{figure}
\begin{centering}
\includegraphics[width=0.4\textwidth]{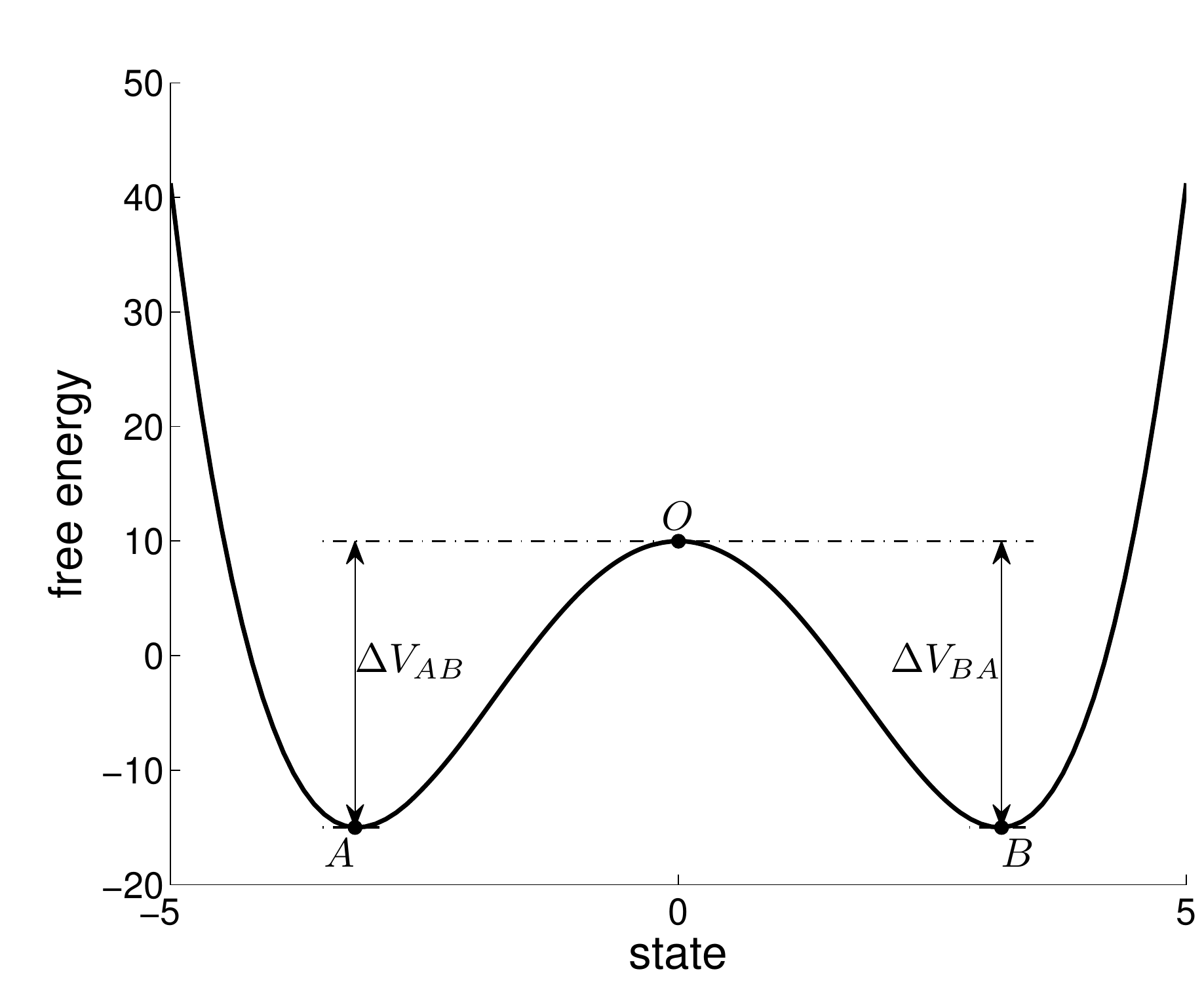}
\par\end{centering}

\caption{\label{fig:Free-energy-profile-Markov}Free energy profile of the
reference system, which has two potential wells with minima at $A$
and $B$ separated by an energy barrier. The highest energy position
$O$ of the barrier represents the transition state, and the energy
barrier heights for transitions $A\to B$ and $B\to A$ are defined
as $\Delta V_{AB}=\left|V_{O}-V_{A}\right|$ and $\Delta V_{BA}=\left|V_{O}-V_{B}\right|$
with $V_{A}$, $V_{B}$ and $V_{O}$ the potentials of $A$, $B$
and $O$. }
\end{figure}

For umbrella sampling simulations, we design the following $15$ different
biased potentials:
\begin{equation}
U^{k}=\left[U_{i}^{k}\right]=\left[4\left(s_{i}+\frac{15}{14}k-\frac{60}{7}\right)^{2}\right],\quad1\le k\le15
\end{equation}
Note that these potentials will be repeatedly used if the simulation
number is larger than $15$, i.e., $U^{k}=U^{\left(\left(k-1\right)\bmod15\right)+1}$
if $k>15$. The simulation trajectory $x_{0:M_{k}}^{k}$ is generated
by a Metropolis simulation model (see Appendix \ref{sec:Metropolis-sampling-model}
for details), which is a reversible Markov chain with initial distribution
\begin{equation}
\Pr\left(x_{0}^{k}=s_{i}\right)\propto\exp\left(-U_{i}^{k}\right)
\end{equation}
and stationary distribution
\begin{equation}
\pi_{i}^{k}\propto\exp\left(-V_{i}-U_{i}^{k}\right)
\end{equation}

The comparisons between the estimation methods are based on the mean
error of approximations of energy barrier heights:
\begin{equation}
e_{\Delta V}=\frac{1}{2}\left(\left|\Delta V_{AB}-\Delta V_{AB}^{\mathrm{approx}}\right|+\left|\Delta V_{BA}-\Delta V_{BA}^{\mathrm{approx}}\right|\right)
\end{equation}
where the definitions of $\Delta V_{AB}$ and $\Delta V_{BA}$ are
given in Fig.~\ref{fig:projector}, and the superscript ``$\mathrm{approx}$''
represents the approximate value obtained from the estimated $V$.

We first set $K=15$ and $M_{0}=500,910,1657,3017,5493,10000$, and
perform $30$ independent $\mathrm{MBS}\left(K,M_{0}\right)$ for
each value of $M_{0}$. Fig.~\ref{fig:M1-mse} displays the average
$e_{\Delta V}$ of the approximate MLE, MMMM and WHAM for different
$M_{0}$, and Figs.~\ref{fig:M1-shortest} and \ref{fig:M1-longest}
show the estimates of $V$ obtained from a run of $\mathrm{MBS}\left(15,500\right)$
and $\mathrm{MBS}\left(15,10000\right)$. It can be seen that the
estimation errors of all the three methods decrease with increasing
simulation length, and the proposed approximate MLE performs significantly
better than the other two methods. Note that MMMM is also a Markov
chain model based method, but its performance turns out to be worse
than WHAM in this numerical experiment, especially for large simulation
lengths $M_{0}$.

Next, we validate whether the estimation methods can reconstruct the
free energy $V$ from very short simulations. Here we fix the total
simulation time $M_{0}K$, and set $M_{0}=\left[22500/K\right]$ with
$K=45,90,135,180,225,270$. The estimation results are summarized
in Figs.~\ref{fig:M2-mse}, \ref{fig:M2-longest} and \ref{fig:M2-shortest}.
(The $1\sigma$ confidence intervals in Fig.~\ref{fig:M2-shortest}
are provided by using the sample standard deviation of $\check{V}$
calculated from the $30$ independent runs of $\mathrm{MBS}\left(270,83\right)$
because the simulation length is too short such that the error analysis
approach in Section \ref{sub:Error-analysis} is not applicable.)
It should be noticed the equilibrium assumption used by WHAM does
not hold if $M_{0}$ is too small, because the initial distributions
of simulations differ from the biased stationary distributions (see
Appendix \ref{sec:Metropolis-sampling-model}). Therefore the estimation
accuracy of WHAM is reduced when the individual simulation lengths
are shorter although the total data size stays almost the same. In
contrast, the proposed approximate MLE and MMMM are less affected
by the change in the length of individual simulations. This is because
these methods rely on having local rather then global equilibrium
assumptions. Furthermore, the proposed method outperforms both WHAM
and MMMM in this numerical experiment.

\begin{figure}
\subfloat[\label{fig:M1-mse}Average $e_{\Delta V}$ calculated over $30$ independent
runs of $\mathrm{MBS}\left(K,M_{0}\right)$ for $K=15$ and $M_{0}=500,910,1657,3017,5493,10000$.]{\includegraphics[width=0.45\textwidth]{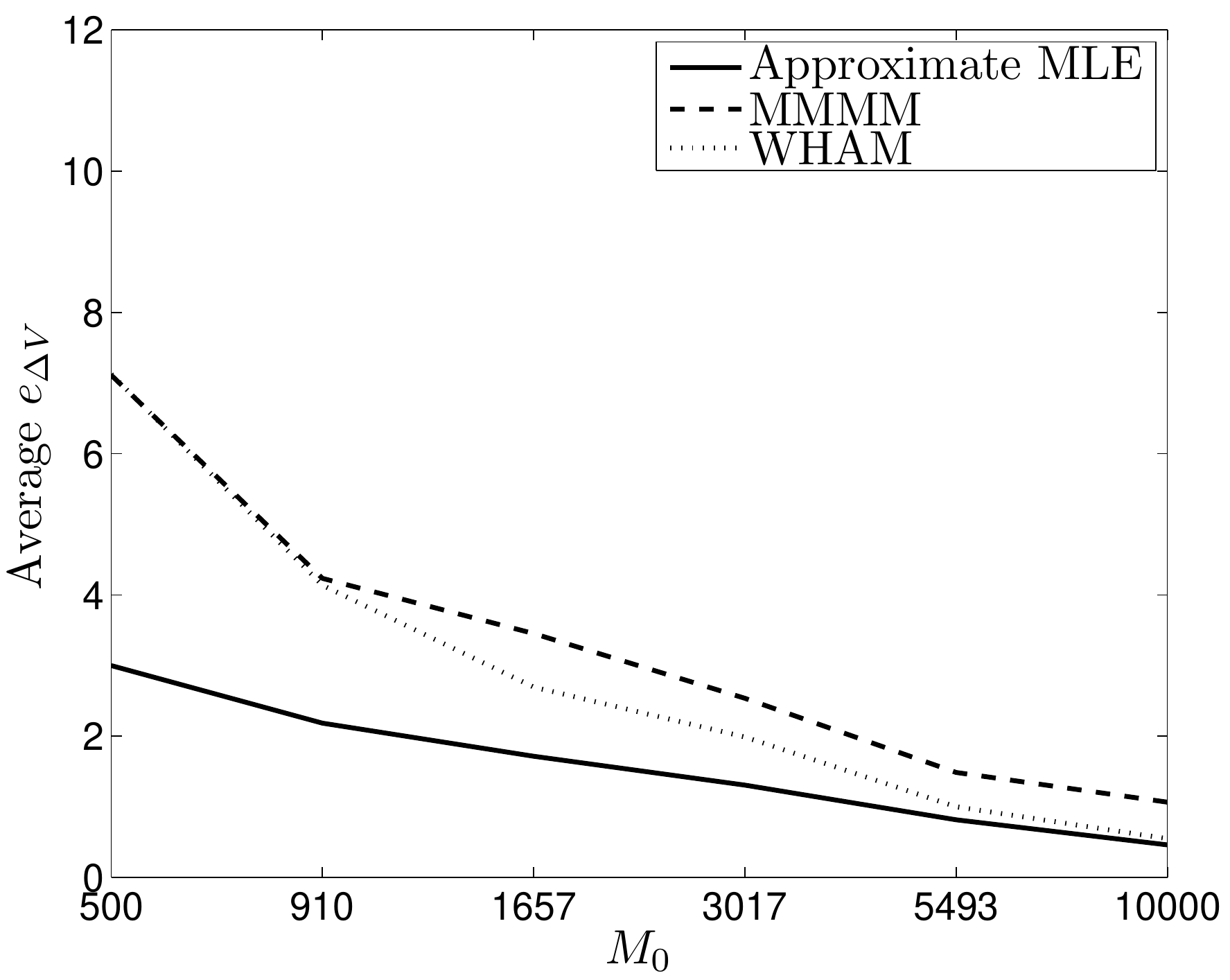}}\hfill{}\subfloat[{\label{fig:M2-mse}Average $e_{\Delta V}$ calculated over $30$ independent
runs of $\mathrm{MBS}\left(K,M_{0}\right)$ for $K=45,90,135,180,225,270$
and $M_{0}=\left[22500/K\right]$.}]{\includegraphics[width=0.45\textwidth]{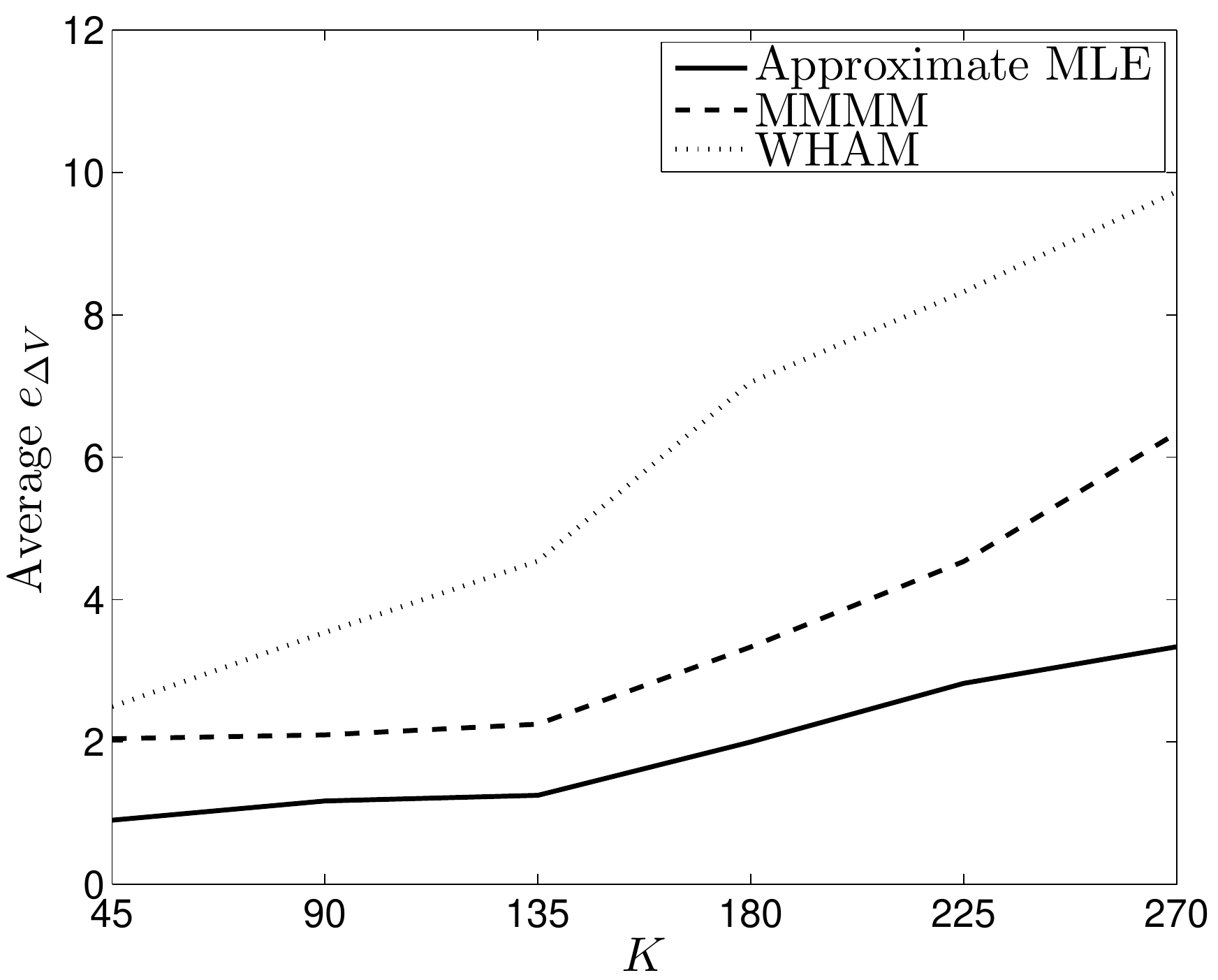}}

\subfloat[\label{fig:M1-shortest}Estimates of $V$ generated by the different
estimators on a run of $\mathrm{MBS}\left(15,500\right)$ where the
$e_{\Delta V}$ of approximate MLE $=3.5015$, $e_{\Delta V}$ of
MMMM $=6.3153$ and $e_{\Delta V}$ of WHAM $=6.3500$.]{\includegraphics[width=0.45\textwidth]{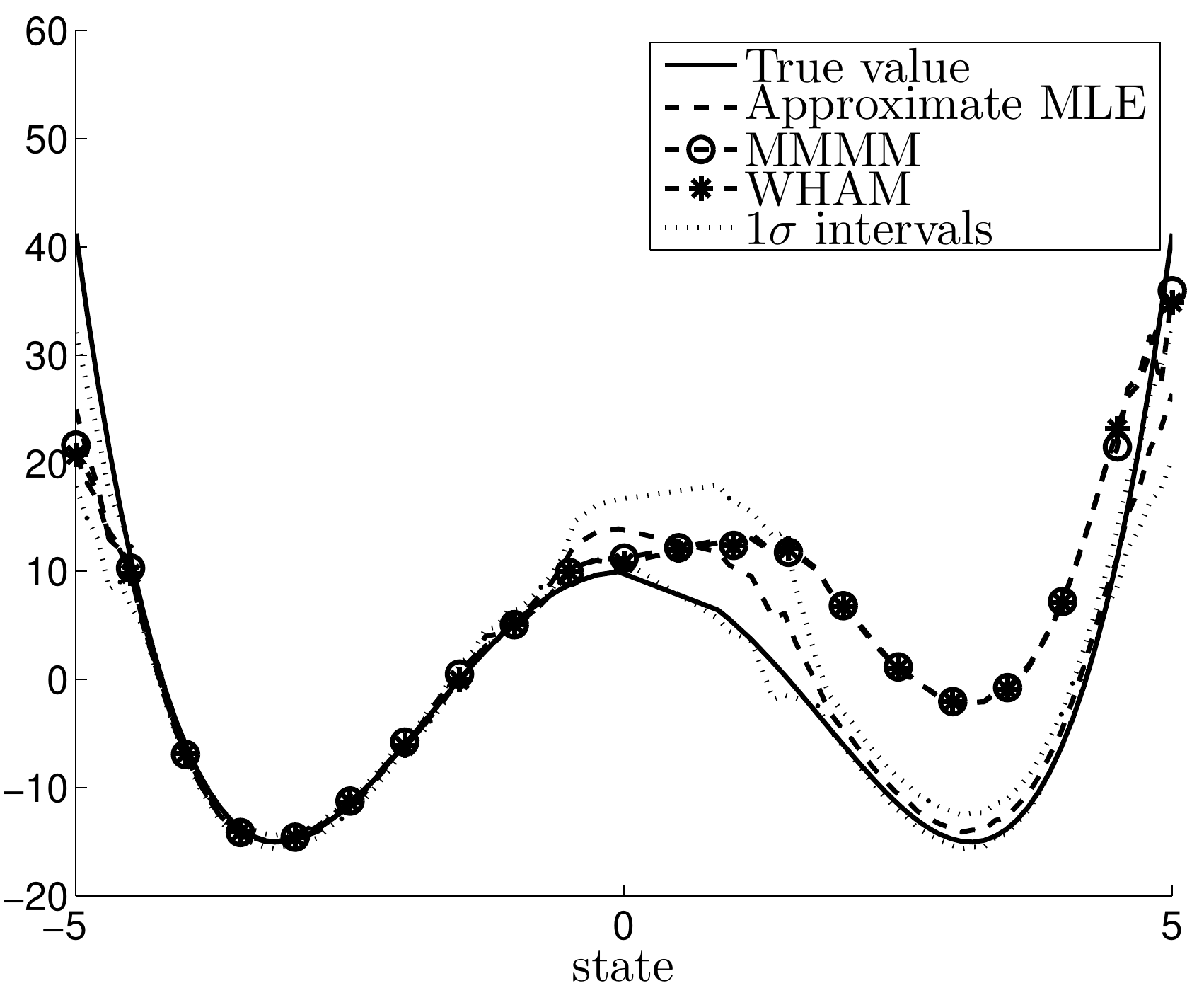}}\hfill{}\subfloat[\label{fig:M1-longest}Estimates of $V$ generated by the different
estimators on a run of $\mathrm{MBS}\left(15,10000\right)$ where
the $e_{\Delta V}$ of approximate MLE $=0.3060$, $e_{\Delta V}$
of MMMM $=0.6002$ and $e_{\Delta V}$ of WHAM $=0.3397$.]{\includegraphics[width=0.45\textwidth]{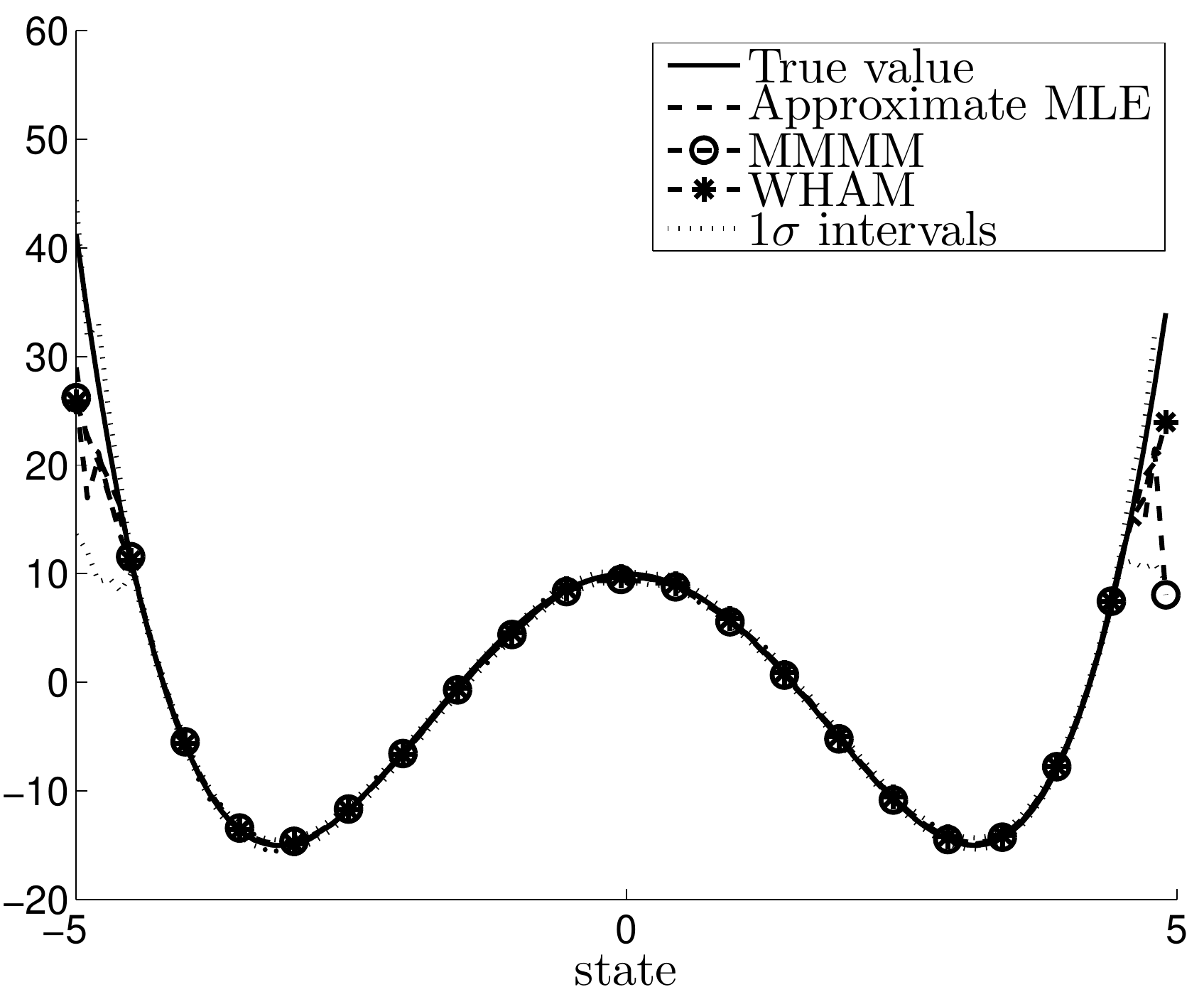}}

\subfloat[\label{fig:M2-longest}Estimates of $V$ generated by the different
estimators on a run of $\mathrm{MBS}\left(45,500\right)$ where the
$e_{\Delta V}$ of approximate MLE $=0.4570$, $e_{\Delta V}$ of
MMMM $=1.5012$ and $e_{\Delta V}$ of WHAM $=1.8948$.]{\includegraphics[width=0.45\textwidth]{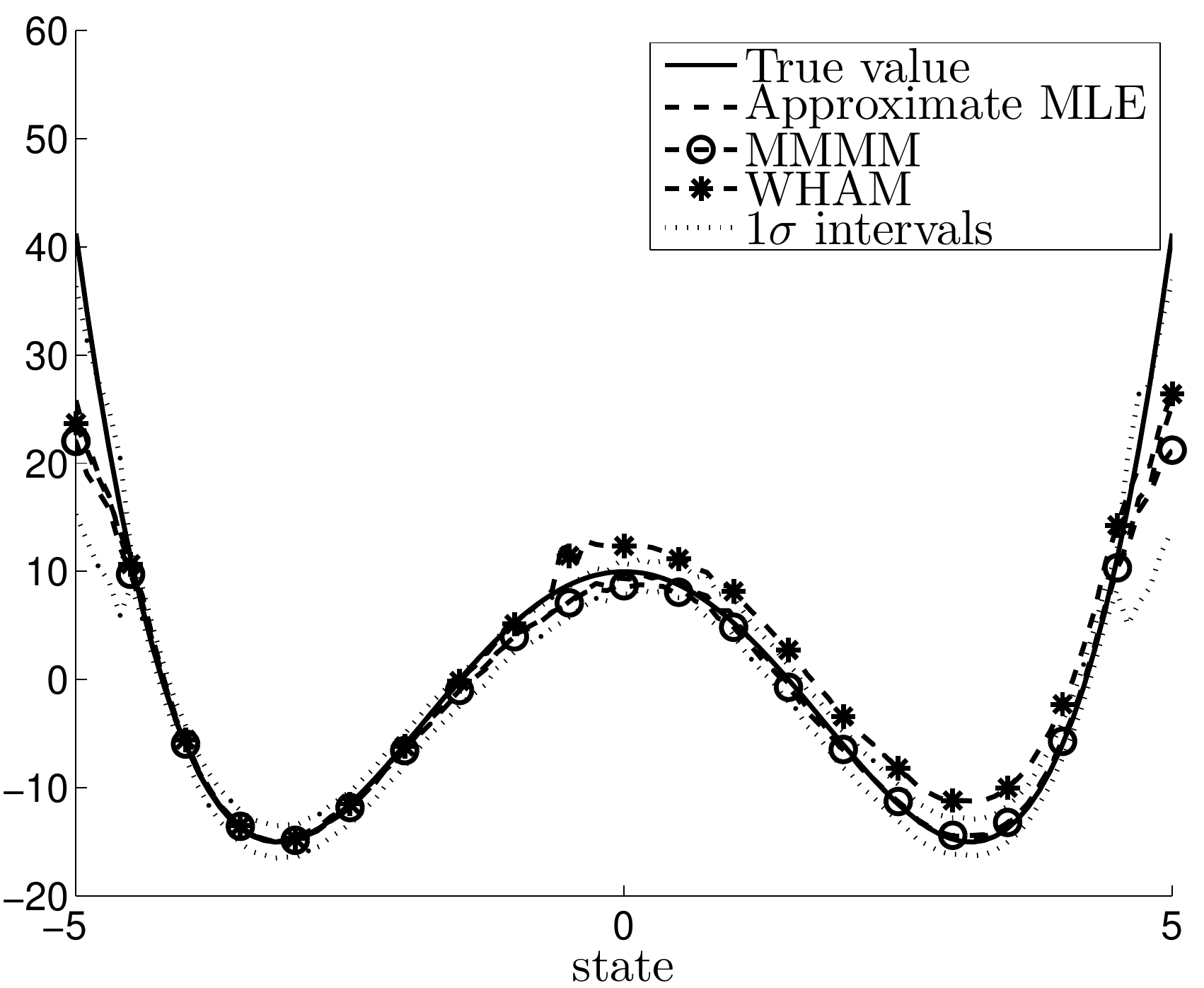}}\hfill{}\subfloat[\label{fig:M2-shortest}Estimates of $V$ generated by the different
estimators on a run of $\mathrm{MBS}\left(270,83\right)$ where the
$e_{\Delta V}$ of approximate MLE $=3.1939$, $e_{\Delta V}$ of
MMMM $=4.2559$ and $e_{\Delta V}$ of WHAM $=7.2764$.]{\includegraphics[width=0.45\textwidth]{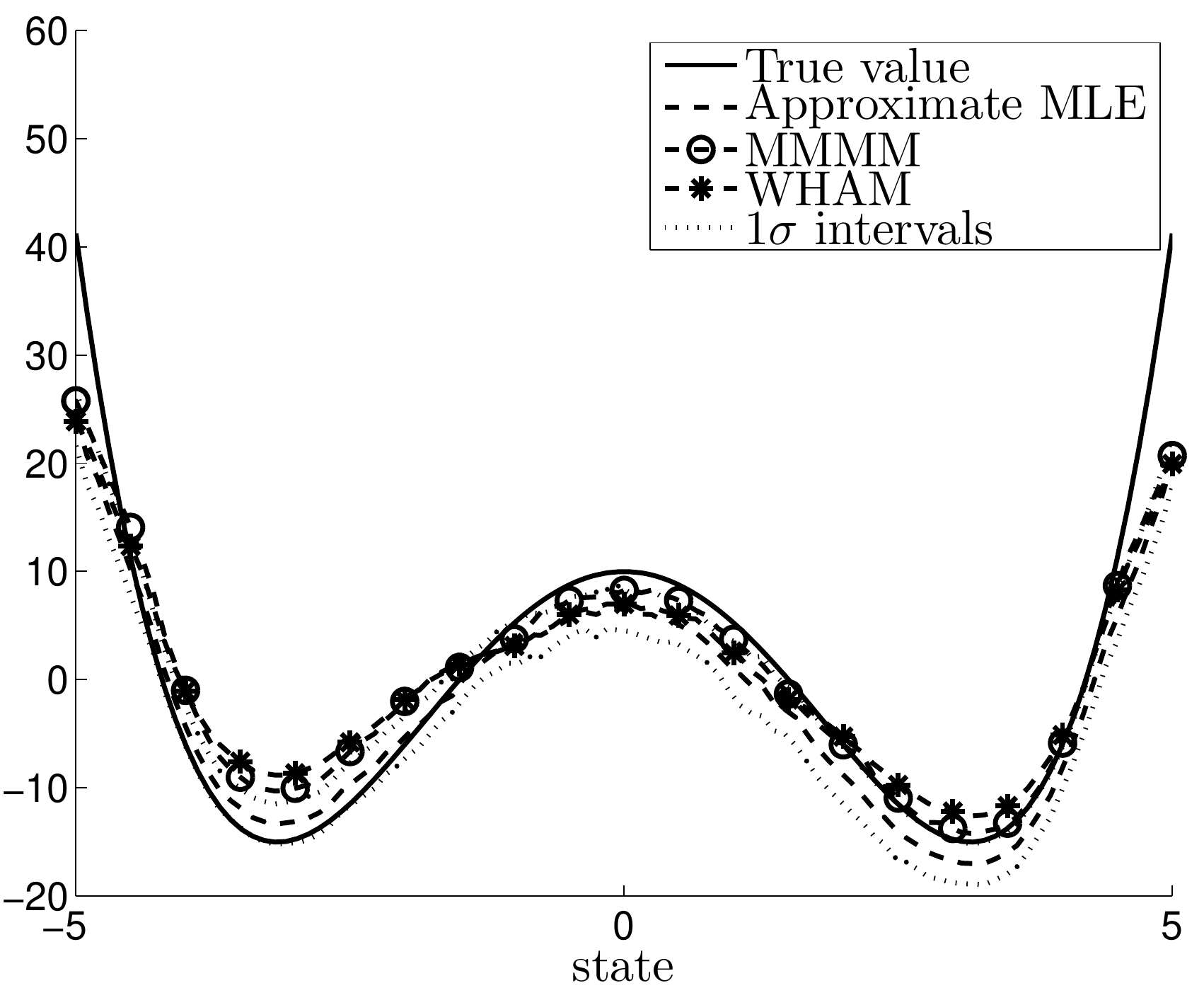}}

\caption{Estimation results of umbrella sampling with Markovian simulations.
The $1\sigma$ confidence intervals in (c), (d) and (e) are obtained
by the approach described in Section \ref{sub:Error-analysis}, and
those in (f) are obtained from the sample standard deviation of $\check{V}$
in the $30$ independent runs.}
\end{figure}

\subsection{Umbrella sampling with non-Markovian simulations\label{sub:Umbrella-sampling-nonMarkov}}

We now consider the estimation problem from an umbrella sampling simulation
in the case that the Markov assumption does not hold, i.e., the bins
used to estimate the free energy do not correspond to the Markov states
of the underlying simulation. The simulation model and the other settings
in this section is basically the same as in Section \ref{sub:Umbrella-sampling-Markov}
except that the state set is defined as $\mathcal{S}=\left\{ \bar{s}_{1},\ldots,\bar{s}_{10}\right\} $
with $\bar{s}_{1}=\left\{ s_{1},\ldots,s_{10}\right\} $, $\bar{s}_{2}=\left\{ s_{11},\ldots,s_{15}\right\} $,
$\bar{s}_{3}=\left\{ s_{16},\ldots,s_{20}\right\} $ \ldots{} $\bar{s}_{17}=\left\{ s_{86},\ldots,s_{90}\right\} $
and $\bar{s}_{18}=\left\{ s_{91},\ldots,s_{100}\right\} $. It is
clear that the observed state sequences in simulations do not satisfy
the Markov property with this definition of states.

We utilize the three methods to approximate the free energy $V$ by
using the non-Markovian simulation data, and the estimation results
with different $\left(K,M_{0}\right)$ are shown in Fig.~\ref{fig:N1N2-mse},
where $e_{\Delta V}$ is defined in the same way as in Section \ref{sub:Umbrella-sampling-Markov}
with $A,B$ and $O$ the local minimum and peak positions in $\mathcal{S}$.
As observed from the figures, the estimates obtained from the approximate
MLE are more precise than those obtained from the other estimators
for various values of $\left(K,M_{0}\right)$.

\begin{figure}
\subfloat[\label{fig:N1-mse}$e_{\Delta V}$ calculated over $30$ independent
runs of $\mathrm{MBS}\left(K,M_{0}\right)$ for $K=15$ and $M_{0}=500,910,1657,3017,5493,10000$.]{\includegraphics[width=0.45\textwidth]{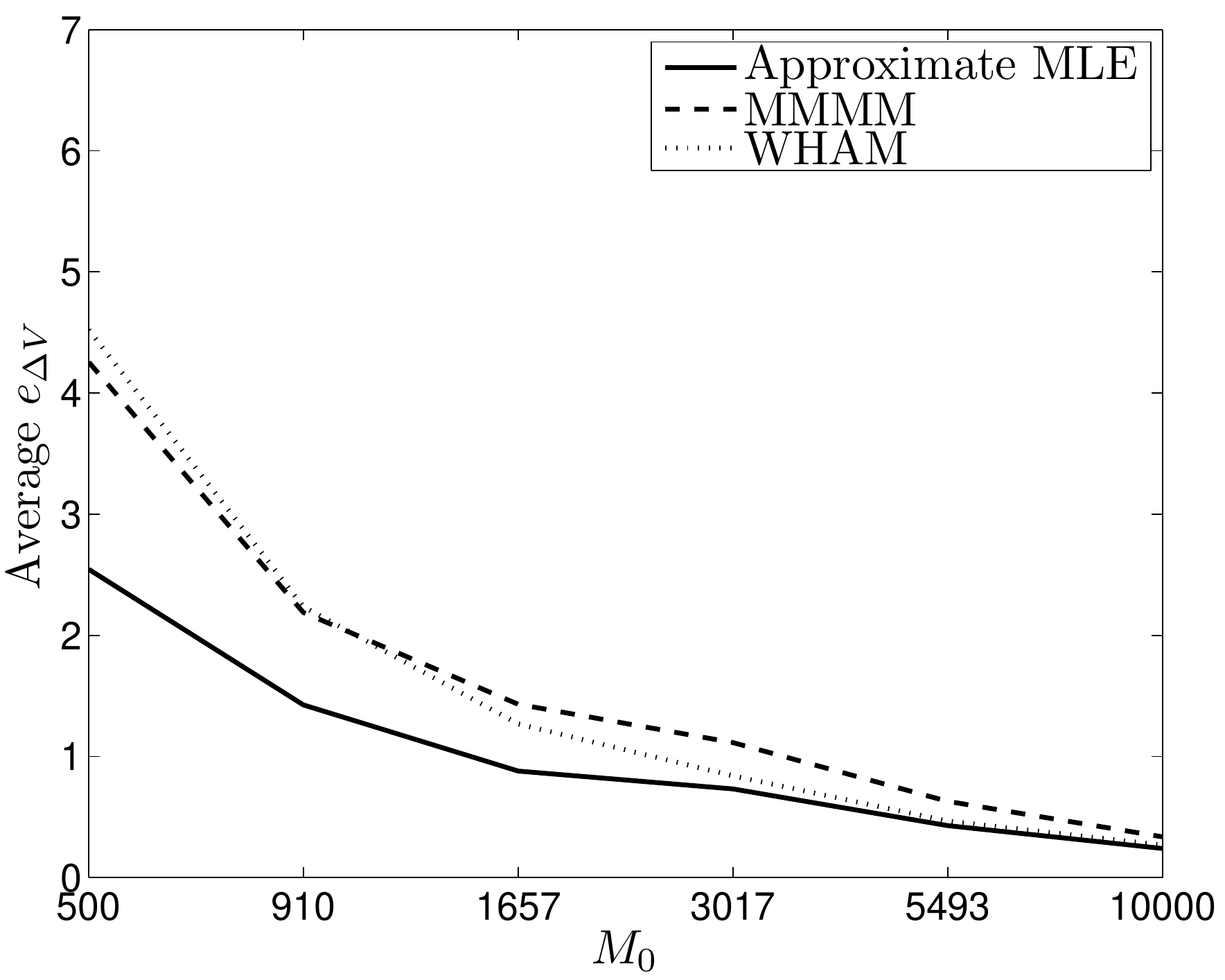}}\hfill{}\subfloat[{\label{fig:N2-mse}Average $e_{\Delta V}$ calculated over $30$ independent
runs of $\mathrm{MBS}\left(K,M_{0}\right)$ for $K=45,90,135,180,225,270$
and $M_{0}=\left[22500/K\right]$.}]{\includegraphics[width=0.45\textwidth]{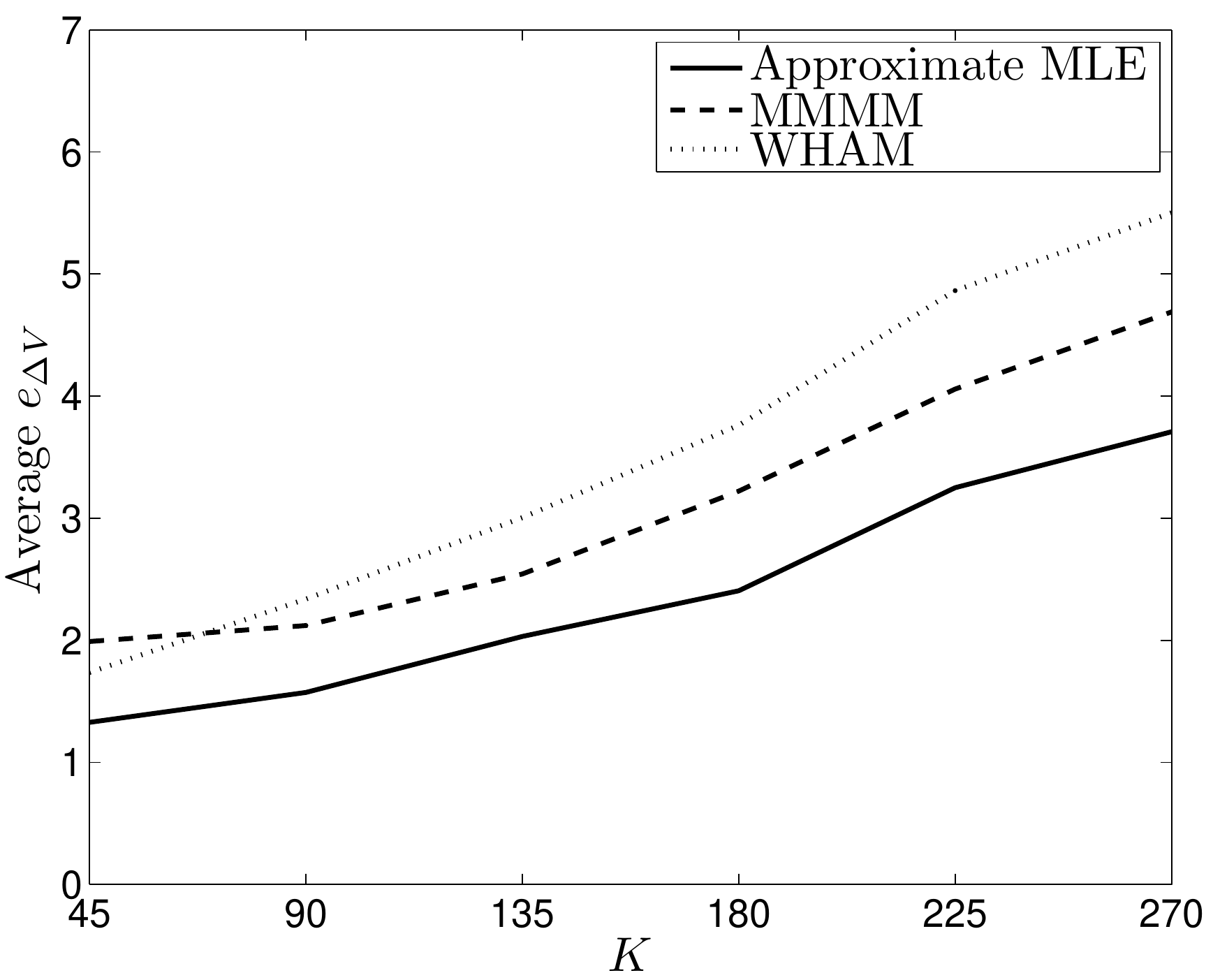}}

\subfloat[\label{fig:N1-shortest}Estimates of $V$ generated by the different
estimators on a run of $\mathrm{MBS}\left(15,500\right)$ where the
$e_{\Delta V}$ of approximate MLE $=2.4208$, $e_{\Delta V}$ of
MMMM $=6.2008$ and $e_{\Delta V}$ of WHAM $=6.1413$.]{\includegraphics[width=0.45\textwidth]{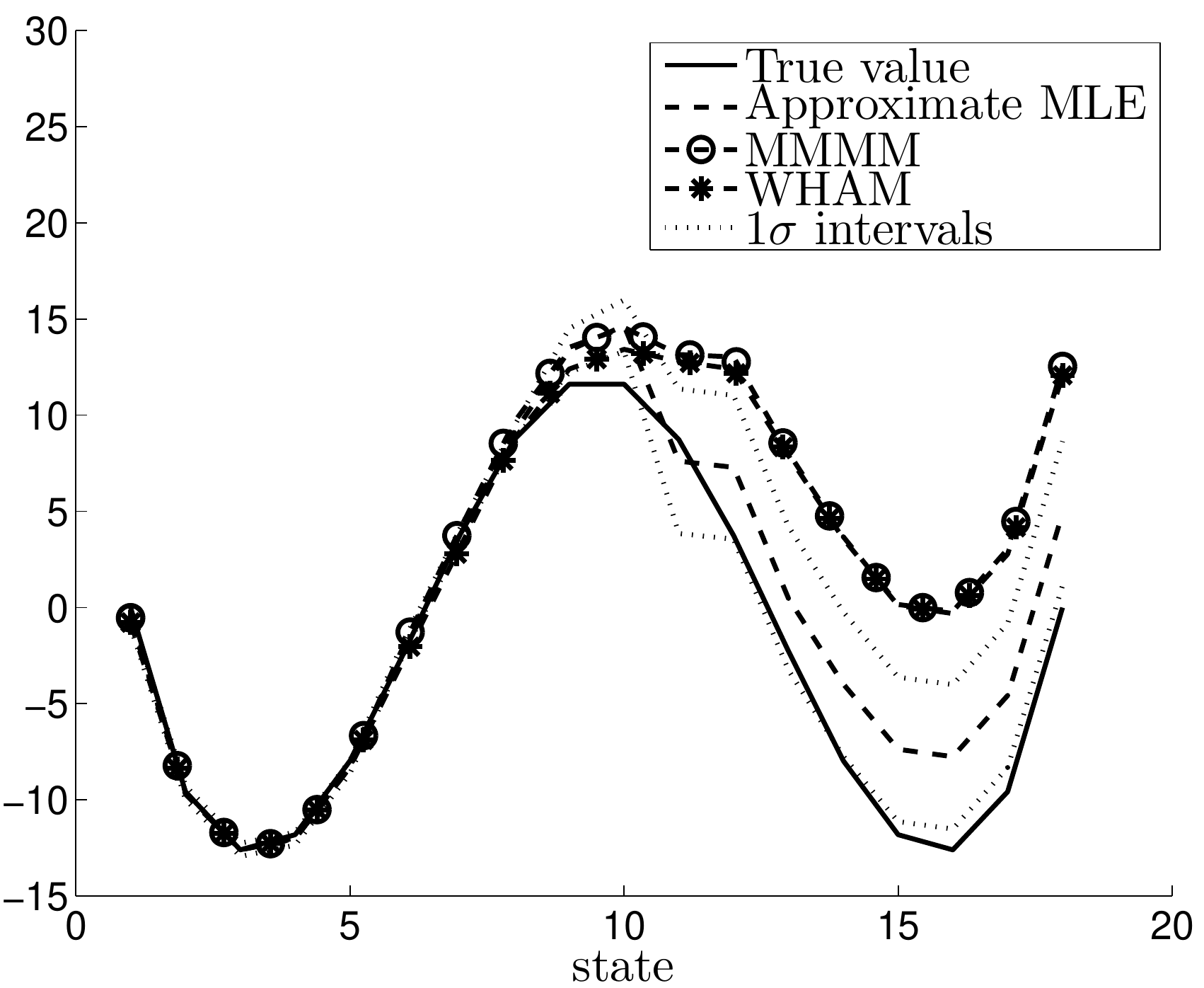}}\hfill{}\subfloat[\label{fig:N1-longest}Estimates of $V$ generated by the different
estimators on a run of $\mathrm{MBS}\left(15,10000\right)$ where
the $e_{\Delta V}$ of approximate MLE $=0.4532$, $e_{\Delta V}$
of MMMM $=0.6289$ and $e_{\Delta V}$ of WHAM $=0.4711$.]{\includegraphics[width=0.45\textwidth]{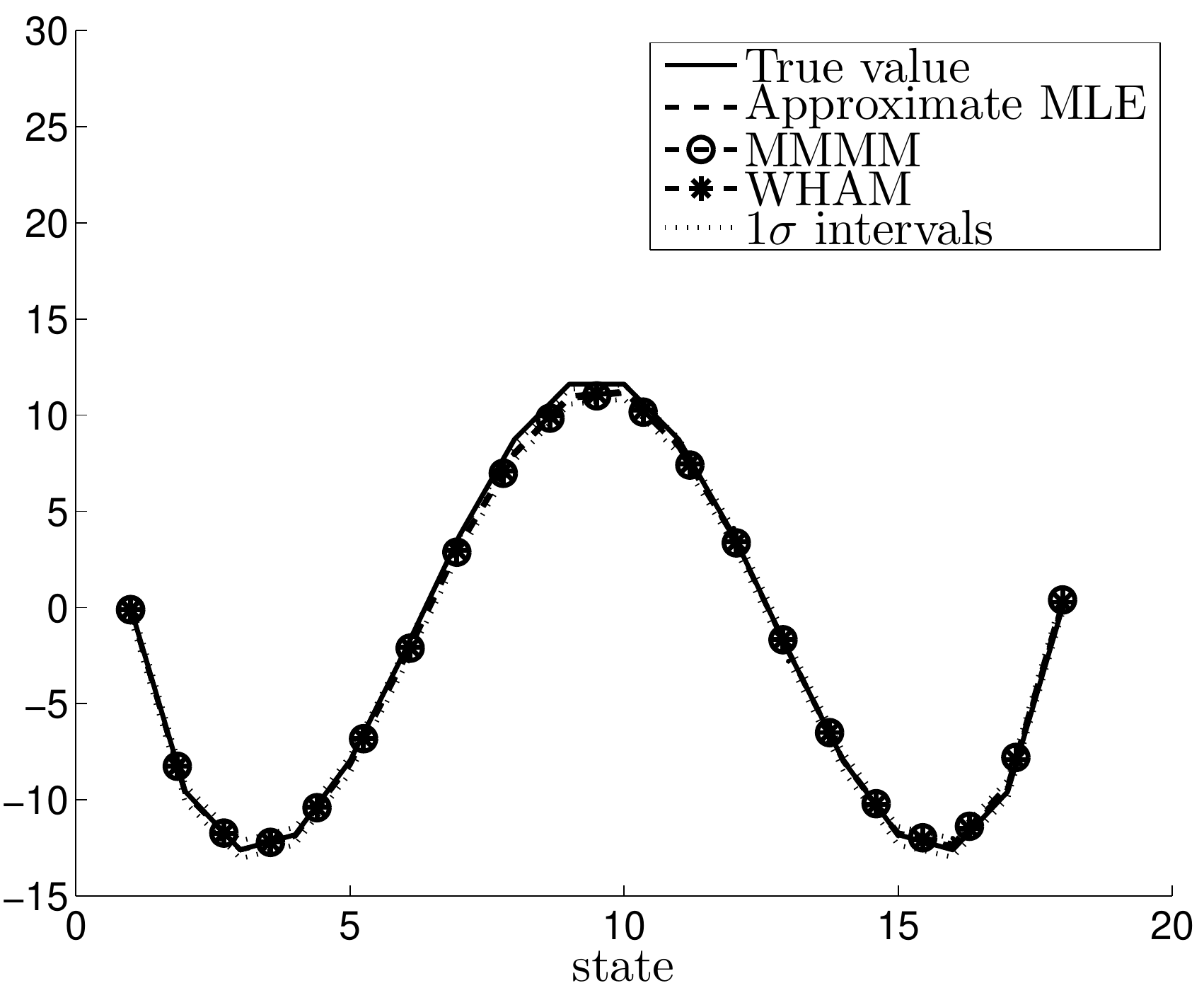}}

\subfloat[\label{fig:N2-longest}Estimates of $V$ generated by the different
estimators on a run of $\mathrm{MBS}\left(45,500\right)$ where the
$e_{\Delta V}$ of approximate MLE $=0.4989$, $e_{\Delta V}$ of
MMMM $=1.6107$ and $e_{\Delta V}$ of WHAM $=0.6257$.]{\includegraphics[width=0.45\textwidth]{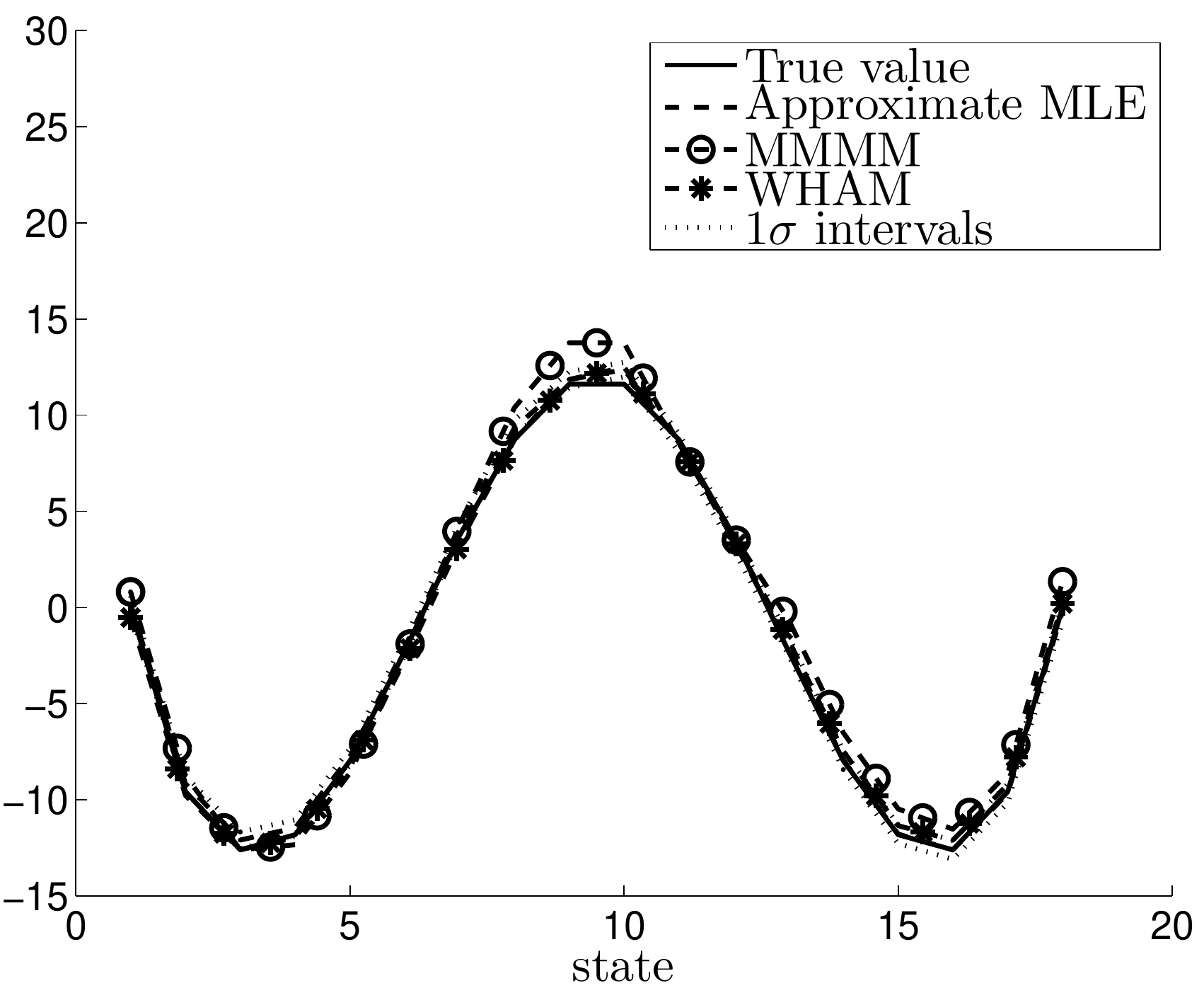}}\hfill{}\subfloat[\label{fig:N2-shortest}Estimates of $V$ generated by the different
estimators on a run of $\mathrm{MBS}\left(270,83\right)$ where the
$e_{\Delta V}$ of approximate MLE $=2.5862$, $e_{\Delta V}$ of
MMMM $=3.6059$ and $e_{\Delta V}$ of WHAM $=4.7896$.]{\includegraphics[width=0.45\textwidth]{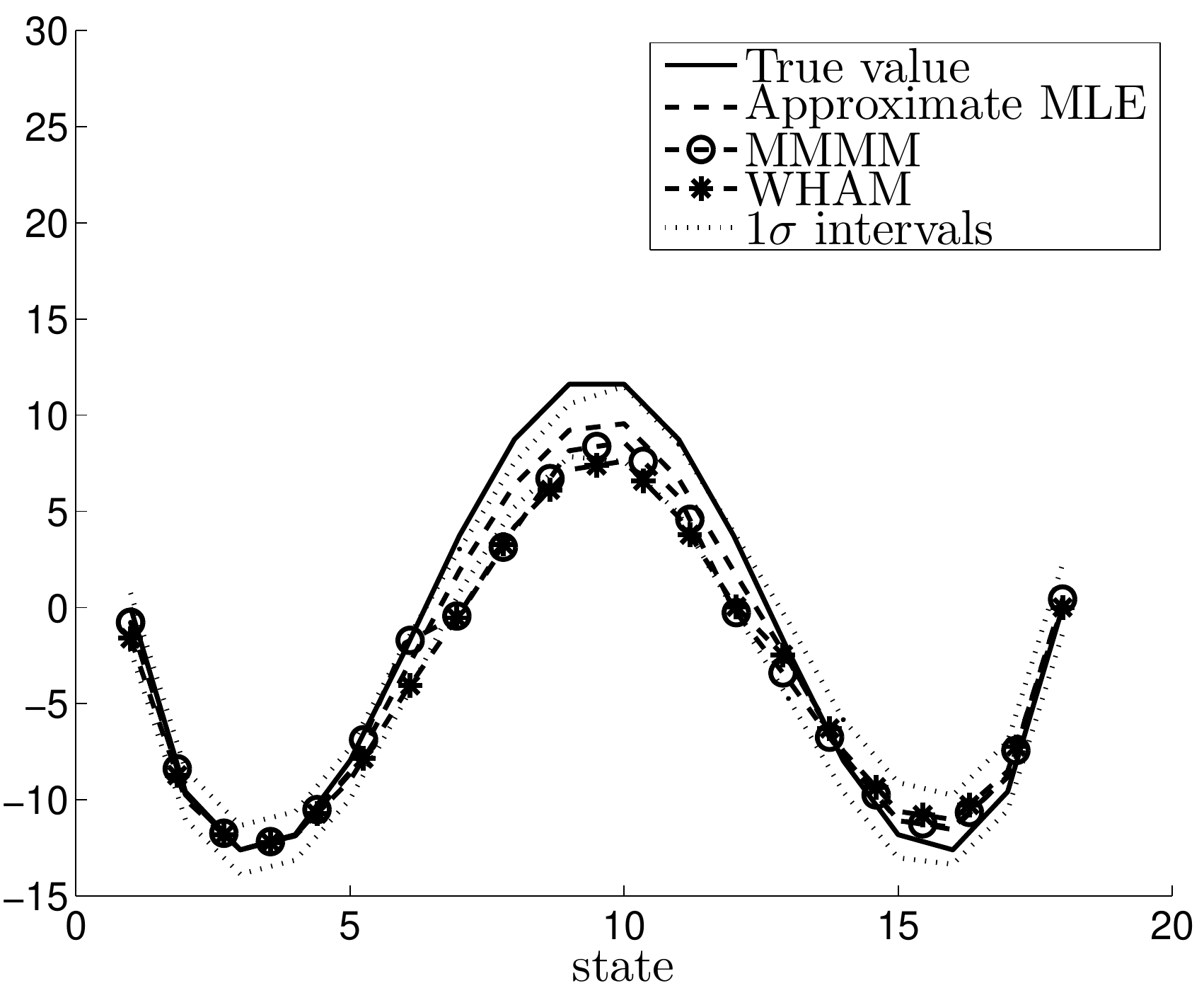}}

\caption{\label{fig:N1N2-mse}Estimation results of umbrella sampling with
non-Markovian simulations. The $1\sigma$ confidence intervals in
(c), (d) and (e) are obtained by the approach described in Section
\ref{sub:Error-analysis}, and those in (f) are obtained from the
sample standard deviation of $\check{V}$ in the $30$ independent
runs.}
\end{figure}

\subsection{Metadynamics with Markovian simulations\label{sub:Metadynamics-Markov}}

Metadynamics is another biased simulation technique often employed
in computational physics and chemistry, which is able to escape local
free energy minima and improve the searching properties of simulations
through iteratively modifying the biasing potential. Given $K,M_{0}$
and a reference system as in Section \ref{sub:Umbrella-sampling-Markov},
a metadynamics procedure can also be expressed as a run of $\mathrm{MBS}\left(K,M_{0}\right)$
with
\begin{equation}
U_{i}^{k}=\left\{ \begin{array}{ll}
0, & k=1\\
U_{i}^{k-1}+u_{c}\left(s_{i}|x_{M_{0}}^{k-1}\right), & k>1
\end{array}\right.,\quad\text{for }k=1,\ldots,K
\end{equation}
where $u_{c}\left(s|x\right)$ denote a Gaussian function of $s$
centered at $x$. Thus, for each of the $K$ simulations in a $\mathrm{MBS}$
run, a Gaussian hat is added to the potential at the last point of
the previous simulation. This effectively fills up the potential energy
basins with increasing $k$. Ultimately the effective potential becomes
approximately flat. Here we define $u_{c}\left(s|x\right)=5\exp\left(-\left(s-x\right)^{2}\right)$,
and the simulation data $x_{0:M_{0}}^{k}$ is also generated by the
Metropolis sampling model with $x_{0}^{k}=x_{M_{0}}^{k-1}$.

The three estimation methods are applied to reconstruct the free energy
of the reference system by data generated by metadynamics with different
$\left(K,M_{0}\right)$, and the estimation results are shown in Fig.~\ref{fig:M3M4-mse}.
The superior performance of the presented method is clearly evident
from the figures.

\begin{figure}
\subfloat[Average $e_{\Delta V}$ calculated over $30$ independent runs of
$\mathrm{MBS}\left(K,M_{0}\right)$ for $K=15$ and $M_{0}=500,910,1657,3017,5493,10000$.]{\includegraphics[width=0.45\textwidth]{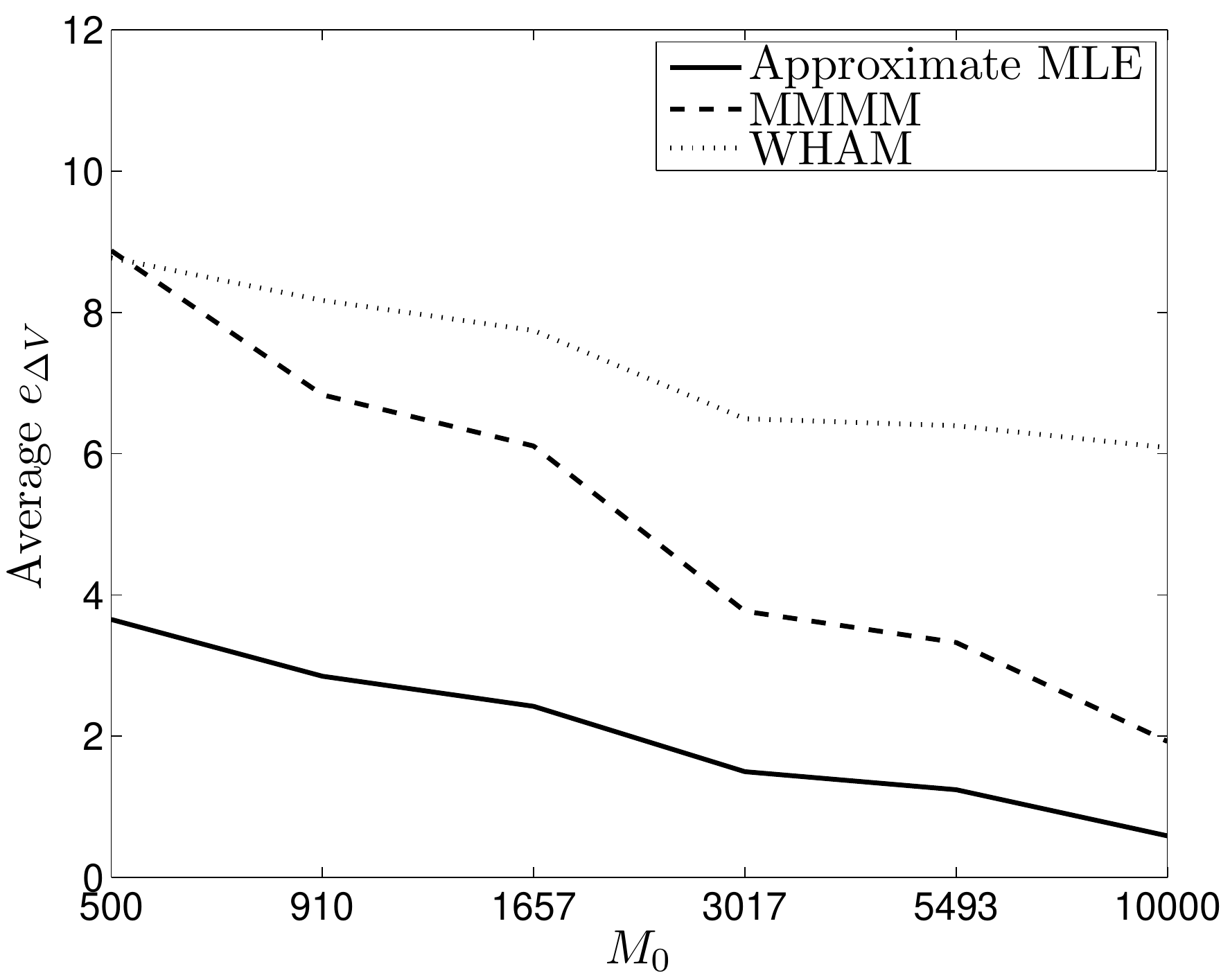}}\hfill{}\subfloat[{Average $e_{\Delta V}$ calculated over $30$ independent runs of
$\mathrm{MBS}\left(K,M_{0}\right)$ for $K=40,80,120,160,200$ and
$M_{0}=\left[2000/K\right]$.}]{\includegraphics[width=0.45\textwidth]{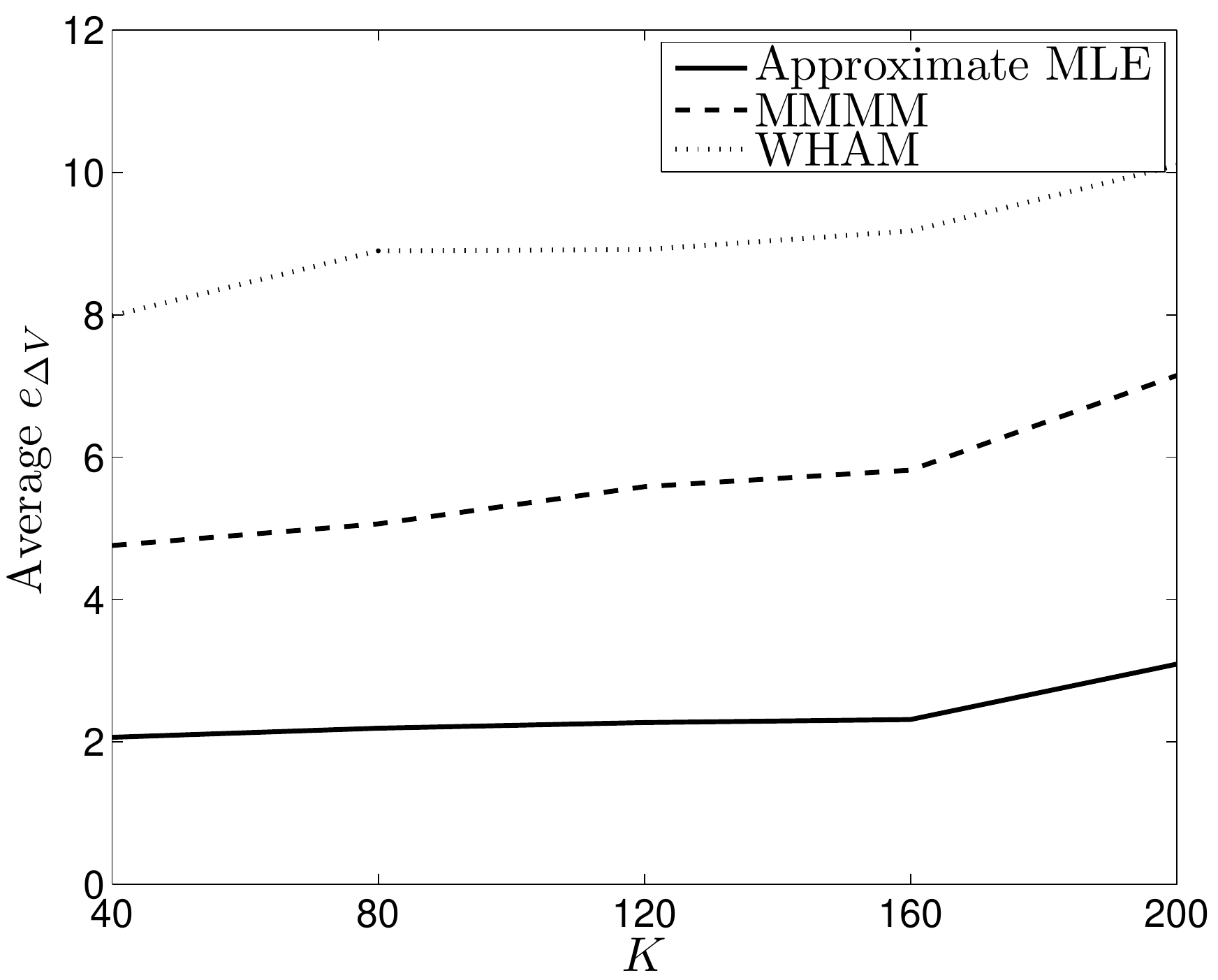}}

\subfloat[Estimates of $V$ generated by the different estimators on a run of
$\mathrm{MBS}\left(15,500\right)$ where the $e_{\Delta V}$ of approximate
MLE $=6.2858$, $e_{\Delta V}$ of MMMM $=11.9022$ and $e_{\Delta V}$
of WHAM $=12.4188$.]{\includegraphics[width=0.45\textwidth]{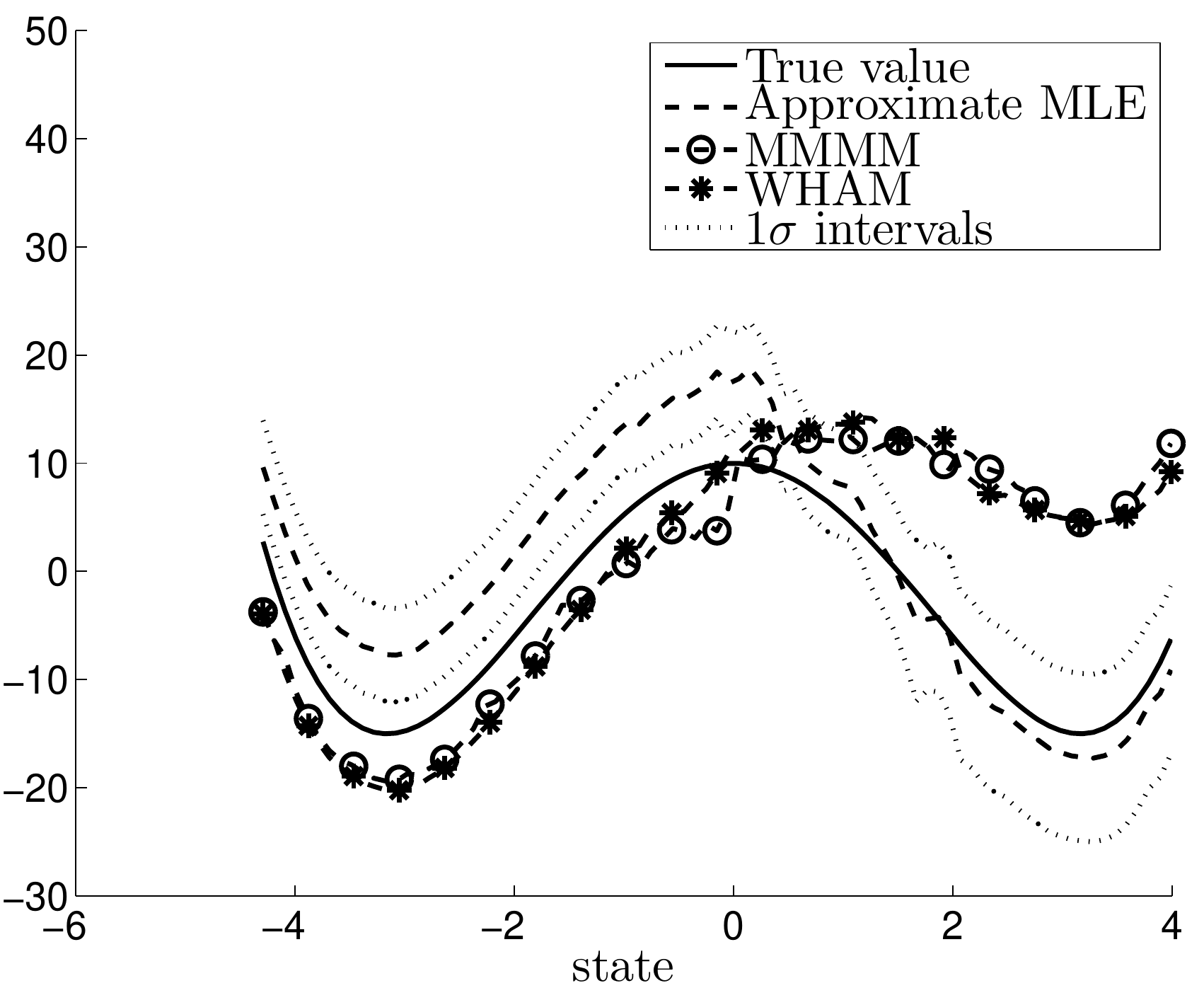}}\hfill{}\subfloat[Estimates of $V$ generated by the different estimators on a run of
$\mathrm{MBS}\left(15,10000\right)$ where the $e_{\Delta V}$ of
approximate MLE $=0.8246$, $e_{\Delta V}$ of MMMM $=3.1642$ and
$e_{\Delta V}$ of WHAM $=5.7601$.]{\includegraphics[width=0.45\textwidth]{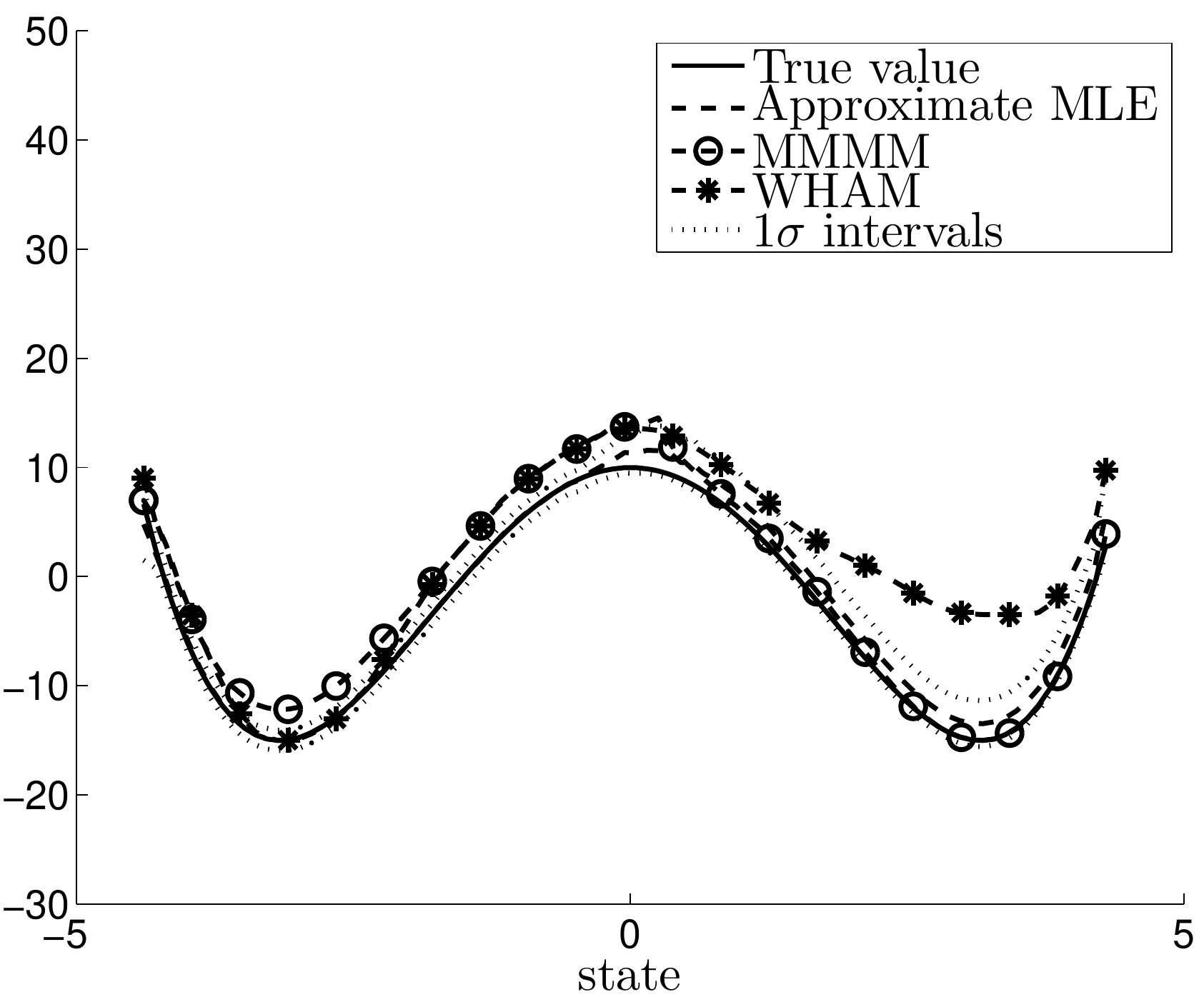}}

\subfloat[Estimates of $V$ generated by the different estimators on a run of
$\mathrm{MBS}\left(40,50\right)$ where the $e_{\Delta V}$ of approximate
MLE $=1.6589$, $e_{\Delta V}$ of MMMM $=3.6707$ and $e_{\Delta V}$
of WHAM $=7.5342$.]{\includegraphics[width=0.45\textwidth]{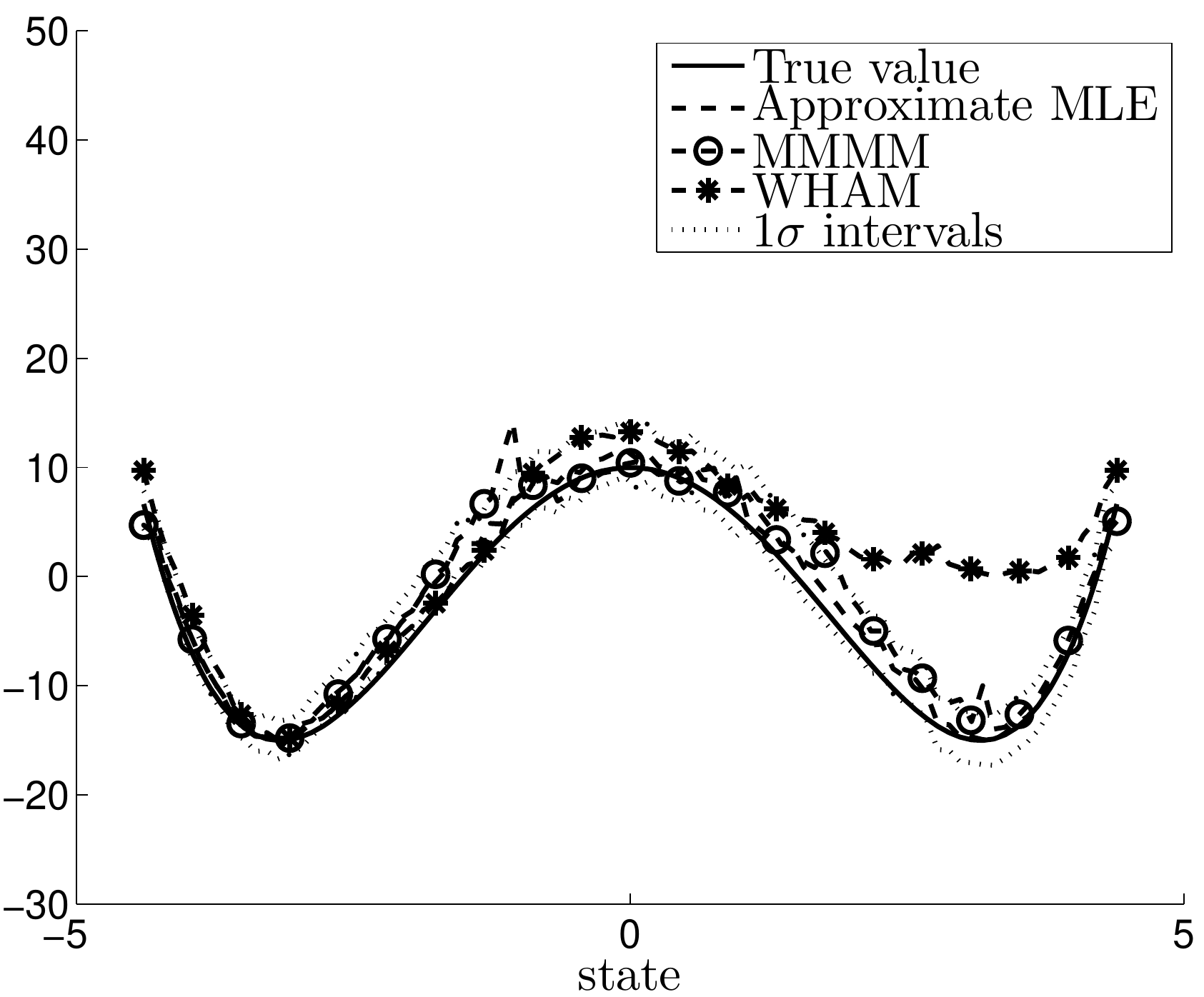}}\hfill{}\subfloat[Estimates of $V$ generated by the different estimators on a run of
$\mathrm{MBS}\left(200,10\right)$ where the $e_{\Delta V}$ of approximate
MLE $=2.8518$, $e_{\Delta V}$ of MMMM $=6.5913$ and $e_{\Delta V}$
of WHAM $=10.9540$.]{\includegraphics[width=0.45\textwidth]{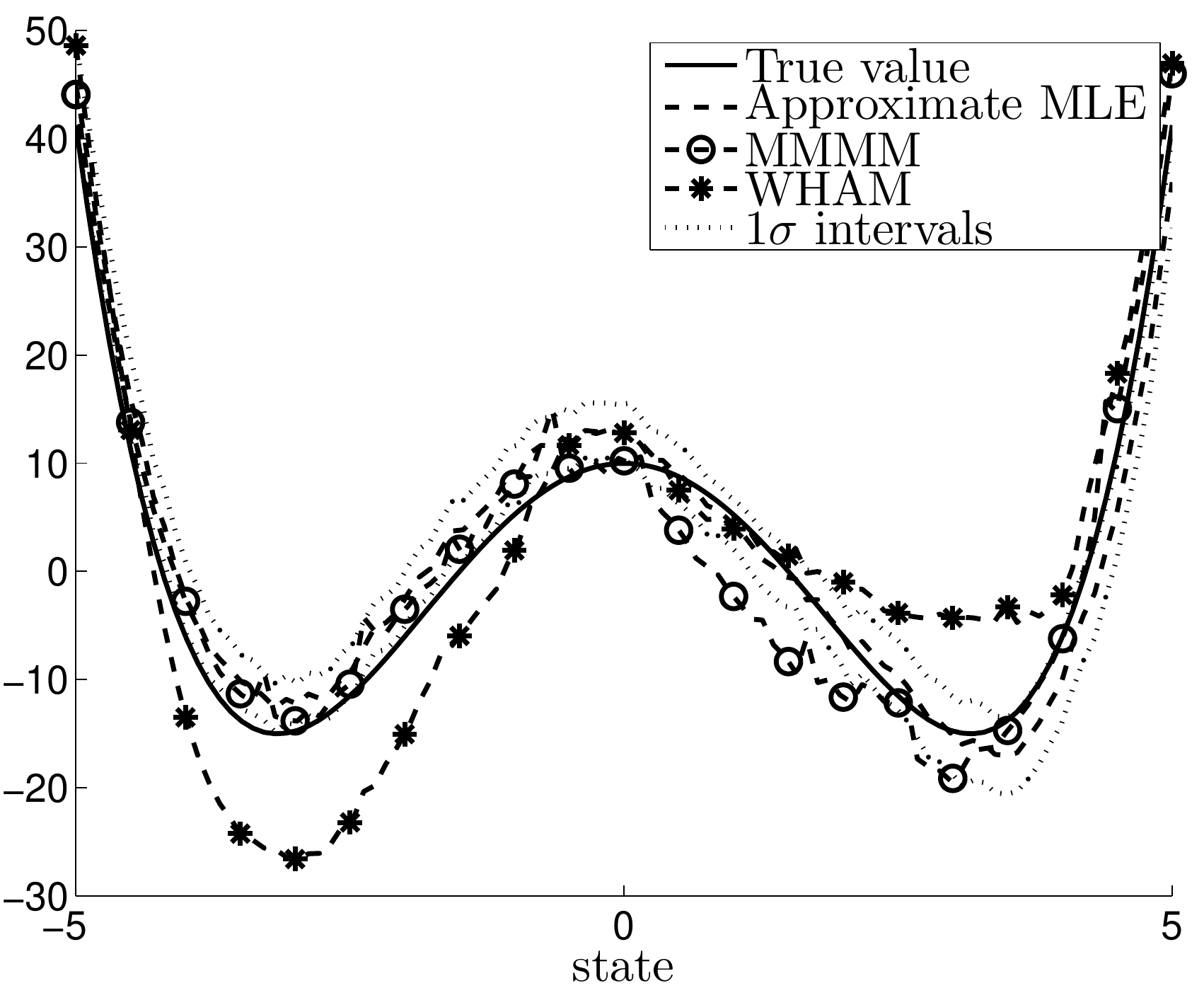}}

\caption{\label{fig:M3M4-mse}Estimation results of metadynamics with Markovian
simulations. The $1\sigma$ confidence intervals in (c) and (d) are
obtained by the approach described in Section \ref{sub:Error-analysis},
and those in (e) and (f) are obtained from the sample standard deviation
of $\check{V}$ in the $30$ independent runs.}
\end{figure}

\subsection{Metadynamics with non-Markovian simulations}

In this example, the free energy estimation problem of metadynamics
with non-Markovian simulations is investigated. We generate the simulation
data as in Section \ref{sub:Metadynamics-Markov} and convert the
state sequences to non-Markov processes as in Section \ref{sub:Umbrella-sampling-nonMarkov}.
Then the three methods can be used to estimate the unbiased free energy
of states $\bar{s}_{1},\ldots,\bar{s}_{18}$.

All the estimation results are displayed in Fig.~\ref{fig:N3N4-mse}.
It is obvious that the approximate MLE does a much better job in the
free energy estimation than the other two methods.

\begin{figure}
\subfloat[Average $e_{\Delta V}$ calculated over $30$ independent runs of
$\mathrm{MBS}\left(K,M_{0}\right)$ for $K=15$ and $M_{0}=500,910,1657,3017,5493,10000$.]{\includegraphics[width=0.45\textwidth]{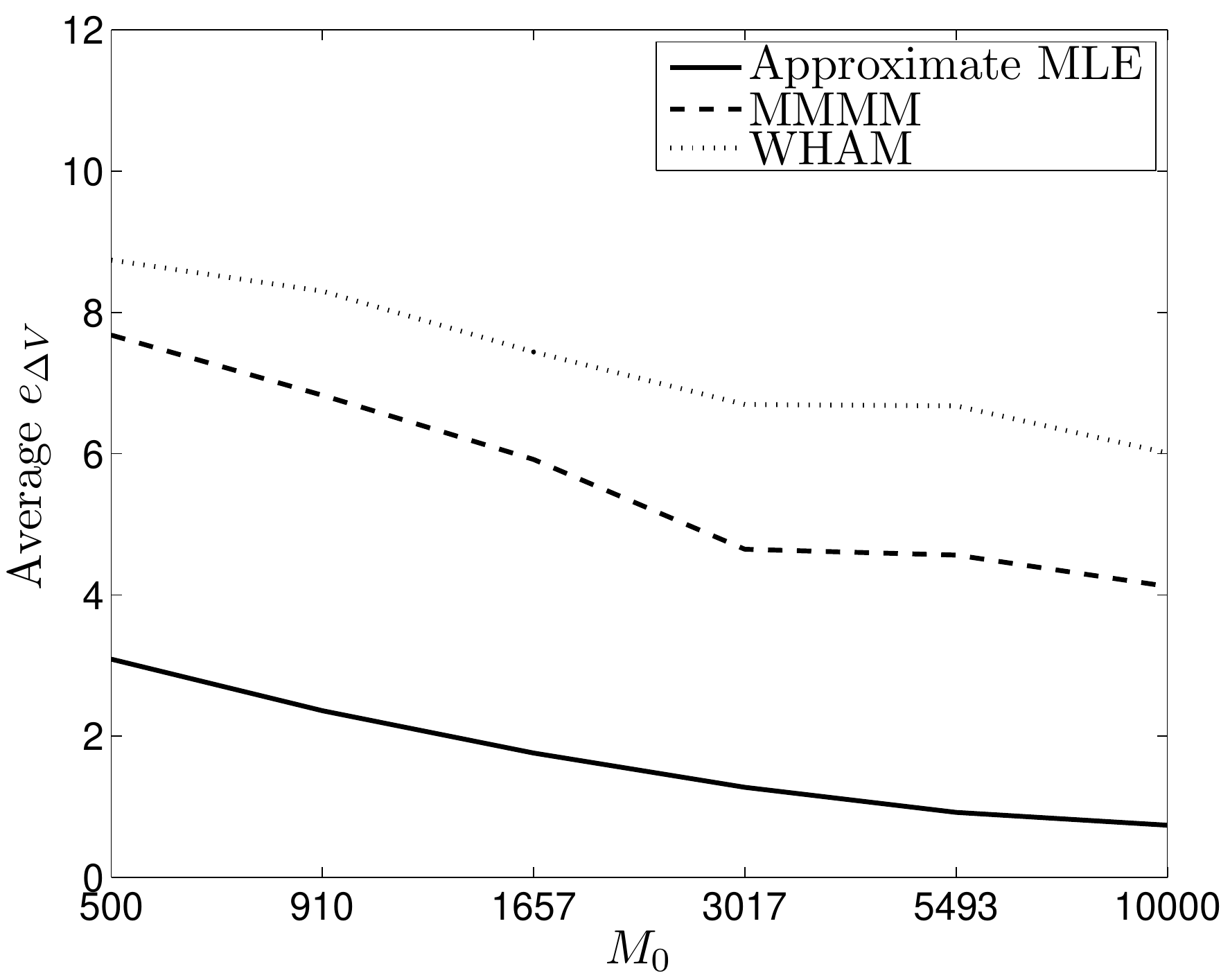}}\hfill{}\subfloat[{Average $e_{\Delta V}$ calculated over $30$ independent runs of
$\mathrm{MBS}\left(K,M_{0}\right)$ for $K=40,80,120,160,200$ and
$M_{0}=\left[2000/K\right]$.}]{\includegraphics[width=0.45\textwidth]{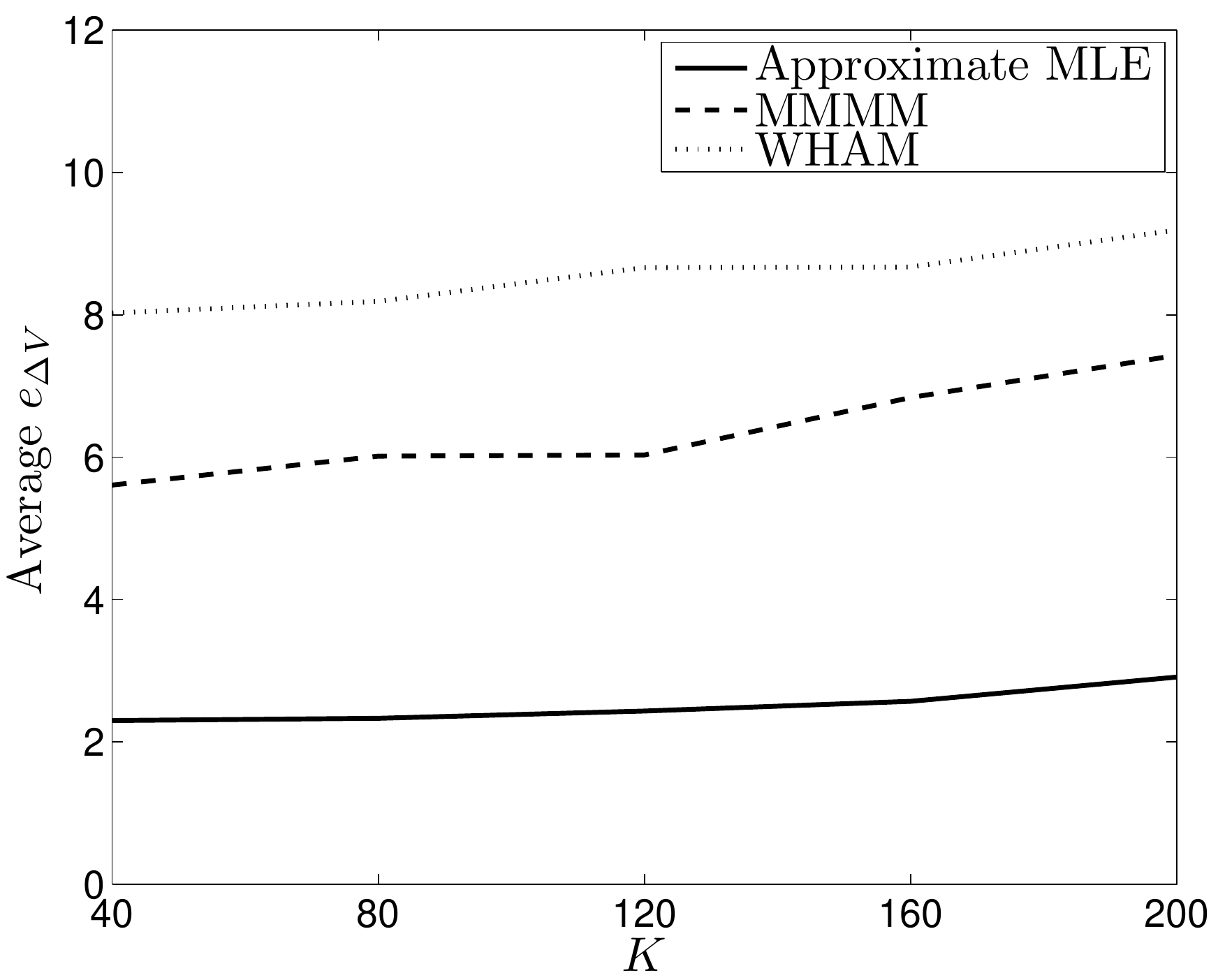}}

\subfloat[Estimates of $V$ generated by the different estimators on a run of
$\mathrm{MBS}\left(15,500\right)$ where the $e_{\Delta V}$ of approximate
MLE $=2.6950$, $e_{\Delta V}$ of MMMM $=8.2297$ and $e_{\Delta V}$
of WHAM $=9.7315$.]{\includegraphics[width=0.45\textwidth]{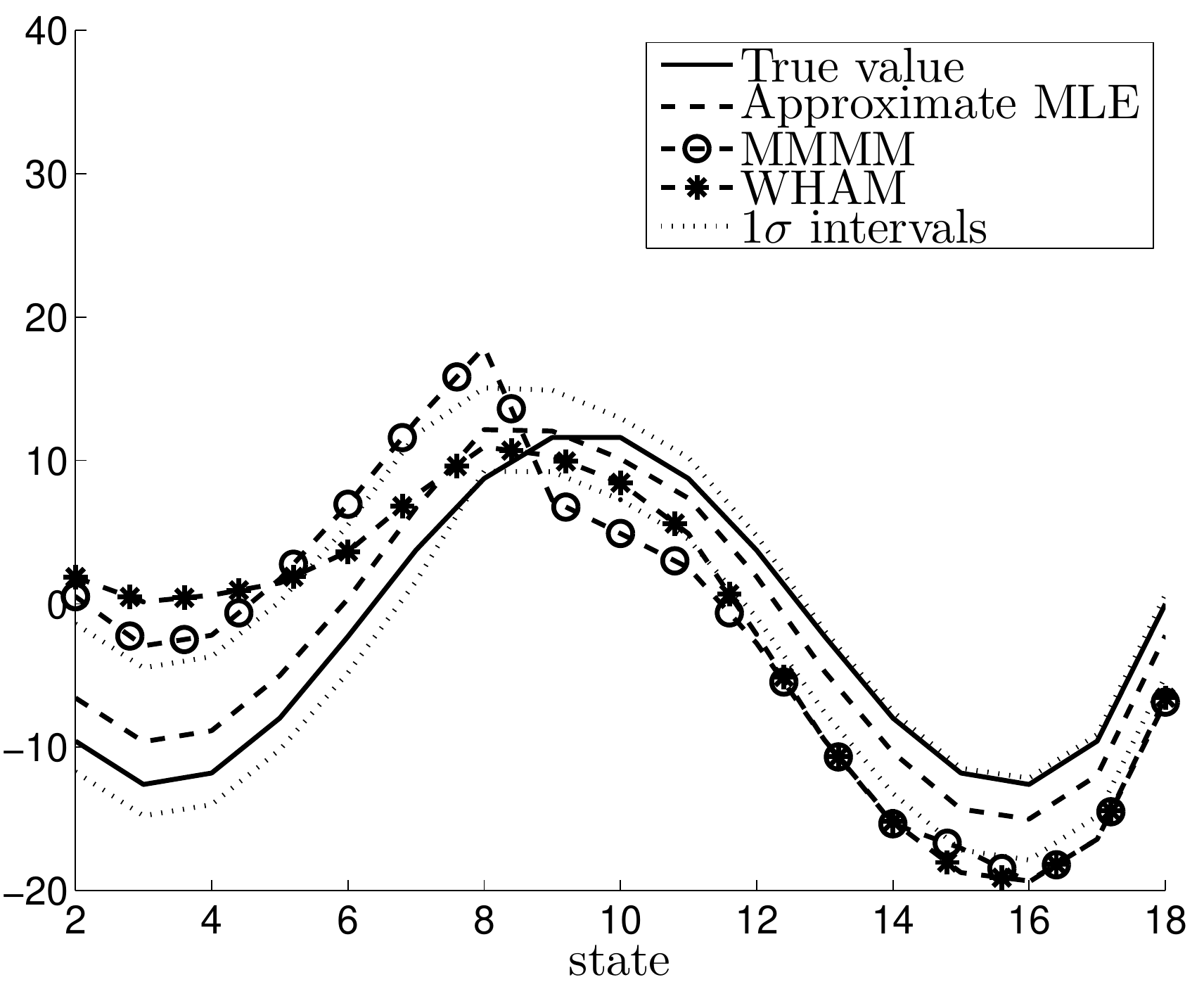}}\hfill{}\subfloat[Estimates of $V$ generated by the different estimators on a run of
$\mathrm{MBS}\left(15,10000\right)$ where the $e_{\Delta V}$ of
approximate MLE $=0.2325$, $e_{\Delta V}$ of MMMM $=2.5563$ and
$e_{\Delta V}$ of WHAM $=3.7233$.]{\includegraphics[width=0.45\textwidth]{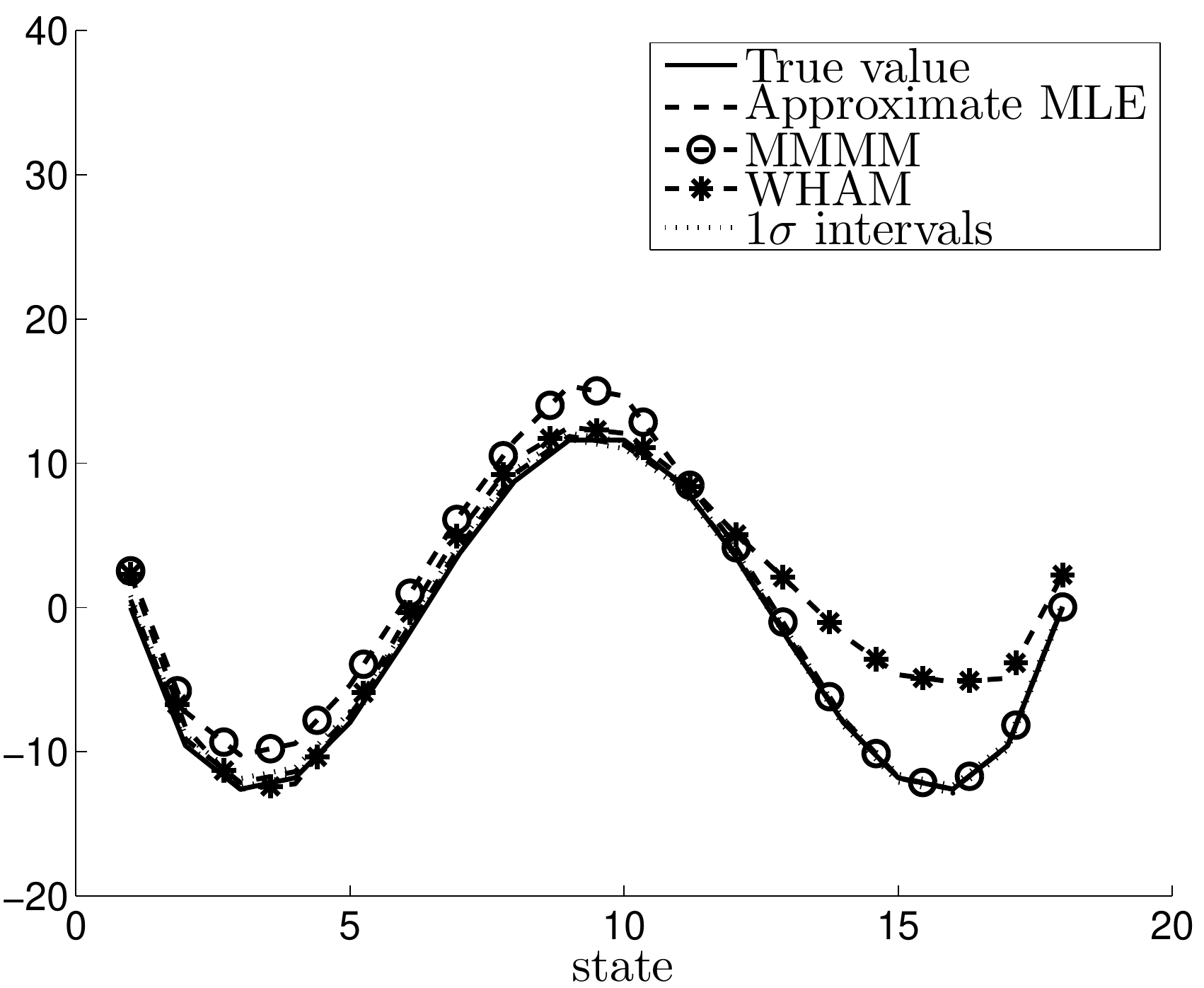}}

\subfloat[Estimates of $V$ generated by the different estimators on a run of
$\mathrm{MBS}\left(40,50\right)$ where the $e_{\Delta V}$ of approximate
MLE $=2.4020$, $e_{\Delta V}$ of MMMM $=4.0281$ and $e_{\Delta V}$
of WHAM $=6.3157$.]{\includegraphics[width=0.45\textwidth]{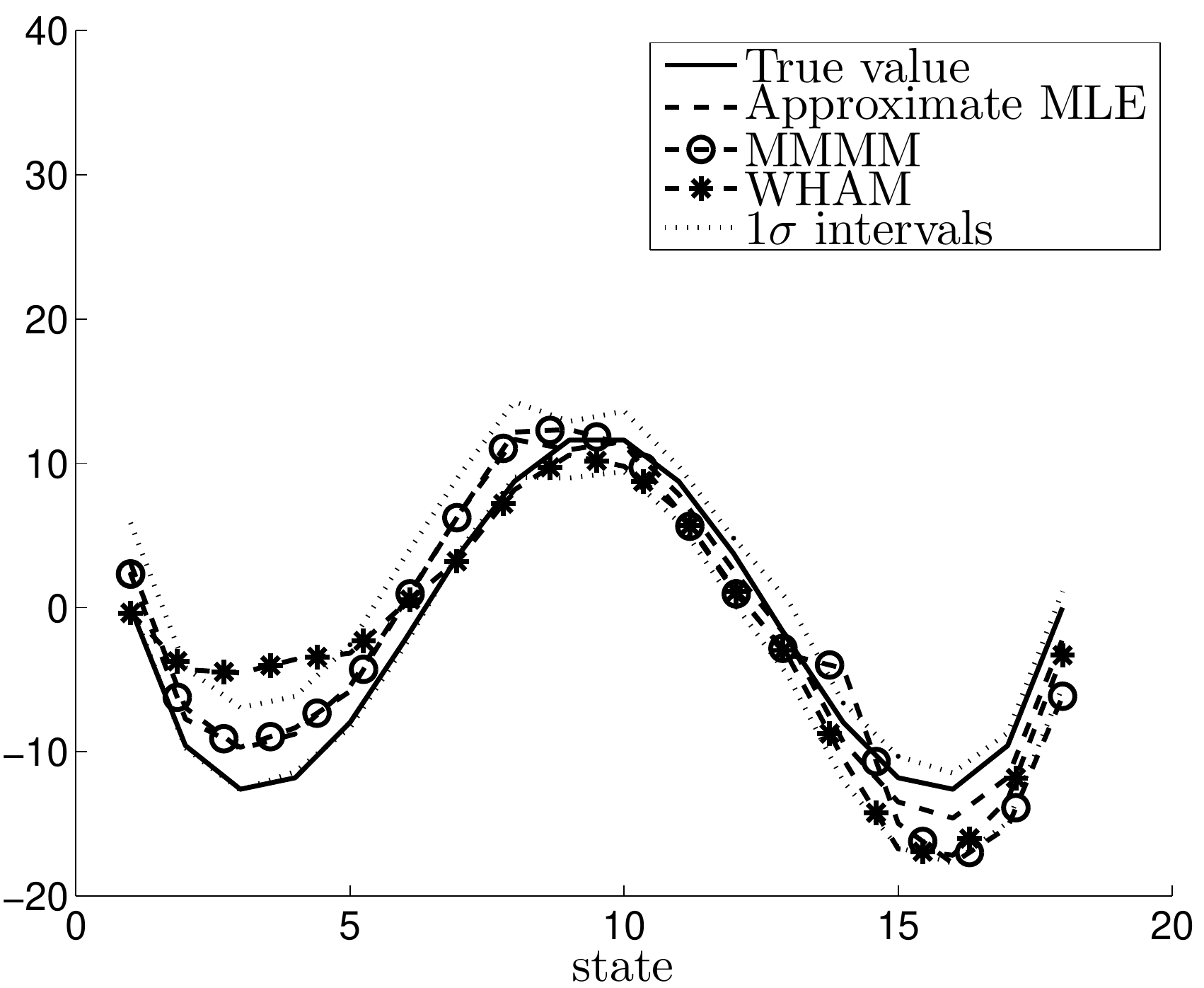}}\hfill{}\subfloat[Estimates of $V$ generated by the different estimators on a run of
$\mathrm{MBS}\left(200,10\right)$ where the $e_{\Delta V}$ of approximate
MLE $=2.5255$, $e_{\Delta V}$ of MMMM $=8.0936$ and $e_{\Delta V}$
of WHAM $=8.6291$.]{\includegraphics[width=0.45\textwidth]{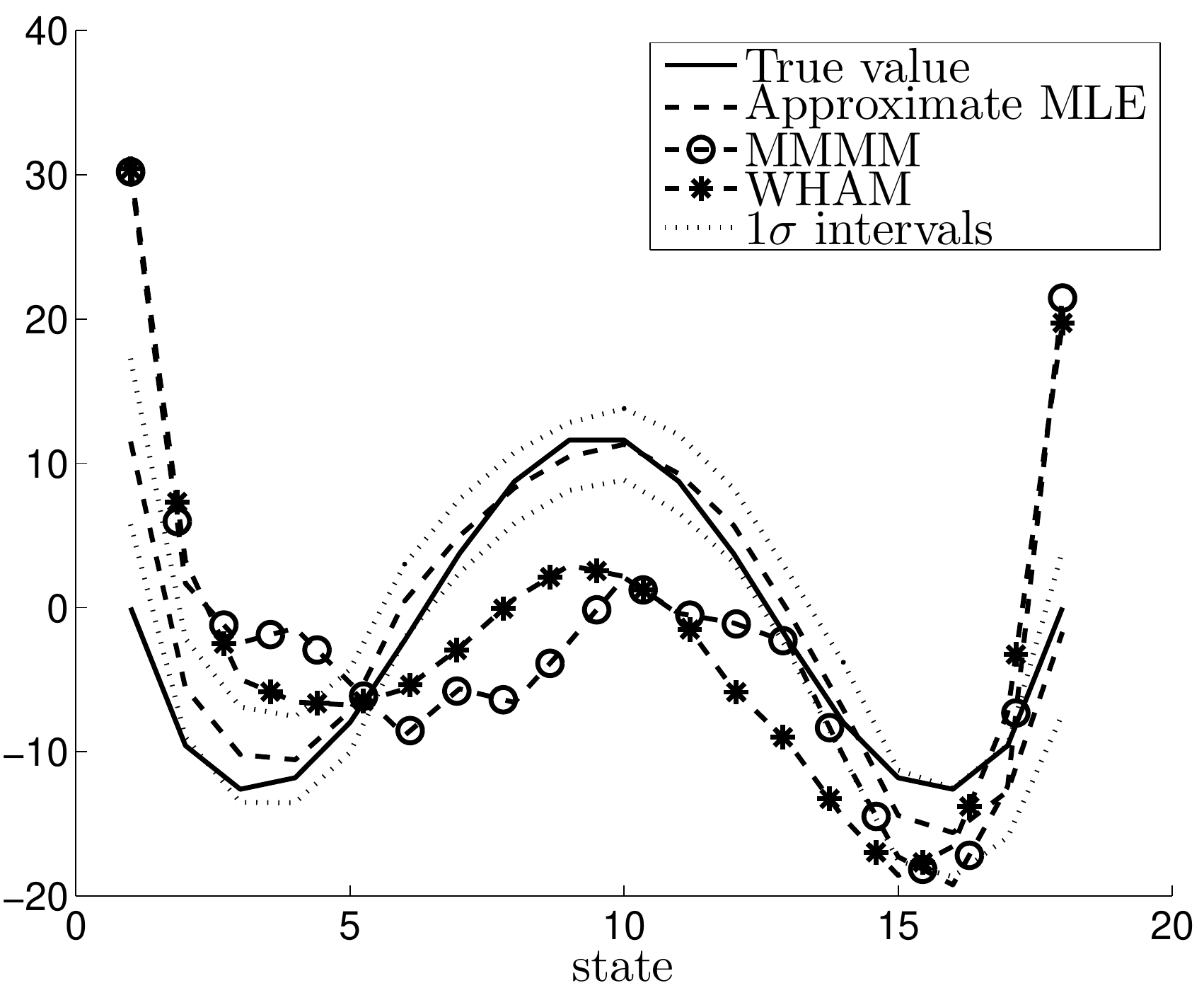}}

\caption{\label{fig:N3N4-mse}Estimation results of metadynamics with Markovian
simulations. The $1\sigma$ confidence intervals in (c) and (d) are
obtained by the approach described in Section \ref{sub:Error-analysis},
and those in (e) and (f) are obtained from the sample standard deviation
of $\check{V}$ in the $30$ independent runs.}
\end{figure}

\section{Conclusions}

We have presented a transition-matrix-based estimation method for
stationary distributions or free energy profiles using data from biased
simulations, such as umbrella sampling or metadynamics. In contrast
to existing estimators such as the weighted histogram method (WHAM),
the present estimators is not based on absolute counts in histogram
bins, but rather based on the transition counts between an arbitrary
state space discretization. This discretization may be in a single
or a few order parameters, e.g. those order parameters in which the
umbrella sampling or metadynamics simulations are driven, or they
may come from the clustering of a higher-dimensional space, such as
frequently use in Markov modeling. The only condition is that the
energy bias used in the biased simulations can be associated to the
discrete states, suggesting that at least the order parameters used
to drive the umbrella sampling/metadynamics simulation should be discretized
finely. The stationary probabilities or free energies are then reconstructed
on the discrete states used. The estimator presented here has a number
of advantages over existing methods such as WHAM. Most importantly,
in all scenarios tested here, the estimation error of the transition-matrix-based
estimator was significantly smaller than that of existing estimation
methods. The reason for this is that the estimator does not rely the
biased simulation to fully equilibrate within one simulation condition,
but only asks for local equilibrium in the discrete states - which
is a much weaker requirement. As a consequence, the present method
can also be used to estimate free energy profiles and stationary distributions
from metadynamics simulations using all simulation data. Previously,
metadynamics simulations could only be analyzed using the fraction
of the simulation generated after the free energy minima have been
filled and the simulation samples from an approximately flat free
energy landscape. These advantages may lead to very substantial savings
of CPU time for a given system, and on the other hand, permit the
simulation of systems that were otherwise out of reach.

\FloatBarrier

\appendix

\section{Proof of Theorem \ref{thm:optimization-problem}\label{sec:Proof-of-Theorem-optimization-problem}}

For convenience, here we define $\Theta$ to be the solution set defined
by constraints in \eqref{eq:original_ml}, $\Theta_{1}$ be the set
of feasible solutions which satisfies $T_{ij}^{k}=0$ for $\left(i,j,k\right)\in\left\{ \left(i,j,k\right)|\ensuremath{C_{ij}^{k}+C_{ji}^{k}=0}\text{ and }\ensuremath{i\neq j}\right\} $,
$\Theta_{2}$ be the set of feasible solutions which satisfies $1_{T_{ij}^{k}>0}=1_{C_{ij}^{k}>0}$
for all $i,j,k$, and the condition
\begin{equation}
\omega:\quad\ensuremath{C_{ii}^{k}>0}\text{ and }\ensuremath{1_{C_{ij}^{k}>0}=1_{C_{ji}^{k}>0}}\text{ for all }i,j,k
\end{equation}

\paragraph*{Part (1)}

In this part, we will prove the optimal solution existence of \eqref{eq:original_ml}.
Suppose that $\left(\pi',{T'}^{1},\ldots,{T'}^{K}\right)$ is a feasible
solution with objective value $L'>-\infty$. We can define a new objective
function $L_{+}\left(\pi,T^{1},\ldots,T^{K}\right)=\max\{L\left(\pi,T^{1},\ldots,T^{K}\right),L'-a\}$
where $a>0$ is a constant. It is easy to verify that $L_{+}$ is
a continuous function on $\Theta$. Thus, the optimization problem
$\max_{\left(\pi,T^{1},\ldots,T^{K}\right)\in\Theta}L_{+}\left(\pi,T^{1},\ldots,T^{K}\right)$
has a global optimal solution $\left(\pi'',{T''}^{1},\ldots,{T''}^{K}\right)$
with $L_{+}\left(\pi'',{T''}^{1},\ldots,{T''}^{K}\right)=L''$ for
$\Theta$ is a closed set. Noting that $L''\ge L'>L'-a$, we have
$L\left(\pi'',{T''}^{1},\ldots,{T''}^{K}\right)=L''$. Therefore,
for any $\left(\pi,T^{1},\ldots,T^{K}\right)\in\Theta$, $L\left(\pi,T^{1},\ldots,T^{K}\right)\le L_{+}\left(\pi,T^{1},\ldots,T^{K}\right)=L\left(\pi'',{T''}^{1},\ldots,{T''}^{K}\right)$.

\paragraph*{Part (2)}

In this part, we will prove the first conclusion of the theorem. Suppose
that $\left(\pi',{T'}^{1},\ldots,{T'}^{K}\right)$ is an optimal solution.
We can define a new solution $\left(\pi',{T''}^{1},\ldots,{T''}^{K}\right)$
with
\begin{equation}
{T''}_{ij}^{k}=\left\{ \begin{array}{ll}
1_{C_{ij}^{k}+C_{ji}^{k}>0}\cdot{T'}_{ij}^{k}, & i\neq j\\
1-\sum_{l\neq i}{T''}_{il}^{k}, & i=j
\end{array}\right.
\end{equation}
Obviously, $\left(\pi',{T''}^{1},\ldots,{T''}^{K}\right)$ is a feasible
solution belonging to $\Theta_{1}$, and ${T''}_{ii}^{k}\ge{T'}_{ii}^{k}$.
We have
\begin{eqnarray}
L\left(\pi',{T''}^{1},\ldots,{T''}^{K}\right) & = & L\left(\pi',{T'}^{1},\ldots,{T'}^{K}\right)+\sum_{i,k}C_{ii}^{k}\left(\log{T''}_{ii}^{k}-\log{T'}_{ii}^{k}\right)\nonumber \\
 & \ge & L\left(\pi',{T'}^{1},\ldots,{T'}^{K}\right)
\end{eqnarray}
Therefore, $\left(\pi',{T''}^{1},\ldots,{T''}^{K}\right)$ is also
an optimal solution.

\paragraph*{Part (3)}

We now prove the second conclusion. Suppose there is an optimal solution
$\left(\pi',{T'}^{1},\ldots,{T'}^{K}\right)$ belonging to $\Theta_{1}\backslash\Theta_{2}$.
Then there exist $i,j,k$ such that ${T'}_{ij}^{k}=0$ and $C_{ij}^{k}>0$,
and $L\left(\pi',{T'}^{1},\ldots,{T'}^{K}\right)=-\infty$. This leads
to a contradiction with the optimality of $\left(\pi',{T'}^{1},\ldots,{T'}^{K}\right)$.
Thus, the optimal solution belonging $\Theta_{1}$ must be an element
of $\Theta_{2}$ if $\omega$ holds.

\section{Proof of Theorem \ref{thm:mle-convergence-p}\label{sec:Proof-of-Theorem-mle-convergence-p}}

From Assumption \ref{ass:X-property}, we have $\hat{X}_{ij}^{k}=0$
if $\bar{X}_{ij}^{k}=0$. Then $\hat{Q}\left(\bar{\theta}\right),\bar{Q}\left(\bar{\theta}\right)>-\infty$
and we can define the following new functions, $\hat{Q}_{+}\left(\theta\right)=\max\left\{ \hat{Q}\left(\theta\right),\hat{Q}\left(\bar{\theta}\right)-a\right\} $
and $\bar{Q}_{+}\left(\theta\right)=\max\left\{ \bar{Q}\left(\theta\right),\bar{Q}\left(\bar{\theta}\right)-a\right\} $,
where $a>0$ is a constant.

\paragraph*{Part (1)}

First of all, we will prove that $\bar{\theta}$ is the unique solution
of $\max_{\theta\in\Theta}\bar{Q}\left(\theta\right)$. Noting that
$\bar{\pi}_{i}^{k}=\sum_{l}\bar{X}_{il}^{k}$, we have
\begin{equation}
\bar{Q}\left(\theta\right)=\sum_{k}\bar{w}_{k}\sum_{i}\left(\sum_{l}\bar{X}_{il}^{k}\right)\sum_{j}\bar{T}_{ij}^{k}\log T_{ij}^{k}
\end{equation}
According to the property of the KL divergence, $Q\left(\theta\right)$
can achieve the maximal value if and only if $T_{ij}^{k}=\bar{T}_{ij}^{k}$
for all $i,j,k$. Then we can conclude from Assumption \ref{ass:pi-property}
that $\theta=\bar{\theta}$ is the unique solution of $\max_{\theta\in\Theta}\bar{Q}\left(\theta\right)$.

\paragraph*{Part (2)}

It is easy to verify that
\begin{equation}
\theta=\arg\max_{\theta\in\Theta}\hat{Q}\left(\theta\right)\Leftrightarrow\theta=\arg\max_{\theta\in\Theta}\hat{Q}_{+}\left(\theta\right)\label{eq:maxQh-maxQhplus}
\end{equation}
and $\bar{\theta}=\arg\max_{\theta\in\Theta}\bar{Q}_{+}\left(\theta\right)$.
The proof is omitted because it is trivial.

\paragraph*{Part (3)}

In this part, we will prove that $\sup_{\theta\in\Theta}\left|\hat{Q}_{+}\left(\theta\right)-\bar{Q}_{+}\left(\theta\right)\right|\stackrel{p}{\to}0$.
Define event
\begin{equation}
\omega:\quad\mbox{\ensuremath{\hat{w}_{k}\hat{X}_{ij}^{k}\ge\epsilon}\text{ for all }\ensuremath{\left(i,j,k\right)\in S_{I}}\text{ and }\ensuremath{\hat{Q}\left(\bar{\theta}\right)\ge\bar{Q}\left(\bar{\theta}\right)-\epsilon}}
\end{equation}
and set
\begin{equation}
\Theta_{1}=\left\{ \theta|T_{ij}^{k}\ge\exp\left(\frac{\bar{Q}\left(\bar{\theta}\right)-\epsilon-a}{\epsilon}\right)\textrm{ for }\ensuremath{\left(i,j,k\right)\in S_{I}}\right\} \cap\Theta
\end{equation}
where $S_{I}=\left\{ \left(i,j,k\right)|\bar{X}_{ij}^{k}>0\right\} $
and $\epsilon\in\left(0,\min_{\left(i,j,k\right)\in S_{I}}\bar{w}_{k}\bar{X}_{ij}^{k}\right)$.
We now analyze the value of $\left|\hat{Q}_{+}\left(\theta\right)-\bar{Q}_{+}\left(\theta\right)\right|$
when $\omega$ holds.
\begin{description}
\item [{Case (i)}] $\theta\in\Theta\backslash\Theta_{1}$. There is a
$\left(i,j,k\right)$ such that $\left(i,j,k\right)\in S_{I}$ and
$T_{ij}^{k}<\exp\left(\frac{\bar{Q}\left(\bar{\theta}\right)-\epsilon-a}{\epsilon}\right)$.
We have
\begin{equation}
\hat{Q}\left(\theta\right)\le\hat{w}_{k}\hat{X}_{ij}^{k}\log T_{ij}^{k}<\bar{Q}\left(\bar{\theta}\right)-\epsilon-a\le\hat{Q}\left(\bar{\theta}\right)-a\label{eq:Qh-smallerthan-Qhb}
\end{equation}
and
\begin{equation}
\bar{Q}\left(\theta\right)\le\bar{w}_{k}\bar{Y}_{ij}^{k}\log T_{ij}^{k}<\bar{Q}\left(\bar{\theta}\right)-\epsilon-a<\bar{Q}\left(\bar{\theta}\right)-a
\end{equation}
Then $\bar{\theta}\in\Theta_{1}$ and
\begin{eqnarray}
\left|\hat{Q}_{+}\left(\theta\right)-\bar{Q}_{+}\left(\theta\right)\right| & \le & \sum_{\left(i,j,k\right)\in S_{I}}\left|\hat{w}_{k}\hat{X}_{ij}^{k}-\bar{w}_{k}\bar{X}_{ij}^{k}\right|\left|\log\bar{T}_{ij}^{k}\right|\nonumber \\
 & \le & \left|\frac{\bar{Q}\left(\bar{\theta}\right)-\epsilon-a}{\epsilon}\right|\sum_{\left(i,j,k\right)\in S_{I}}\left|\hat{w}_{k}\hat{X}_{ij}^{k}-\bar{w}_{k}\bar{X}_{ij}^{k}\right|\label{eq:tmp-upperbound-1}
\end{eqnarray}

\item [{Case (ii)}] $\theta\in\Theta_{1}$. Investigating all possible
orders among values of $\{\hat{Q}\left(\theta\right),\hat{Q}\left(\bar{\theta}\right)-a,\bar{Q}\left(\theta\right),\bar{Q}\left(\bar{\theta}\right)-a\}$
(i.e., $\hat{Q}\left(\theta\right)\le\hat{Q}\left(\bar{\theta}\right)-a\le\bar{Q}\left(\theta\right)\le\bar{Q}\left(\bar{\theta}\right)-a$,
$\hat{Q}\left(\theta\right)\le\hat{Q}\left(\bar{\theta}\right)-a\le\bar{Q}\left(\bar{\theta}\right)-a\le\bar{Q}\left(\theta\right)$
and so on), we have
\begin{eqnarray}
\left|\hat{Q}_{+}\left(\theta\right)-\bar{Q}_{+}\left(\theta\right)\right| & \le & \max\left\{ \left|\hat{Q}\left(\theta\right)-\bar{Q}\left(\theta\right)\right|,\left|\hat{Q}\left(\bar{\theta}\right)-\bar{Q}\left(\bar{\theta}\right)\right|\right\} \nonumber \\
 & \le & \sum_{\left(i,j,k\right)\in S_{I}}\left|\hat{w}_{k}\hat{X}_{ij}^{k}-\bar{w}_{k}\bar{X}_{ij}^{k}\right|\max\left\{ \left|\log T_{ij}^{k}\right|,\left|\log\bar{T}_{ij}^{k}\right|\right\} \nonumber \\
 & \le & \left|\frac{\bar{Q}\left(\bar{\theta}\right)-\epsilon-a}{\epsilon}\right|\sum_{\left(i,j,k\right)\in S_{I}}\left|\hat{w}_{k}\hat{X}_{ij}^{k}-\bar{w}_{k}\bar{X}_{ij}^{k}\right|
\end{eqnarray}

\end{description}
Combining the above results, we get
\begin{equation}
1_{\omega}\cdot\sup_{\theta\in\Theta}\left|\hat{Q}_{+}\left(\theta\right)-\bar{Q}_{+}\left(\theta\right)\right|\le\left|\frac{\bar{Q}\left(\bar{\theta}\right)-\epsilon-a}{\epsilon}\right|\sum_{\left(i,j,k\right)\in S_{I}}\left|\hat{w}_{k}\hat{X}_{ij}^{k}-\bar{w}_{k}\bar{X}_{ij}^{k}\right|
\end{equation}
Moreover, considering that $\hat{w}_{k}\stackrel{p}{\to}\bar{w}_{k}$,
$\hat{X}_{ij}^{k}\stackrel{p}{\to}\bar{X}_{ij}^{k}$ and
\begin{equation}
\left|\hat{Q}\left(\bar{\theta}\right)-\bar{Q}\left(\bar{\theta}\right)\right|\le\left|\frac{\bar{Q}\left(\bar{\theta}\right)-\epsilon-a}{\epsilon}\right|\sum_{\left(i,j,k\right)\in S_{I}}\left|\hat{w}_{k}\hat{X}_{ij}^{k}-\bar{w}_{k}\bar{X}_{ij}^{k}\right|
\end{equation}
we have $1_{\omega}\stackrel{p}{\to}1$. Therefore
\begin{equation}
\sup_{\theta\in\Theta}\left|\hat{Q}_{+}\left(\theta\right)-\bar{Q}_{+}\left(\theta\right)\right|\stackrel{p}{\to}0
\end{equation}

According to the definitions of $\hat{Q}_{+}\left(\theta\right)$
and $\bar{Q}_{+}\left(\theta\right)$ and conclusions of Parts (1)-(3),
it can be easily verified that $\hat{Q}_{+}\left(\theta\right)$ satisfies
the following conditions: (i) $\bar{Q}_{+}\left(\theta\right)$ is
uniquely maximized at $\bar{\theta}$, (ii) $\Theta$ is compact;
(iii) $\bar{Q}_{+}\left(\theta\right)$ is continuous; (iv) $\hat{Q}_{+}\left(\theta\right)$
converges uniformly in probability to $\bar{Q}_{+}\left(\theta\right)$.
Then we have $\tilde{\theta}\stackrel{p}{\to}\bar{\theta}$ by using
Theorem 2.1 in \cite{newey1994large} and \eqref{eq:maxQh-maxQhplus}.

\section{Proof of Theorem \ref{thm:asymptotic-convergence}\label{sec:Proof-of-Theorem-asymptotic-convergence}}

For any $k$,
\begin{eqnarray}
\sqrt{M}\hat{w}_{k}\hat{V}_{X}^{k}-\sqrt{M}\bar{w}_{k}\bar{V}_{X}^{k} & = & \sqrt{\hat{w}_{k}}\sqrt{M_{k}}\left(\hat{V}_{X}^{k}-\bar{V}_{X}^{k}\right)+\sqrt{M}\left(\hat{w}_{k}-\bar{w}_{k}\right)\bar{V}_{X}^{k}\nonumber \\
 & \stackrel{d}{\to} & \mathcal{N}\left(0,\bar{w}_{k}\Sigma_{X}^{k}\right)
\end{eqnarray}
Then $\sqrt{M}\left(\hat{V}_{Xw}-\bar{V}_{Xw}\right)\stackrel{d}{\to}\mathcal{N}\left(0,\Sigma_{Xw}\right)$.

Let $\Theta_{r}$ be the feasible set of $\theta_{r}$ defined by
constraints in \eqref{eq:original_ml}, where $T_{ij}^{k}$ and $\pi_{\left|\mathcal{S}\right|}$
which do not belong to $\theta_{r}$ can be treated as functions of
$\theta_{r}$ since the MLE is performed by using \eqref{eq:reduced_ml}.
It is easy to see that $\bar{\theta}_{r}$ is an interior point of
$\Theta_{r}$, and
\begin{equation}
\Theta_{1}=\left\{ \theta_{r}|T_{ij}^{k}\left(\theta_{r}\right)\ge\epsilon\text{ for }\left(i,j,k\right)\in S_{I}\text{ and }\pi_{i}\left(\theta_{r}\right)\ge\epsilon\text{ for all }i\right\} \cap\Theta_{r}
\end{equation}
is a closed neighborhood of $\bar{\theta}_{r}$, where $\epsilon\in\left(0,\min\left\{ \min_{\left(i,j,k\right)\in S_{I}}\bar{T}_{ij}^{k},\min_{i}\bar{\pi}_{i}\right\} \right)$
and $S_{I}$ has the same definition as in Appendix \ref{sec:Proof-of-Theorem-mle-convergence-p}.

It is easy to verify that
\begin{eqnarray*}
\sqrt{M}\left(\nabla_{\theta_{r}}\hat{Q}\left(\theta\left(\bar{\theta}_{r}\right)\right)-\nabla_{\theta_{r}}\bar{Q}\left(\theta\left(\bar{\theta}_{r}\right)\right)\right)^{\mathrm{T}} & = & \sqrt{M}\nabla_{\theta_{r}}\Phi\left(\theta\left(\bar{\theta}_{r}\right)\right)^{\mathrm{T}}\left(\hat{V}_{Xw}-\bar{V}_{Xw}\right)\\
 & \stackrel{d}{\to} & \mathcal{N}\left(0,\Sigma\right)
\end{eqnarray*}
and $\sup_{\theta_{r}\in\Theta_{1}}\left\Vert \nabla_{\theta_{r}\theta_{r}}\hat{Q}\left(\theta\left(\theta_{r}\right)\right)-\nabla_{\theta_{r}\theta_{r}}\bar{Q}\left(\theta\left(\theta_{r}\right)\right)\right\Vert \stackrel{p}{\to}0$.
Then by Theorem 3.1 in \cite{newey1994large}, $\sqrt{M}\left(\tilde{\theta}_{r}-\bar{\theta}_{r}\right)\stackrel{d}{\to}\mathcal{N}\left(0,H^{-1}\Sigma H^{-1}\right)$.

\section{Proof of Remark \ref{rem:clt}\label{sec:Proof-of-Remark-clt}}

In this appendix we will show that \eqref{eq:X-clt} holds if $x_{t}^{k}=f^{k}\left(y_{t}^{k}\right)$
and $\left\{ y_{t}^{k}\right\} $ is a geometrically ergodic Markov
chain.

Let $\mathcal{R}$ and $\mu\left(\cdot\right)$ denote the state space
and stationary distribution of $\left\{ y_{t}^{k}\right\} $. Then
for any $y_{0}^{k}$, $y_{1}^{k}$, $t>1$ and set $A\subset\mathcal{R}\times\mathcal{R}$,
we have
\begin{eqnarray}
 &  & \left|\Pr\left(\left(y_{t}^{k},y_{t+1}^{k}\right)\in A|y_{0}^{k},y_{1}^{k}\right)-\int_{A_{0}}\mu\left(\mathrm{d}y_{t}^{k}\right)\mathcal{K}\left(y_{t}^{k},A_{1}\left(y_{t}^{k}\right)\right)\right|\nonumber \\
 & = & \left|\int_{\mathcal{R}}\mathcal{K}^{t-1}\left(y_{1}^{k},\mathrm{d}y_{t}^{k}\right)g\left(y_{t}^{k}\right)-\int_{\mathcal{R}}\mu\left(\mathrm{d}y_{t}^{k}\right)g\left(y_{t}^{k}\right)\right|
\end{eqnarray}
with
\begin{equation}
g\left(y_{t}^{k}\right)=\left\{ \begin{array}{ll}
\mathcal{K}\left(y_{t}^{k},A_{1}\left(y_{t}^{k}\right)\right), & y_{t}^{k}\in A_{0}\\
0, & y_{t}^{k}\not\notin A_{0}
\end{array}\right.
\end{equation}
where $A_{0}=\left\{ u|\text{there exists }v\text{ such that }\left(u,v\right)\in A\right\} $,
$A_{1}\left(u\right)=\left\{ v|\left(u,v\right)\in A\right\} $, $\mathcal{K}\left(u,B\right)=\Pr\left(y_{t+1}^{k}\in B|y_{t}^{k}=u\right)$
and $\mathcal{K}^{n}\left(u,B\right)=\Pr\left(y_{t+n}^{k}\in B|y_{t}^{k}=u\right)$
denote the transition kernel and multi-step transition kernel of $\left\{ y_{t}^{k}\right\} $.
It is clear that $\left|g\left(y_{t}^{k}\right)\right|\le1$ for any
$y_{t}^{k}$. Considering that $\left\{ y_{t}^{k}\right\} $ is geometrically
ergodic, we can deduce that there exist a function $a\left(\cdot\right)$
and a constant $\eta\in\left(0,1\right)$ such that
\begin{eqnarray}
 &  & \left|\Pr\left(\left(y_{t}^{k},y_{t+1}^{k}\right)\in A|y_{0}^{k},y_{1}^{k}\right)-\int_{A_{0}}\mu\left(\mathrm{d}y_{t}^{k}\right)\mathcal{K}\left(y_{t}^{k},A_{1}\left(y_{t}^{k}\right)\right)\right|\nonumber \\
 & \le & 2\mathrm{d_{TV}}\left(\mathcal{K}^{t-1}\left(y_{1}^{k},\cdot\right),\mu\left(\cdot\right)\right)\nonumber \\
 & \le & 2a\left(y_{1}^{k}\right)\eta^{t-1}\label{eq:appendix-geometrically-ergodic-ytyt1}
\end{eqnarray}
Because \eqref{eq:appendix-geometrically-ergodic-ytyt1} holds for
any $A\subset\mathcal{R}\times\mathcal{R}$, the stochastic process
$\left\{ {y'}_{t}^{k}\right\} $ with ${y'}_{t}^{k}=\left(y_{t}^{k},y_{t+1}^{k}\right)$
is also a geometrically ergodic Markov chain with stationary distribution
$\mu'\left(\mathrm{d}y_{t}^{k},\mathrm{d}y_{t+1}^{k}\right)=\mu\left(\mathrm{d}y_{t}^{k}\right)\mathcal{K}\left(y_{t}^{k},\mathrm{d}y_{t+1}^{k}\right)$.
Hence, combining Theorem 1.2 in \cite{haggstrom2005central} and the
Cramer-Wold device \cite{roussas2005introduction}, we can conclude
that $\sqrt{M_{k}}\left(\mathcal{V}\left(\hat{X}^{k}\right)-\mathcal{V}\left(\bar{X}^{k}\right)\right)$
converges in distribution to a multivariate normal distribution with
zero mean.

\section{Proof of Theorem \ref{thm:modified-count-matrix}\label{sec:Proof-of-Theorem-modified-count-matrix}}

\paragraph*{Part (1)}

It is clear that the modified count matrix $\underline{C}^{k}$ in
\eqref{eq:modified-C} satisfies that $1_{\underline{C}_{ij}^{k}>0}\equiv1_{\underline{C}_{ji}^{k}>0}$
and its diagonal elements are all positive. Here we only prove the
irreducibility of $\underline{C}^{k}$ by contradiction. Assume without
loss of generality that there exists $n\in\left[1,\left|\mathcal{S}^{k}\right|\right)$
such that
\begin{equation}
\underline{C}_{ij}^{k}=0,\quad\text{for }i\le n\text{ and }j>n
\end{equation}
Then $C_{ij}^{k}=C_{ji}^{k}=0$ for $i\le n$ and $j>n$, which implies
the states $s_{1}$ and $s_{2}$ can not both appear in the $k$-th
simulation if $I_{\mathcal{S}^{k}}\left(s_{1}\right)\le n$ and $I_{\mathcal{S}^{k}}\left(s_{2}\right)>n$.
This contradicts the condition that $\sum_{t=0}^{M_{k}}1_{I_{\mathcal{S}^{k}}\left(x_{t}^{k}\right)=i}>0$
for all $i$. So $\underline{C}^{k}$ is irreducible.

\paragraph*{Part (2)}

We now show that $\bar{X}^{k}$ is irreducible in the same way as
in the first part. Assume without loss of generality that there exists
$n\in\left[1,\left|\mathcal{S}^{k}\right|\right)$ such that
\begin{equation}
\bar{X}_{ij}^{k}=0,\quad\text{for }i\le n\text{ and }j>n
\end{equation}
Then $\bar{X}_{ij}^{k}=\bar{X}_{ji}^{k}=0$ for $i\le n$ and $j>n$
since $\bar{X}^{k}$ is symmetric. According to Assumption \ref{ass:X-property},
for any $s_{1},s_{2}\in\mathcal{S}^{k}$ which satisfy $I_{\mathcal{S}^{k}}\left(s_{1}\right)\le n$
and $I_{\mathcal{S}^{k}}\left(s_{2}\right)>n$, the probability of
that $s_{1}$ and $s_{2}$ both appear in the $k$-th simulation is
zero. This contradicts the ergodicity of the simulation. So $\bar{X}^{k}$
is irreducible.

\paragraph*{Part (3)}

In this part we will prove that $1_{\underline{C}^{k}=C^{k}}\stackrel{p}{\to}1$.
Define
\begin{equation}
\omega:\quad C^{k}\text{ is an irreducible matrix and satisfy }C_{ii}^{k}>0\text{ and }1_{C_{ij}^{k}>0}=1_{C_{ij}^{k}>0}
\end{equation}
According to Assumptions \ref{ass:simulations} and \ref{ass:X-property},
we have $1_{C_{ij}^{k}>0}\stackrel{p}{\to}1_{\bar{X}_{ij}^{k}>0}$.
Therefore $1_{\omega}\stackrel{p}{\to}1$ since $\bar{X}^{k}$ is
a symmetric and irreducible matrix with positive diagonal elements.
On the other hand, $\underline{C}^{k}=C^{k}$ if $\omega$ holds.
Hence $1_{\underline{C}^{k}=C^{k}}\stackrel{p}{\to}1$.

\section{Proof of $H_{\rho}^{k}<0$\label{sec:Properties-of-Hrhok}}

Define $\rho_{a}^{k}$ to be the vector consisting of $\left\{ Z_{ij}^{k}|\underline{C}_{ij}^{k}>0\right\} $
and $H_{\rho_{a}}^{k}=\sum_{i}\underline{C}_{i}^{k}\nabla_{\rho_{a}^{k}\rho_{a}^{k}}Z_{i}^{k}$.
It is clear that $\rho_{a}^{k}$ can be expressed as a linear function
of $\rho^{k}$ with $\rho_{a}^{k}=A_{a}^{k}\rho^{k}$. Then $H_{\rho}^{k}$
can be written as $H_{\rho}^{k}=\left(A_{a}^{k}\right)^{\mathrm{T}}H_{\rho_{a}}^{k}A_{a}^{k}$.

According to \eqref{eq:Zik-hess}, for any ${\rho'}_{a}^{k}=\left\{ {Z'}_{ij}^{k}|\underline{C}_{ij}^{k}>0\right\} $,
\begin{eqnarray}
\left({\rho'}_{a}^{k}\right)^{\mathrm{T}}H_{\rho_{a}}^{k}{\rho'}_{a}^{k} & = & \sum_{i}\underline{C}_{i}^{k}\sum_{m,n\in\left\{ j|C_{ij}^{k}>0\right\} }\frac{\partial^{2}Z_{i}^{k}}{\partial Z_{im}^{k}\partial Z_{in}^{k}}{Z'}_{im}^{k}{Z'}_{in}^{k}\nonumber \\
 & = & -\sum_{i}\underline{C}_{i}^{k}\left(\sum_{j\in\left\{ j|C_{ij}^{k}>0\right\} }T_{ij}^{k}\left({Z'}_{ij}^{k}\right)^{2}-\left(\sum_{j\in\left\{ j|C_{ij}^{k}>0\right\} }T_{ij}^{k}{Z'}_{ij}^{k}\right)^{2}\right)
\end{eqnarray}
Noting that for any $i$, the term
\[
\sum_{j\in\left\{ j|C_{ij}^{k}>0\right\} }T_{ij}^{k}\left({Z'}_{ij}^{k}\right)^{2}-\left(\sum_{j\in\left\{ j|C_{ij}^{k}>0\right\} }T_{ij}^{k}{Z'}_{ij}^{k}\right)^{2}
\]
is equal to the variance of a discrete random variable $a_{i}$ with
distribution $\Pr\left(a_{i}={Z'}_{ij}^{k}\right)=T_{ij}^{k}$ and
support $\left\{ {Z'}_{ij}^{k}|\underline{C}_{ij}^{k}>0\right\} $
($T_{ij}^{k}=\exp\left(-Z_{ij}^{k}\right)/\sum_{l}\exp\left(-Z_{il}^{k}\right)>0$
if $\underline{C}_{ij}^{k}>0$). We can conclude that $H_{\rho_{a}}^{k}\le0$,
and $\left({\rho'}_{a}^{k}\right)^{\mathrm{T}}H_{\rho_{a}}^{k}{\rho'}_{a}^{k}=0$
if and only if
\begin{equation}
{Z'}_{im}^{k}={Z'}_{in}^{k}\text{ for any }m,n\in\left\{ j|\underline{C}_{ij}^{k}>0\right\} ,1\le i\le\left|\mathcal{S}^{k}\right|\label{eq:Z-row-equal}
\end{equation}
Further, if there exists a ${\rho'}^{k}$ such that ${\rho'}_{a}^{k}=A_{a}^{k}{\rho'}^{k}$,
${\rho'}_{a}^{k}$ should satisfy
\begin{equation}
\ensuremath{{Z'}_{ij}^{k}={Z'}_{ji}^{k}}\text{ for any \ensuremath{i,j}, and }\sum_{\left(i,j\right)\in\left\{ \left(i,j\right)|\underline{C}_{ij}^{k}>0\right\} }{Z'}_{ij}^{k}=0\label{eq:rhoa-constraint}
\end{equation}
According to \eqref{eq:Z-row-equal}, \eqref{eq:rhoa-constraint}
and Assumption \ref{ass:count-matrix}, it is easy to see ${\rho'}^{k\mathrm{T}}A_{a}^{k\mathrm{T}}H_{\rho_{a}}^{k}A_{a}^{k}{\rho'}^{k}=0$
if and only if ${\rho'}^{k}=\mathbf{0}$, which implies $H_{\rho}^{k}=A_{a}^{k\mathrm{T}}H_{\rho_{a}}^{k}A_{a}^{k}<0$.

\section{Proof of Lemma \ref{lem:rhoV}\label{sec:Proof-of-Lemma-rhoV}}

Under Assumptions \ref{ass:potential}-\ref{ass:pi-property} and
\ref{ass:count-matrix}, it is easy to see that $\bar{X}^{k}$ is
irreducible and its diagonal elements are positive (see Appendix \ref{sec:Proof-of-Theorem-modified-count-matrix}),
which implies that $\bar{T}^{k}$ is an ergodic transition matrix.
Thus $\bar{V}^{k}=V_{Z}^{k}\left(\bar{\rho}^{k}\right)$.

Define event
\begin{equation}
\omega^{k}:\quad\mbox{\ensuremath{1_{C_{ij}^{k}>0}=1_{\bar{X}_{ij}^{k}>0}}\text{ for all \ensuremath{i,j}}}
\end{equation}
and
\begin{equation}
\underline{\check{\rho}}^{k}=\left\{ \begin{array}{ll}
\check{\rho}^{k}, & \mathrm{dim}\left(\check{\rho}^{k}\right)=\mathrm{dim}\left(\bar{\rho}^{k}\right)\\
\mathbf{0}, & \mathrm{dim}\left(\check{\rho}^{k}\right)\neq\mathrm{dim}\left(\bar{\rho}^{k}\right)
\end{array}\right.
\end{equation}
It is clear that the functions $1_{\omega^{k}}\cdot\check{Q}^{k}\left(\underline{\rho}^{k}\right)$
and $\bar{Q}^{k}\left(\underline{\rho}^{k}\right)$ satisfy: (i) $\check{\underline{\rho}}^{k}=\arg\max_{\underline{\rho}^{k}}1_{\omega^{k}}\cdot\check{Q}^{k}\left(\underline{\rho}^{k}\right)$
and $\bar{Q}^{k}\left(\underline{\rho}^{k}\right)$ is uniquely maximized
at $\bar{\rho}^{k}$, (ii) $1_{\omega^{k}}\cdot\check{Q}^{k}\left(\underline{\rho}^{k}\right)$
is concave for $\nabla_{\underline{\rho}^{k}\underline{\rho}^{k}}\check{Q}^{k}\left(\underline{\rho}^{k}\right)=H_{\rho}^{k}\left(\hat{w}^{k}\hat{\underline{X}}^{k},\underline{\rho}^{k}\right)<0$,
(iii) $1_{\omega^{k}}\cdot\check{Q}^{k}\left(\underline{\rho}^{k}\right)\stackrel{p}{\to}\bar{Q}^{k}\left(\underline{\rho}^{k}\right)$
for any $\underline{\rho}^{k}$ because $1_{\omega^{k}}\stackrel{p}{\to}1$.
Then we have $\check{\rho}^{k}\stackrel{p}{\to}\underline{\check{\rho}}^{k}\stackrel{p}{\to}\bar{\rho}^{k}$
according to Theorem 2.7 in \cite{newey1994large}, and $\check{V}^{k}\stackrel{p}{\to}\bar{V}^{k}$
since $V_{Z}^{k}\left(\underline{\rho}^{k}\right)$ is a continuous
function of $\underline{\rho}^{k}$.

We now show the second conclusion of the lemma. Because $\bar{\rho}^{k}=\arg\max_{\underline{\rho}^{k}}\bar{Q}^{k}\left(\underline{\rho}^{k}\right)$,
$\nabla_{\underline{\rho}_{k}}\bar{Q}^{k}\left(\bar{\rho}^{k}\right)=\mathbf{0}^{\mathrm{T}}$.
Then we have
\begin{eqnarray}
\sqrt{M_{k}}\left(\nabla_{\underline{\rho}_{k}}\left(1_{\omega^{k}}\cdot\check{Q}^{k}\left(\bar{\rho}^{k}\right)\right)\right)^{\mathrm{T}} & = & 1_{\omega^{k}}\cdot\hat{w}_{k}\cdot\left(\nabla_{\underline{\rho}_{k}}\Phi^{k}\left(\bar{\rho}^{k}\right)\right)^{\mathrm{T}}\cdot\left(\mathcal{V}\left(\hat{\underline{X}}^{k}\right)-\mathcal{V}\left(\bar{X}^{k}\right)\right)\nonumber \\
 & \stackrel{p}{\to} & 1_{\omega^{k}}\cdot\hat{w}_{k}\cdot\left(\nabla_{\underline{\rho}_{k}}\Phi^{k}\left(\bar{\rho}^{k}\right)\right)^{\mathrm{T}}\cdot\left(\mathcal{V}\left(\hat{X}^{k}\right)-\mathcal{V}\left(\bar{X}^{k}\right)\right)\nonumber \\
 & \stackrel{d}{\to} & \mathcal{N}\left(\mathbf{0},\left(\bar{w}_{k}\right)^{2}\left(\nabla_{\underline{\rho}_{k}}\Phi^{k}\left(\bar{\rho}^{k}\right)\right)^{\mathrm{T}}\Sigma_{X}^{k}\left(\nabla_{\underline{\rho}_{k}}\Phi^{k}\left(\bar{\rho}^{k}\right)\right)\right)
\end{eqnarray}
Furthermore, the Hessian matrices of $1_{\omega^{k}}\cdot\check{Q}^{k}\left(\underline{\rho}^{k}\right)$
and $\bar{Q}^{k}\left(\underline{\rho}^{k}\right)$ satisfy $\nabla_{\underline{\rho}^{k}\underline{\rho}^{k}}\left(1_{\omega^{k}}\cdot\check{Q}^{k}\left(\underline{\rho}^{k}\right)\right)\stackrel{p}{\to}\nabla_{\underline{\rho}^{k}\underline{\rho}^{k}}\bar{Q}^{k}\left(\underline{\rho}^{k}\right)$
for any $\underline{\rho}^{k}$. Using the mean value theorem and
Theorem 3.1 in \cite{newey1994large}, we can conclude that $\sqrt{M_{k}}\left(\check{\underline{\rho}}^{k}-\bar{\rho}^{k}\right)\stackrel{d}{\to}\mathcal{N}\left(\mathbf{0},\Sigma_{\rho}^{k}\right)$
and
\begin{eqnarray}
\sqrt{M_{k}}\left(\check{V}^{k}-\bar{V}^{k}\right) & \stackrel{p}{\to} & \sqrt{M_{k}}\nabla_{\underline{\rho}^{k}}V_{Z}^{k}\left(a\check{\underline{\rho}}^{k}+\left(1-a\right)\bar{\rho}^{k}\right)\left(\check{\underline{\rho}}^{k}-\bar{\rho}^{k}\right)\nonumber \\
 & \stackrel{d}{\to} & \mathcal{N}\left(\mathbf{0},\Sigma_{V}^{k}\right)
\end{eqnarray}
where $a\in\left[0,1\right]$ is a function of $\check{\underline{\rho}}^{k}$.

\section{Proof of Theorem \ref{thm:convergence-approximate-mle}\label{sec:Proof-of-Theorem-convergence-approximate-mle}}

Since the value of $\check{V}$ will not be affected if we replace
$\underline{C}^{k}$ with $\hat{w}^{k}\hat{\underline{X}}^{k}=\underline{C}^{k}/M$,
$\check{V}$ can be expressed as
\begin{eqnarray}
\check{V} & = & \left(\sum_{k=1}^{K}A^{k\mathrm{T}}\Xi^{k}\left(\hat{w}^{k}\hat{\underline{X}}^{k},\check{\rho}^{k}\right)A^{k}\right)^{+}\nonumber \\
 &  & \cdot\left(\sum_{k=1}^{K}A^{k\mathrm{T}}\Xi^{k}\left(\hat{w}^{k}\hat{\underline{X}}^{k},\check{\rho}^{k}\right)\left(\check{V}^{k}-A^{k}U^{k}\right)\right)\label{eq:Vcheck-new-expression}
\end{eqnarray}
and $\check{V}=\arg\max_{\mathbf{1}^{\mathrm{T}}V=0}Q_{q}\left(V;\left\{ \check{V}^{k}\right\} \right)$
with
\begin{eqnarray}
Q_{q}\left(V;\left\{ \check{V}^{k}\right\} \right) & = & \sum_{k=1}^{K}\frac{1}{2}\left(A^{k}\left(V+U^{k}\right)-\check{V}^{k}\right)^{\mathrm{T}}\Xi^{k}\left(\hat{w}^{k}\hat{\underline{X}}^{k},\check{\rho}^{k}\right)\nonumber \\
 &  & \cdot\left(A^{k}\left(V+U^{k}\right)-\check{V}^{k}\right)
\end{eqnarray}

\paragraph*{Part (1)}

We first prove that $\bar{V}$ is the unique maximum point of $Q_{q}(V;\{\bar{V}^{k}\})$
under the constraint $\mathbf{1}^{\mathrm{T}}V=0$. It is clear that
$\bar{V}$ is a maximum point since $Q_{q}(\bar{V};\{\bar{V}^{k}\})=0$
and $Q_{q}(V;\{\bar{V}^{k}\})\le0$. We now show the uniqueness of
$\bar{V}$ by contradiction. Suppose that $V'$ is another maximum
point which satisfies $Q_{q}(V';\{\bar{V}^{k}\})=0$ and $V'\neq\bar{V}$.
Then $\left[\begin{array}{cc}
\mathbf{I} & \mathbf{0}\end{array}\right]\left(A^{k}\left(V'+U^{k}\right)-\bar{V}^{k}\right)=\mathbf{0}$ for any $k$ because $\left[\begin{array}{cc}
\mathbf{I} & \mathbf{0}\end{array}\right]\cdot\Xi^{k}\left(\hat{w}^{k}\hat{\underline{X}}^{k},\check{\rho}^{k}\right)\cdot\left[\begin{array}{cc}
\mathbf{I} & \mathbf{0}\end{array}\right]^{\mathrm{T}}<0$, which implies that the first $\left|\mathcal{S}^{k}\right|-1$ elements
of $A^{k}\left(V'+U^{k}\right)-\bar{V}^{k}$ are zero. Further, considering
that $\mathbf{1}^{\mathrm{T}}A^{k}\left(V'+U^{k}\right)=\mathbf{1}^{\mathrm{T}}\bar{V}^{k}=0$,
we can conclude that $A^{k}\left(V'+U^{k}\right)=\bar{V}^{k}$ for
each $k$. Therefore probability distribution $\pi'=\left[{\pi'}_{i}\right]$
with ${\pi'}_{i}\propto\exp\left(-{V'}_{i}\right)$ satisfies \eqref{eq:Qh-smallerthan-Qhb},
which contradicts Assumption \ref{ass:pi-property}. Thus $\bar{V}$
is the unique solution of $\arg\max_{\mathbf{1}^{\mathrm{T}}V=0}Q_{q}(V;\{\bar{V}^{k}\})$.

\paragraph*{Part (2)}

In this part we will prove that $\mathrm{null}\left(\sum_{k=1}^{K}A^{k\mathrm{T}}\Xi^{k}\left(\hat{w}^{k}\hat{\underline{X}}^{k},\check{\rho}^{k}\right)A^{k}\right)=\mathrm{span}\left(\mathbf{1}\right)$
for any $\hat{w}^{k},\hat{X}^{k},\check{\rho}^{k}$. It is clear that
$\mathrm{span}\left(\mathbf{1}\right)\subseteq\mathrm{null}\left(\sum_{k=1}^{K}A^{k\mathrm{T}}\Xi^{k}\left(\hat{w}^{k}\hat{\underline{X}}^{k},\check{\rho}^{k}\right)A^{k}\right)$
since $A^{k}\mathbf{1}=\mathbf{0}$. Suppose that $v\notin\mathrm{span}\left(\mathbf{1}\right)$
is another vector which belongs to $\mathbf{1}^{\mathrm{T}}v=0$,
then $\bar{V}+v\neq\bar{V}$ is also a maximum point of $\arg\max_{\mathbf{1}^{\mathrm{T}}V=0}Q_{q}\left(V;\left\{ \bar{V}^{k}\right\} \right)$.
This is a contradiction to the result of Part (1). Therefore $\mathrm{null}\left(\sum_{k=1}^{K}A^{k\mathrm{T}}\Xi^{k}\left(\hat{w}^{k}\hat{\underline{X}}^{k},\check{\rho}^{k}\right)A^{k}\right)=\mathrm{span}\left(\mathbf{1}\right)$.

\paragraph*{Part (3)}

According to the result of Part (2) and Theorem 5.2 in \cite{stewart1969continuity},
we have
\begin{equation}
\left(\sum_{k=1}^{K}A^{k\mathrm{T}}\Xi^{k}\left(\hat{w}^{k}\hat{\underline{X}}^{k},\check{\rho}^{k}\right)A^{k}\right)^{+}\stackrel{p}{\to}\left(\sum_{k=1}^{K}A^{k\mathrm{T}}\Xi^{k}\left(\bar{w}^{k}\bar{X}^{k},\bar{\rho}^{k}\right)A^{k}\right)^{+}\label{eq:pinv-convergence}
\end{equation}
Substituting \eqref{eq:Vcheck-new-expression} and \eqref{eq:pinv-convergence}
into $\check{V}-\bar{V}$ and $\sqrt{M}\left(\check{V}-\bar{V}\right)$
leads to the conclusions of the theorem.

\section{Proof of Theorem \ref{thm:SigmaX-kappa}\label{sec:Proof-of-Theorem-SigmaX-kappa}}

Let $G\left(\cdot\right)$ be a function of $y_{t}^{k}$ with
\begin{equation}
G\left(y_{t}^{k}\right)=\left[G_{ij}\left(y_{t}^{k}\right)\right]=\left[1_{I_{\mathcal{S}^{k}}\left(f^{k}\left(y_{t}^{k}\right)\right)=i}\Pr\left(I_{\mathcal{S}^{k}}\left(x_{t+1}^{k}\right)=j|y_{t}^{k}\right)\right]
\end{equation}
and $\kappa_{G}^{k}\left(h\right)=\cov\left(\mathcal{V}\left(G\left(y_{t}^{k}\right)\right),\mathcal{V}\left(G\left(y_{t+h}^{k}\right)\right)\right)$
be the $h$ lag autocovariance of $\left\{ \mathcal{V}\left(G\left(y_{t}^{k}\right)\right)\right\} $.
It is easy to verify that $\kappa_{G}^{k}\left(h\right)$ and $\kappa^{k}\left(h+1\right)$
composed of the same elements but in different arrangements for $h\ge0$,
and $\eta^{k}\left(l\right)=\mathrm{tr}\left(\kappa_{G}^{k}\left(2l\right)+\kappa_{G}^{k}\left(2l+1\right)\right)$.

From the above results, we can conclude that $\eta^{k}\left(l\right)=\mathrm{tr}\left(\kappa_{G}^{k}\left(2l\right)+\kappa_{G}^{k}\left(2l+1\right)\right)$
is a non-negative and decreasing function of $l$ for $l\ge0$, and
$\kappa_{G}^{k}\left(2l\right)+\kappa_{G}^{k}\left(2l+1\right)\ge0$
by using Theorem 3.1 in \cite{geyer1992practical}. Therefore the
second conclusion of the theorem can be proved since $\mathrm{tr}\left(\kappa_{G}^{k}\left(2l\right)+\kappa_{G}^{k}\left(2l+1\right)\right)\ge\left\Vert \kappa_{G}^{k}\left(2l\right)+\kappa_{G}^{k}\left(2l+1\right)\right\Vert _{\max}$.

We now show the first conclusion. According to Theorem 2.1 in \cite{geyer1992practical}
and considering that $\sum_{h=0}^{\infty}\kappa_{G}^{k}\left(h\right)$
is convergent, we have
\begin{equation}
\lim_{N\to\infty}N\var\left(\frac{1}{N}\sum_{t=0}^{N-1}\mathcal{V}\left(G\left(y_{t}^{k}\right)\right)\right)=\kappa_{G}^{k}\left(0\right)+2\sum_{h=1}^{\infty}\kappa_{G}^{k}\left(h\right)
\end{equation}
On the other hand,
\begin{equation}
N\var\left(\frac{1}{N}\sum_{t=0}^{N-1}\mathcal{V}\left(G\left(y_{t}^{k}\right)\right)\right)=\kappa_{G}^{k}\left(0\right)+2\sum_{h=1}^{N-1}\frac{N-h}{N}\kappa_{G}^{k}\left(h\right)
\end{equation}
Combining the above two equations leads to
\begin{equation}
\lim_{N\to\infty}\sum_{h=1}^{N-1}\frac{h}{N}\kappa_{G}^{k}\left(h\right)=\lim_{N\to\infty}\sum_{h=1}^{N-1}\frac{h-1}{N}\kappa^{k}\left(h\right)=\mathbf{0}\mathbf{0}^{\mathrm{T}}
\end{equation}
Therefore

\begin{eqnarray}
N\var\left(\frac{1}{N}\sum_{t=1}^{N}\mathcal{V}\left(\Delta C_{t}^{k}\right)\right) & = & \kappa^{k}\left(0\right)+\sum_{h=1}^{N-1}\frac{N-h}{N}\left(\kappa^{k}\left(h\right)+\kappa^{k}\left(h\right)^{\mathrm{T}}\right)\nonumber \\
 & \to & \kappa^{k}\left(0\right)+\sum_{h=1}^{\infty}\left(\kappa^{k}\left(h\right)+\kappa^{k}\left(h\right)^{\mathrm{T}}\right)
\end{eqnarray}
as $N\to\infty$, and $\Sigma_{X}^{k}=\kappa^{k}\left(0\right)+\sum_{l=0}^{\infty}\left(\Gamma^{k}\left(l\right)+\Gamma^{k}\left(l\right)^{\mathrm{T}}\right)$.

\section{Metropolis simulation model \label{sec:Metropolis-sampling-model}}

The Metropolis simulation model generates the $x_{0:M_{k}}^{k}$ by
the following steps if the reference $V$ and biasing potential $U^{k}$
are given:
\begin{description}
\item [{Step 1.}] Let $t:=0$ and sample $x_{0}^{k}\sim\Pr\left(x_{0}^{k}=s_{i}\right)\propto\exp\left(-U_{i}^{k}\right)$.
\item [{Step 2.}] Sample $x'=s_{j}$ from the proposal distribution $\Pr\left(x'=s_{j}|x_{t}^{k}=s_{i}\right)=q_{ij}\propto1_{\left|i-j\right|\le2}$
and calculate the acceptance ratio $a=\min\left\{ \frac{\exp\left(-V_{j}-U_{j}^{k}\right)q_{ij}}{\exp\left(-V_{i}-U_{i}^{k}\right)q_{ji}},1\right\} $
with $x_{t}^{k}=s_{i}$.
\item [{Step 3.}] Set $x_{t+1}^{k}$ to be $x'$ with probability $a$
and $x_{t}^{k}$ with probability $1-a$.
\item [{Step 4.}] Terminate the simulation if $t=M_{k}$. Otherwise, let
$t:=t+1$ and go to Step 2.
\end{description}
\bibliographystyle{siam}
\bibliography{BiasedMLEstimation}

\end{document}